\documentclass{elsarticle}

\usepackage[utf8]{inputenc}
\usepackage[T1]{fontenc}

\usepackage{amsthm,amssymb,cmll,bussproofs,stmaryrd,mathtools,tikz}
\usepackage{tensor}
\usepackage{hyperref}
\usepackage{cleveref}
\usepackage{fullmacros-thomas}
\usepackage{environ}
\usepackage{tikz}
\usetikzlibrary{matrix,arrows,calc,shapes}
\usepackage{microtype}
\usepackage{refcount}
\usepackage{color}

\newtheorem{theorem}{Theorem}
\newtheorem{proposition}[theorem]{Proposition}
\newtheorem{corollary}[theorem]{Corollary}
\newtheorem{lemma}[theorem]{Lemma}

\theoremstyle{definition}
\newtheorem{definition}[theorem]{Definition}

\theoremstyle{remark}

\newtheorem*{remark}{Remark}
\newtheorem{example}[theorem]{Example}

\newcounter{temp}

\newcommand{\pointmedian}{\fontfamily{cmr}\selectfont\textperiodcentered}

\newcommand{\gct}{{\sc gct}\xspace}

\newcommand{\seq}[1]{\mathbf{#1}}
\newcommand{\len}{\mathrm{len}}

\newcommand{\inputs}[1]{\mathrm{In}(#1)}
\newcommand{\outputs}[1]{\mathrm{Out}(#1)}
\newcommand{\inputspaceE}[1]{\Space{In}^{\mathrm{E}}(#1)}
\newcommand{\inputspaceI}[1]{\Space{In}^{\mathrm{I}}(#1)}

\newcommand{\history}[1]{\mathfrak{h}(#1)}
\newcommand{\benorvar}[2][\seq{e'}]{x_{#2,\history{#1}(#2)}}

\newcommand{\cotree}[2][T]{\textsc{coT}_{#2}(#1)}

\newcommand{\maxflow}{\texttt{maxflow}\xspace}

\newcommand{\prams}{{\sc pram}{\rm s}\xspace}
\newcommand{\pram}{{\sc pram}\xspace}
\newcommand{\acts}{{\sc act}{\rm s}\xspace}

\newcommand{\algcirc}{{\sc algcirc}\xspace}
\newcommand{\srams}{{\sc ram}{\rm s}\xspace}
\newcommand{\sram}{{\sc ram}\xspace}

\newcommand{\cuckersproblem}{\mathfrak{F}\mathrm{er}}

\newcommand{\crew}{\textsc{crew}\xspace}

\newcommand{\amcact}{\alpha_{\mathrm{act}}}
\newcommand{\amcram}{\alpha_{\mathrm{ram}}}
\newcommand{\amcprams}{\alpha_{\mathrm{pram}}}
\newcommand{\amcrealram}{\alpha_{\realN\mathrm{ram}}}
\newcommand{\amcrealpram}{\alpha_{\realN\mathrm{pram}}}
\newcommand{\amcfull}{\alpha_{\mathrm{full}}}
\newcommand{\amcrealfull}{\alpha_{\realN\mathrm{full}}}

\newcommand{\instruction}[1]{\mathtt{#1}}

\newcommand{\copyy}[2]{\instruction{copy}(#1,#2)}
\newcommand{\copyref}[2]{\instruction{copy}(#1,\sharp{}#2)}
\newcommand{\refcopy}[2]{\instruction{copy}(\sharp{}#1,#2)}

\newcommand{\boadd}[3]{\instruction{+}_{#1}(#2,#3)}
\newcommand{\bosubstract}[3]{\instruction{-}_{#1}(#2,#3)}
\newcommand{\bomultiply}[3]{\instruction{\times}_{#1}(#2,#3)}
\newcommand{\boaddconst}[3]{\instruction{+}^{#3}_{#1}(#2)}
\newcommand{\bosubstractconst}[3]{\instruction{-}^{#3}_{#1}(#2)}
\newcommand{\bomultiplyconst}[3]{\instruction{\times}^{#3}_{#1}(#2)}
\newcommand{\bodivide}[3]{\instruction{/}_{#1}(#2,#3)}
\newcommand{\bodivideconst}[3]{\instruction{/}^{#3}_{#1}(#2)}
\newcommand{\bogen}[3]{\star_{#1}(#2,#3)}
\newcommand{\bogenconst}[3]{\star^{#3}_{#1}(#2)}
\newcommand{\boeuclidivide}[3]{\instruction{//}_{#1}(#2,#3)}
\newcommand{\boconst}[2]{\instruction{const}_{#1}(#2)}

\newcommand{\bosqrtn}[2]{\instruction{\sqrt[n]{#1}}(#2)}
\newcommand{\bosqrt}[3][n]{\instruction{\sqrt[#1]{#2}}(#3)}

\newcommand{\rootdegree}[1]{\sqrt[\partial]{#1}}

\newcommand{\seqedges}[2]{\mathrm{Seq}_{#1}(#2)}

\newcommand{\command}[1]{\mathtt{#1}}
\newcommand{\instr}[2][M]{\mathrm{Inst}_{#1}(#2)}
\newcommand{\ncproduct}[2]{\tensor*[_{#1}]{\ast}{_{#2}}}

\newcommand{\plusplus}{\textrm{\scriptsize{$+\!+$}}}

\newcommand{\admss}[2][G]{\mathrm{Adm}_{#2}(#1)}

\newcommand{\hyperplan}[1]{\mathbb{H}(#1)}
\newcommand{\complementset}[1]{#1^{\mathrm{c}}}

\newcommand{\simref}{\sim_{\mathrm{ref}}}
\newcommand{\Interpret}[1]{\mathopen{|\![} #1 \mathclose{]\!|}}

\newcommand{\degre}{\mathop{\textnormal{deg}}}
\newcommand{\rank}{\mathop{\textnormal{rank}}}

\newcommand{\conditional}[3]{\texttt{if}~ {\command{#1}}
  ~ \texttt{goto} ~ {#2} ~ \texttt{else} ~ {#3}}
\newcommand{\Skip}{\texttt{skip}}
\newcommand{\length}[1]{\left| #1 \right|}

\newcommand{\opencovers}[1]{\mathrm{Cov}(\Space{#1})}
\newcommand{\finiteopencovers}[1]{\mathrm{FCov}(\Space{#1})}

\newcommand{\OptProb}{\mathcal{P}_{\mathrm{opt}}}
\newcommand{\DecProb}{\mathcal{P}_{\mathrm{dec}}}
\newcommand{\MaxOptProb}{\mathrm{Max}\OptProb}
\newcommand{\Parametrization}{\mathsf{Param}}
\newcommand{\AffinePlane}{\mathcal{A}_1}
\newcommand{\Deg}{\mathbf{d}(S)}
\newcommand{\Frontier}{\mathsf{Front}}
\newcommand{\Fan}{\mathsf{Fan}}
\newcommand{\Collins}[1]{\mathrm{Col}_{#1}}

\newcommand{\projectionAz}[1]{\Pi(#1)}

\newcommand{\bigre}{{\sc BiGRE}\xspace}
\newcommand{\lobe}{{\sc LoBE}\xspace}
\newcommand{\cohop}{{\sc CoHOp}\xspace}
\newcommand{\dysco}{{\sc DySCo}\xspace}

\begin{document}

\title{Unifying lower bounds for algebraic machines, semantically}

\author[1]{Thomas Seiller\corref{cor1}\fnref{fn1}}
\ead{thomas.seiller@cnrs.fr}
\author[2]{Luc Pellissier\fnref{fn2}} \ead{luc.pellissier@lacl.fr}
\author[3]{Ulysse Léchine} \ead{lechine@lipn.fr}
\cortext[cor1]{Corresponding author}
\fntext[fn1]{T. Seiller was partially supported by the European Commission Horizon 2020 programme Marie Sk\l{}odowska-Curie Individual Fellowship
(H2020-MSCA-IF-2014) project 659920 - ReACT, the INS2I grants \bigre and \lobe, the Ile-de-France DIM RFSI Exploratory project Exploratory project \cohop, and the ANR-22-CE48-0003-01 project \dysco.} \fntext[fn2]{L. Pellissier was partially supported by
ANR-14-CE25-0005 project
\textsc{Elica} and
ANII project “Realizabilidad, forcing y
computación cuántica” FCE\_1\_2014\_1\_104800.}
\affiliation[1]{organization={CNRS}, addressline={LIPN -- UMR 7030 CNRS \& University of Paris 13, 99 avenue Jean-Baptiste Clément},
postcode={93430}, city={Villetaneuse}, country={France}}
\affiliation[2]{organization={University Paris Est Creteil}, addressline={LACL, Faculté des Sciences et Technologie, 61 avenue du Général de Gaulle}, postcode={94010}, city={Creteil}, country={France}}
\affiliation[3]{organization={Université Sorbonne Paris Nord}, addressline={LIPN -- UMR 7030 CNRS \& University of Paris 13, 99 avenue Jean-Baptiste Clément},
postcode={93430}, city={Villetaneuse}, country={France}}

\begin{abstract}
  This paper presents a new abstract method for proving lower bounds in 
  computational complexity. Based on the notion of topological and measurable 
  entropy for dynamical systems, it is shown to generalise three previous lower 
  bounds results from the literature in algebraic complexity. We use it to prove that 
  \maxflow, a \Ptime complete problem, is not computable in polylogarithmic time 
  on parallel random access machines (\prams) working with real numbers. This 
  improves, albeit slightly, on a result of Mulmuley since the class of machines 
  considered extends the class \enquote{\prams without bit operations}, making 
  more precise the relationship between Mulmuley's result and similar lower bounds
  on real \prams. 
  
  More importantly, we show our method  
  captures previous lower bounds results from the literature, thus providing a unifying
  framework for "topological" proofs of lower bounds:
  Steele and Yao's lower bounds for algebraic decision trees \cite{SteeleYao82}, 
  Ben-Or's lower bounds for algebraic computation trees \cite{Ben-Or83}, Cucker's
  proof that \NC is not equal to \Ptime in the real case \cite{Cucker92}, and Mulmuley's
  lower bounds for \enquote{\prams without bit operations} \cite{Mulmuley99}.
\end{abstract}

\maketitle

\newpage
\setcounter{tocdepth}{1}
\tableofcontents

\section{Introduction}
\subsection{Computational Complexity}

The field of computational complexity was initiated soon after the conception of the first computers. 
While theoretical results had already established a definition of the notion of \enquote{computable function} 
on the set of natural numbers, it became quickly apparent that computable did not mean practical, 
as many functions considered computable could not be computed within a reasonable time. 

The first complexity class defined was that of \emph{feasible functions} 
\cite{cobham,edmonds65,rabin67}, which is now known as  $\Ptime$: the set of polynomial 
time computable functions, i.e. functions $f$ for which there exists a polynomial $p$ and a (usually Turing) machine $M$ 
computing $f$ whose running time on an input $n$ is bounded by $p(n)$. This class, apart from being 
the first ever complexity class to appear in the literature, is arguably the most important one in computer 
science. Many fundamental problems concern its relation to other classes, such as knowing whether $\NPtime$, 
the extension of $\Ptime$ if one allows for non-deterministic machines in the above definition, is equal to $\Ptime$. These 
problems, however, are still open. 

Beyond the relationship between \Ptime and other classes, the general question of \emph{classifying} 
the complexity classes became one of the main objectives of the field, and a number of important results 
were obtained within the first years. 

\subsection{Separation, Lower bounds and Barriers}
As part of the classification problem, complexity theory has traditionally been
concerned with proving \emph{separation results}. Among the numerous open
separation problems lies the much advertised \Ptime vs. \NPtime problem of
showing that some problems considered hard to solve but efficient to verify do
not have a polynomial time algorithm solving them.

Proving that two classes $B\subset A$ are not equal can be reduced to 
finding lower bounds for problems in $A$: by proving that certain problems cannot be solved with less than certain
resources on a specific model of computation, one can show that two classes are not equal.
Conversely, proving a separation result $B\subsetneq A$ provides a lower bound for the 
problems that are $A$\emph{-complete} \cite{cooknpcomplete} -- i.e. problems that are in some way \emph{universal} for 
the class $A$.


The proven lower bound results are however very few, and most separation problems
remain as generally accepted conjectures. For instance, a proof that the class
of non-deterministic exponential problems is not included in what is thought of
as a very small class of circuits was not achieved until very recently
\cite{Williams}.
%

The failure of most techniques of proof has been studied in itself, which lead
to the proof of the existence of negative results that are commonly called
\emph{barriers}. Altogether, these results show that all proof methods we know
are ineffective with respect to proving interesting lower bounds. Indeed, there
are three barriers: relativisation \cite{relativization}, natural proofs
\cite{naturality} and algebrization \cite{algebraization,AffineRelativisation}, and every known proof
method hits at least one of them, which shows the need for new
methods\footnote{In the words of S. Aaronson and A. Wigderson
  \cite{algebraization}, \enquote{We speculate that going beyond this limit
    [algebrization] will require fundamentally new methods.}}. However, to this
day, only one research program aimed at proving new separation results is
commonly believed to have the ability to bypass all barriers: Mulmuley and
Sohoni's Geometric Complexity Theory (\gct) program \cite{GCTsurvey2}.

\subsection{Algebraic models and Geometric Complexity Theory}

Geometric Complexity Theory (\gct) is widely considered to be a promising
research program that might lead to interesting results. It is also widely
believed to necessitate new and extremely sophisticated pieces of mathematics in
order to achieve its goal. The research program aims in the long run to provide new techniques 
for answering the \Ptime versus \NPtime problem, focussing first on solving the
\VP versus \VNP problem by showing that certain algebraic 
surfaces (representing the permanent and the determinant) cannot be embedded one 
into the other. Although this program has lead to interesting developments in pure
mathematics, it has not enhanced our understanding of complexity lower bounds
for the time being (actually, even for Mulmuley himself, such understanding will
not be achieved in our lifetimes \cite{100years}).

Intuitively, this program develops a proof method for proving 
lower bounds in \emph{algebraic complexity} based on algebraic geometry \cite{GCTsurvey2}: 
separation of the Valiant complexity classes $\VP$ and $\VNP$ could be deduced from the 
impossibility of embedding an algebraic variety $\mathcal{P}$ defined from the permanent 
into an algebraic variety $\mathcal{D}$ defined from the determinant (with constraints on the 
dimensions). Two main approaches were proposed using representation theory and exploiting 
decomposition of representations into irreducible representations:
\begin{itemize}
\item \emph{occurrence obstructions} aims to exhibit an irreducible representation of 
$G=GL_k\complexN$ which occurs as a $G$-subrepresentation in the coordinate ring of 
$\mathcal{P}$ but not occurring as a $G$-subrepresentation in the coordinate ring of 
$\mathcal{D}$, while 
\item \emph{multiplicity obstructions} an irreducible representation of 
$G=GL_k\complexN$ which occurs as a $G$-subrepresentation in both the coordinate ring of 
$\mathcal{P}$ and in the coordinate ring of $\mathcal{D}$, but whose multiplicity of 
occurrence in the coordinate ring of $\mathcal{P}$ is strictly greater than its multiplicity of 
occurrence in the coordinate ring of $\mathcal{D}$. 
\end{itemize}
Obviously, the first approach is easier, as a particular case of the second. 
Recently, some negative results \cite{BIP16} have shown this easiest path 
proposed by \gct fails. Some positive results on a toy model were however obtained 
regarding multiplicity obstructions \cite{IkenmeyerMultiplicity}: although the obtained 
results are not new, they use the multiplicity obstruction method and are considered
a proof-of-concept of the approach.

The \gct program was inspired, according to its creators, by a lower bound
result obtained by Mulmuley \cite{Mulmuley99} for \enquote{\prams without bit
  operations}, a result we strengthen in the present work.
  
  \subsection{Lower bounds for \prams without bit operations}
  
Mulmuley showed in 1999 \cite{Mulmuley99} that a notion of
machine introduced under the name ``\prams without bit operations'' does not
compute \maxflow in polylogarithmic time. This notion of machine, quite exotic
at first sight, corresponds to an algebraic variant of \prams, where registers
contain integers and individual processors are allowed to perform sums,
subtractions and products of integers. 
It is argued by Mulmuley that this notion of machine provides an expressive model 
of computation, able to compute some non trivial problems in \NC such as Neff's 
algorithm for computing approximate roots of polynomials \cite{Neffsalgo}. 
Mulmuley's result is understood as a big step forward in the quest for a proof that \Ptime 
and \NC are not equal. However, the result was not strengthened or reused in the last 20 years, 
and remained the strongest known lower bound result in this line of enquiry.


The \maxflow problem is quite interesting as it is known to be in \Ptime (by reduction 
to linear programming, or the Ford-Fulkerson algorithm \cite{FordFulkerson}). In fact,
it is \Ptime complete \cite{maxflowcomplete}, and (obviously) there are currently no 
known efficiently parallel algorithm solving it. This lower bound proof, despite
being the main inspiration of the well-known \gct research program, remains
seldom cited and has not led to variations applied to other problems.

\begin{remark}
In fact, the article by Mulmuley shows lower bounds for several problems, and not only for the \maxflow problem. While this paper focusses on \maxflow, this is a choice motivated by the particular importance of the latter (i.e. its \Ptime-completeness). However, the lower bounds obtained by Mulmuley for other problems can be obtained and strengthened in a similar way using our technique.
\end{remark}

\subsection{Contributions.}

The first contribution of this work is a strengthening of Mulmuley's lower
bounds result for machines working on real numbers. Indeed, while the latter proves 
that \maxflow is not computable in polylogarithmic time in a variant of \emph{arithmetic 
\prams}, i.e. working with integers, the proof uses in an essential way techniques from
real algebraic geometry. We will explain how this result is in fact a consequence of our 
main technical lemma, i.e. follows from lower bounds for \emph{algebraic \prams}, that 
is machines working on the reals. Indeed, we show that division-free polylogarithmic 
algebraic \prams compute the same sets of integers as division-free polylogarithmic arithmetic 
\prams (\Cref{prop:NCminusZequivR}). 

We then show that \maxflow is in fact not computable in polylogarithmic 
time in a more expressive model of algebraic \prams, in which processors are allowed to 
perform arbitrary divisions and arbitrary roots in addition to the basic operations 
allowed in Mulmuley's case (addition, subtraction, multiplication). We then explain
how this more general result fails to lift to the arithmetic case, pinpointing to the
precise reason it does: that euclidean division (in fact division by $2$) is not computable in polylogarithmic time
by algebraic \prams. This leads us to prove that the corresponding class, which
contains the problems computable by Mulmuley's notion of \prams, is in fact
strictly contained in \NC. This leads to the main technical result of our paper:
\begin{theorem}\label{mainthm}
  Let $N$ be a natural number and $M$ be a real-valued \pram with
  at most $2^{O((\log N)^c)}$ processors, where $c$ is
  any positive integer.
  
  Then $M$ does not compute euclidean division by $2$ on inputs of length $N$ in $O((\log N)^c)$ steps.
\end{theorem}

This result improves the result shown by Mulmuley, since the class of machines he 
considers\footnote{As mentioned above, Mulmuley considers integer-valued \prams, but 
this class computes exactly the restrictions of sets decided by real-valued \prams to 
integral points.} is strictly contained in the class of machines considered in the statement. 

The second, and main, contribution of the paper is the proof method itself, which is 
based on \emph{dynamic semantics} for programs by means of \emph{graphings},
a notion introduced in ergodic theory and recently used to define models 
of linear logic by the first author 
\cite{seiller-igf,seiller-igg,seiller-ignda,seiller-ige}. The dual nature of
graphings, both continuous and discrete, is essential in the present work, 
as it enables invariants from continuous mathematics, in particular the 
notion of \emph{topological entropy} for dynamical systems, while the 
finite representability of graphings is used in the key lemma (as the number
of \emph{edges} appears in the upper bounds of \Cref{thm:graphingsBenOrsystems}
and \Cref{thm:graphingsBenOrsystemsDegree}).

In particular, we show how this proof method captures
known lower bounds and separation results in algebraic models of computation,
namely Steele and Yao's lower bounds for algebraic decision trees \cite{SteeleYao82},
Ben-Or's lower bounds on algebraic computation trees \cite{Ben-Or83}, 
Cucker's proof that \complexityclass[\realN]{NC}{}
is not equal to \PtimeReal \cite{Cucker92} (i.e. answering the 
\NC vs \Ptime problem for computation over the real numbers).

\subsection{A more detailed view of the proof method}

One of the key ingredients in the proof is the representation of programs as
graphings, and \emph{quantitative soundness} results. We refer to the next
sections for a formal statement, and we only provide an intuitive explanation
for the moment. Since a program $P$ is represented as a graphing
$\Interpret{P}$, which is in some way a dynamical system, the computation $P(a)$
on a given input $a$ is represented as a sequence of values
$\Interpret{a}, \Interpret{P}(\Interpret{a}),
\Interpret{P}^2(\Interpret{a}),\dots$ -- an orbit for the dynamical system.
Quantitative soundness states that not only $\Interpret{P}$ computes exactly the
same function (where computation is understood as convergence to a stationary
value) as the original program $P$, but it does so with a constant time
overhead, i.e. if $P(a)$ terminates on a value $b$ in time $k$, then
$\Interpret{P}^{Ck}(\Interpret{a})=\Interpret{b}$, where $C$ is a constant fixed
once and for all for the model of computation.

The second ingredient is the dual nature of graphings, both continuous 
and discrete objects. Indeed, a graphing \emph{representative} is a 
graph-like structure whose edges are represented as continuous 
maps, i.e. a finite representation of a (partial) continuous dynamical 
system. Given a graphing, we define its \emph{$k$th  cell decomposition}, 
which separates the configuration space into cells such that 
two inputs in the same cell are indistinguishable in $k$ steps, i.e. the 
graphing's computational traces on both inputs are equal. We can then
use both the finiteness of the graphing representatives and the 
\emph{topological entropy} of the associated dynamical system to 
provide upper bounds on the size of a further refinement of this 
geometric object, namely the \emph{$k$-th entropic co-tree}
of a graphing -- a kind of final approximation of the graphing by a 
computational tree: intuitively, the $k$-th entropic co-tree is a computational
tree that mimicks the behaviour of the graphing for the $k$ final steps of 
computation.

As we deal with algebraic models of computation, this implies a bound 
on the representation of the $k$th cell decomposition as a semi-algebraic 
variety. In other words, the $k$th cell decomposition is defined by polynomial 
in\pointmedian{}equalities and we provide bounds on the number and 
degree of the involved polynomials. The corresponding statement is the
main technical result of this paper 
(\Cref{thm:graphingsBenOrsystemsDegree}).

This lemma can then be used to obtain lower bounds results that we now detail.
Precise definitions of the complexity classes involved can be found in \Cref{sec:machines}.

\paragraph{Computational trees.}
Using the Milnor-Oleĭnik-Petrovskii-Thom theorem \cite{Milnor:1964,OleinikPetrovskii,Thom} to bound the number of connected components 
of the $k$th cell decomposition, we then recover the lower bounds of 
Steele and Yao on algebraic decision trees, and the refined result of 
Ben-Or providing lower bounds for algebraic computation trees. In
fact, we even slightly generalise Ben-Or's result as we obtain lower bounds
for the model extended with arbitrary roots, while Ben-Or's original paper
only considered square roots. 

\paragraph{Cucker's result.} 
A different argument based on invariant polynomials provides a proof 
of Cucker's result that $\complexityclass[\realN]{NC}{}\neq\PtimeReal$ by showing that a given polynomial 
that belongs to $\PtimeReal$ cannot be computed within $\complexityclass[\realN]{NC}{}$. In fact,
our main technical result also shows that euclidean division by $2$ cannot be computed by algebraic circuits. 
We present a direct proof of this result inspired from the general entropic co-trees approach (\Cref{sec:directproof}). 
This result is also a direct corollary of \cref{mainthm}.

\paragraph{Mulmuley's result.}
Lastly, following Mulmuley's geometric representation of the \maxflow problem,
we are able to strengthen his celebrated result to obtain lower bounds on
the size (depth) of a \pram over the reals computing this problem. While this result 
holds for real-valued \prams, we explain how it fails to extend to integer-valued
machines, leading to a further strengthening of the result.

\paragraph{Euclidean division.} 
We then explain how to further strengthen the result by showing that euclidean
division cannot be computed in this algebraic \prams model in polylogarithmic
time. More precisely, we show that euclidean division by $2$ cannot be computed, 
leading to the result. Intuitively, this is due to the exponential number of breakpoints in the
geometric representation of euclidean division. This extends the direct proof that
division by $2$ cannot be computed by algebraic circuits obtained earlier, using 
the entropic co-tree method introduced in this paper. 

\section{Contents of the paper}

\subsection{Computation models as graphings.}

The present work reports on the first investigations into how the interpretation
of programs as graphings (generalised dynamical systems) could shed a new
light on proofs of lower bounds. This interpretation of programs rely on two
ingredients:
\begin{itemize}[noitemsep,nolistsep]
\item the interpretation of models of computation as monoid actions. In our
  setting, we view the computational principles of a computational model as
  elements that act on a configuration space. As these actions can be composed,
  but are not necessarily reversible, it is natural to interpret them as
  a monoid acting on a configuration space. 
\item the realization of programs as graphings. We abstract programs as graphs
  whose vertices are subspaces of the product of the configuration space and the
  control states and edges are labelled by elements of the acting monoid, acting
  on subspaces of vertices.
\end{itemize}

A detailed development of the approach can be found in the first author's habilitation 
thesis \cite{seiller-hdr}. We here only provide the main intuitions and formal definitions 
of the approach.

Let us illustrate how monoid actions and graphings formalise the notions of model of computation and program. We consider Turing machines, and mathematically represent the model as the following monoid action. We consider the space of configurations $X\times S$, where $X = \{\star, 0, 1\}^{\vert \mathbf{Z}\vert}$ of $\mathbf{Z}$-indexed sequences of symbols $\star,0,1$ that are almost always equal to $\star$ and $S$ is a finite set of \emph{control states}. A given point in this configuration space, extended with a chosen \emph{control state}, describes a configuration of a Turing machine. Now, instructions present in the model give rise to maps from $X$ to $X$: for instance moving the working head to the right can be represented as $\mathtt{right} : X\rightarrow X$, $(a_i)_{i\in\mathbf{Z}} \mapsto (a_{i+1})_{i\in\mathbf{Z}}$, that is the usual shift operator. The set of instructions then generates a monoid action $M\acton X$, or equivalently a monoid of endomorphisms of $X$, namely the monoid generated by the maps induced by the instructions. A graphing is then a collection of \emph{edges} consisting of a source (a subspace of $X\times S$) and a realiser (an element of the monoid $M$ and a target state in $S$). The instruction \enquote{if in control state $a$ and the head is reading a $0$ or a $1$, move to the right and move to control state $b$} is then represented as an edge of source the subspace $\{(a_i)_{i\in\mathbf{Z}} \in X \mid a_0\neq \star\}\times\{a\}$ and realised by the map $\mathtt{right}\times (a\mapsto b)$.

The basic intuitions here can be summarised by the following slogan:
\enquote{Computation, as a dynamical process, can be modelled as a dynamical system}. Of
course, the above affirmation cannot be true of all computational processes; for
instance the traditional notion of dynamical system is deterministic. In
practice, one works with a generalisation of dynamical systems named
\emph{graphings}. Introduced in ergodic theory \cite{adams,gaboriaucost,gaboriaul2}, graphings were recently
used in theoretical computer science to define realisability models of linear logic \cite{seiller-igg,seiller-ige,seiller-igf},
in which they are shown to model non-deterministic and probabilistic computation \cite{seiller-ignda,seiller-Markov}.

To do so, we consider that a computation model is given by a set of generators
(representing basic instructions) and its actions on a space
(representing the configuration space). So, in other words, we define a
computation model as an action of a monoid (presented by its generators and
relations) on a space $\alpha : M \acton \Space{X}$. This action can then be
specified to be continuous, mesurable, etc. depending on the properties
we are interested in.

A program in such a model of computation is then viewed as a \emph{graphing}: a graph
whose vertices are subspaces of the configuration space and edges are generators
of the monoid: in this way, the partiality of certain operations and branching
are both allowed. This point of view is very general, as it can allow to study,
as special model of computations, models that can be discrete or continuous,
algebraic, rewriting-based, etc.

A point in the configuration space is sent (potentially non-deterministically)
through the graphing to other points: they constitute the \emph{orbit} of the
point under the graphing. This orbit can be eventually stationnary, meaning that
the computation has reached a result, or have any kind of complex behavior. In
any case, the study of the orbits of a graphing contain a lot of information on
the graphing as will be clear when studying entropy.

\subsection{The general algebraic model}

We are able to introduce \prams acting over integers or real numbers in this setting. 
They can be described as having a finite number of processors, each having access 
to a private memory on top of the shared memory, and able to perform the operations
$+, -, \times$ -- and possibly division and root operations (depending on the models 
considered) --, as well as branching and indirect addressing. Interestingly,
we can represent these machines in the graphings framework in two steps: first,
by defining the (sequential) \sram model, with just one processor; and then by 
performing an algebraic operation corresponding to parallelisation on the 
corresponding monoid action.

The \sram model on integers has, as configuration space 
\(\integerN^{\vert \integerN\vert}\) an infinite array of cells each containing a real 
number, and the operations are the usual ones. The equivalent model on the 
reals has as configuration space \(\realN^{\vert \realN\vert}\); the indexing
by real numbers eases the definition in the presence of indirect addressing, 
even though in practise only a finite number of registers can be 
accessed and/or modified in a finite computation.

Parallel computation is thus modelled \emph{per se}, at the level of the model
of computation. As usual, one is bound to chose a mode of interaction between
the different processes when dealing with shared memory. We will consider here
only the case of \emph{Concurrent Read Exclusive Write} (\crew), i.e. all
processes can read the shared memory concurrently, but if several processes try
to write in the shared memory only the process with the smallest index is
allowed to do so.

The heart of our approach of parallelism is based on commutation. Among all the
instructions, the ones affecting only the private memory of distinct processors
can commute, while it is not the case of two instructions affecting the central
memory. We do so by considering a notion of product for monoids (\cref{def:crew}) 
that generalizes both the direct product and the free product: we specify, through 
a conflict relation, which of the generators can and can not commute, allowing us 
to build a monoid representing the simultaneous action.

To ease the presentation, and since many different algebraic models are considered
in the paper, we introduce a very general abstract model of computation, e.g. the 
monoid action $\amcrealfull$ for machines computing on the reals 
(\Cref{def:fullactionreal}). We then show that all notions of machines considered in 
the present paper can be adequately represented by considering restrictions of 
$\amcrealfull$. The main technical lemma being proved for the most general action 
$\amcrealfull$, it then applies naturally to all those models: algebraic computational 
trees, algebraic circuits, \prams over the reals, etc.

\subsection{Entropy}

We fix an action $\alpha:M\acton \Space{X}$ for the following discussion (it can be 
thought of as the most general action $\amcrealfull$ that captures all algebraic models
of interest in this work).
One important aspect of the representation of abstract programs as graphings is
that restrictions of graphings correspond to known notions from mathematics. In a 
very natural way, a deterministic $\alpha$-graphing defines a partial dynamical 
system. Conversely, a partial dynamical system whose graph is contained in the
\emph{measured preorder} $\{ (x,y)\in\Space{X}^2\mid \exists m\in M, 
\alpha(m)(x)=y \}$ \cite{seiller-towards} can be associated to an $\alpha$-graphing.

The study of deterministic models of computations can thus profit from the methods
of the theory of dynamical systems. In particular, the methods employed in this
paper relate to the classical notion of \emph{topological entropy}. The topological 
entropy of a dynamical system is a value representing the average exponential 
growth rate of the number of orbit segments distinguishable with a finite (but 
arbitrarily fine) precision. The definition is based on the notion of open covers: 
for each finite open cover $\mathcal{C}$, one can compute the entropy of a map 
w.r.t. $\mathcal{C}$, and the entropy of the map is then the supremum of these
values when $\mathcal{C}$ ranges over the set of all finite covers. As we are
considering graphings and those correspond to partial maps, we explain how 
the techniques adapt to this more general setting and define the entropy 
$h(G,\mathcal{C})$ of a graphing $G$ w.r.t. a cover $\mathcal{C}$, as well as
the topological entropy $h(G)$ defined as the supremum of the values
$h(G,\mathcal{C})$ where $\mathcal{C}$ ranges over all finite open covers.

While the precise results described in this paper use the entropy $h_0(G)$ w.r.t. a 
specific cover (similar bounds could be obtained from the topological entropy, but 
would lack precision), the authors believe entropy could play a much more prominent 
role in future proofs of lower bounds. Indeed, while $h_0(G)$ somehow quantifies over 
one aspect of the computation, namely the branchings, the 
topological entropy computed by considering all possible covers provides a much
more precise picture of the dynamics involved. In particular, it provides information
about the computational principles described by the action; this information may lead 
to more precise bounds based on how some principles are much more complex 
than some others, providing some lower bounds on possible simulations of the 
former with the latter.

\subsection{Known lower bounds and entropy}

All the while only the entropy w.r.t. a given cover will be essential in this work,
the overall techniques related to topological entropy provide a much clearer picture 
of the techniques. We first use the following lemma bounding the $k$-cell
decomposition of a given graphing.

\begin{proposition}\label{prop:entropy-k-cell-h0}
Let $G$ be a deterministic graphing. We consider the \emph{state cover entropy} $h_0([G])=\lim_{n\rightarrow\infty}H_{\Space{X}}^{n}([G],\mathcal{S})$ where $\Space{S}$ is the state cover. The cardinality of the $k$-th cell decomposition of $\Space{X}$ w.r.t. $G$, as a function $c(k)$ of $k$, is asymptotically bounded by $g(k)=2^{k.h_0([G])}$, i.e. $c(k)=O(g(k))$.
\end{proposition}

This lemma can be used to derive known lower bounds from the literature, namely
Steele and Yao's result on \emph{algebraic decision trees}. Algebraic decision
trees are finite ternary trees describing a program deciding a subset of
$\realN^n$: each node verifies whether a chosen polynomial
takes a positive, negative, or null value at the point considered. A $d$-th
order algebraic decision tree is an algebraic decision tree in which all
polynomials are of degree bounded by $d$.

In a very natural manner, an algebraic decision tree can be represented as a
$\iota$-graphing, where $\iota$ is the trivial action on the space $\realN^n$.
We use entropy to provide a bound on the number of connected components of
subsets decided by $\iota$-graphings. These bounds are obtained by combining a
bound in terms of entropy and a variant of the Milnor-Oleĭnik-Petrovskii-Thom theorem due to
Ben-Or. The latter, which we recall below, bounds the number of connected
components of a semi-algebraic set in terms of the number of polynomial
inequalities, their maximal degree, and the dimension of the space considered.
These bounds then lead to the lower bounds on algebraic decision trees obtained
by Steele and Yao.

\begin{corollary}[Steele and Yao \cite{SteeleYao82}]\label{thm:SteeleYao}
A $d$-th order algebraic decision tree deciding a subset $W\subseteq \realN^n$ 
with $N$ connected components has height $\Omega(\log N)$.
\end{corollary}

Mulmuley's original result on \prams without bit operations can also be obtained
as a corollary of this bound, associated with a more involved geometric 
argument. We detail and rephrase the latter (in a way that will allow us to reuse 
it in the last section) in \cref{sec:geometry}. Combining the results of this
section and the previous lemma about the $k$-cell decomposition, one obtains
Mulmuley's theorem.

\begin{corollary}[Mulmuley \cite{Mulmuley99}]\label{thm:Mulmuley}
  Let $M$ be a \pram without bit operations, with at most
  $2^{O((\log N)^c)}$ processors, where $N$ is the length of the inputs and $c$
  any positive integer.
  
  Then $M$ does not decide \maxflow in $O((\log N)^c)$ steps.
\end{corollary}

\subsection{Entropic co-tree and the technical lemma}

The above result of Steele and Yao adapts in a straightforward manner
to a notion of algebraic computation trees describing the construction of the 
polynomials to be tested by mean of multiplications and additions of the 
coordinates. As shown above, this result uses techniques quite similar to 
that of Mulmuley's lower bounds for the model of \emph{\prams without bit
operations}. The authors also quickly realised the techniques is similar to that
used by Cucker in proving that \complexityclass[\realN]{NC}{}$\neq$ \PtimeReal \cite{Cucker92}.

It turns out a refinement of Steele and Yao's method was quickly obtained by
Ben-Or \cite{Ben-Or83} so as to obtain a similar result for an extended notion
of algebraic computation trees allowing for computing divisions and taking
square roots. We adapt Ben-Or techniques within the framework of graphings, in
order to apply this refined approach to Mulmuley's framework, leading to a
strengthened lower bounds result.

This technique refines the bounds on the $k$th-cell decomposition explained
above by considering \emph{entropic co-trees}. They are defined in a similar way
as the $k$-th cell decomposition but further track the different instructions
involved. The result is a directed graph in the form of a tree with all edges
pointing toward the root (hence the name of \emph{co-tree}). This tree can be
understood as a final approximation of the graphing as a computational tree: it
is a computational tree whose behaviour mimics that of the graphing in the last
computation steps leading to acceptation or rejection. At each fixed depth, the
set of vertices of the co-tree refine the $k$-th cell decomposition explained
above. The additional information related to the instructions realising the
edges of the co-tree can be used to derive a set of polynomial
in\pointmedian{}equalities whose total degree can be bounded by the depth, the
state-cover entropy $h_0$, the algebraic degree
(\Cref{def:algebraicdegree}) -- the maximal number of instructions used in a
single edge of the graphing --, and the root degree $\rootdegree{G}$ -- the
largest integer $d$ such that the $d$-th root instruction appears in the
graphing.

This leads to the following technical lemma, from which most of the subsequent 
lower bound results will be obtained.

\begin{lemma}\label{thm:graphingsBenOrsystemsDegree}\label{mainlemma}
Let $G$ be a $\crew^{p}(\amcrealfull)$-computational graphing representative, 
$\seqedges{k}{E}$ the set of length $k$ sequences of edges in $G$, and $D$ 
its algebraic degree. 
Suppose $G$ computes the membership 
problem for $W \subseteq \realN^n$ in $k$ steps, i.e. for each element of 
$\realN^n$, $\pi_{\Space{S}}(G^{k}(x))=\top$ if and only if $x\in W$.
 Then $W$ is a semi-algebraic 
set defined by at most $\card{\seqedges{k}{E}}.2^{k.h_0([G])}$ systems of $pkD$ 
equations of degree at most $\max(2,\rootdegree{G})$ and involving at most 
$pD(k+n)$ variables.
\end{lemma}

%
As a first corollary of this theorem, we obtain a generalisation of 
Ben-Or's result (\Cref{thm:graphingsBenOr}). It follows from the above theorem 
and the Milnor-Oleĭnik-Petrovskii-Thom theorem bounding the number of connected components of 
a semi-algebraic set. A corollary of this general result is Ben-Or's original lower
bounds for Algebraic Computational Trees.

\begin{corollary}[{\cite[Theorem 5]{Ben-Or83}}]\label{ben-or}
  Let $W \subseteq \realN^n$ be any set, and let $N$ be the maximum of the
  number of connected components of $W$ and $\realN^n \setminus W$.
  An algebraic computation tree computing the membership problem for $W$ has
  height $\Omega(\log N)$.
\end{corollary}

The above 
lemma is then applied to obtain Cucker's theorem that 
$\complexityclass[\realN]{NC}{} \neq \PtimeReal$. 

\begin{corollary}\label{thm:cucker}
No algebraic circuit of depth $k=\log^i n$ and size $kp$ compute 
$\cuckersproblem$: 
\[\{x\in \realN^\omega \mid \abs{x}=n \Rightarrow x_1^{2^n}+x_2^{2^n}=1 \}.\]
\end{corollary}

Finally, the lemma can be used to obtain a first strengthening of Mulmuley's
lower bound (\Cref{cor:main-pram}). As explained above, Mulmuley's result for \prams
without bit operations (working on integers) corresponds to lower bounds on
division-free algebraic -- i.e. working on the reals -- \prams. By using the
above technical lemma, we are able to extends the lower bounds for computing
\maxflow to algebraic \prams with division and arbitrary roots.

Before stating this result and explaining how we are able to further generalise
it, we sketch Mulmuley's geometric method. While the results here are not new,
the authors' contribution is that of reformulation. In particular, we reorganise
the geometric part of the proof as the combination of two results: a previous
result of Murty and Carstensen showing the existence of an exponential linear
parametrization of \maxflow, and a general geometric statement
(\Cref{thm:mulmuley-geometric}). This reformulation will be used in later
sections to show that a problem easier than \maxflow cannot be computed in the
algebraic \prams model.

\subsection{Mulmuley's geometrization}
Contrarily to Ben-Or's model, the \pram machines do not decide sets of reals but of
integers, making the use of algebraico-geometric results to uncover their
geometry much less obvious. The mechanisms of Mulmuley's proof rely on twin
geometrizations: one of a special optimization problem that can be represented
by a surface in $\realN^3$ (subsections ~\ref{subsec:optprob} and \ref{subsec:param}), 
the other one by building explicitly, given a \pram, a set of algebraic surfaces such that
the points accepted by the machine are exactly the integer points enclosed by
the set of surfaces.

That second part can be abstracted as a relation between the sets decided
by \prams without bit operations and division-free algebraic \prams. In other 
words, the set of integral points accepted by a \pram without bit operations 
coincides with the integral points lying in the set decided by a corresponding 
algebraic \pram. This is stated as (\Cref{prop:NCminusZequivR}).

Finally, the proof is concluded by a purely geometrical theorem
(\Cref{thm:mulmuley-geometric}). We would like to stress here that
this separation in three movement, with a geometrical tour-de-force, is not
explicit in the original article. We nonetheless believe it greatly improves
the exposition (on top of allowing for a stregthening of the results). This geometric
theorem expresses a tension between the two geometrizations. Our
work focuses here only on the construction of a set of algebraic surfaces
representing the computation of a \pram; the remaining part of our proof follows 
Mulmuley's original technique closely.

\paragraph{Building surfaces}

The first step in Mulmuley's proof is to use the parametric complexity results
of \cite{Carstensen:1983} to represent an instance of the decision
problem associated to \maxflow so that it induces naturally a partition of
$\integerN^3$ that can then be represented by a surface. 

The second step is to represent any partition of $\integerN^3$ induced by the
run of a machine by a set of surfaces in $\realN^3$, in order to be able to use
geometric methods.

Let $K$ be a compact of $\realN^3$ and $P=(P_1,\ldots,P_m)$ be a partition of
$\integerN^3 \cap K$. $P$ can be extended to a partition of the whole of $K$ in
a number of ways, as pictured in Fig. \ref{fig:curves-partition}. In particular,
$P$ can always be extended to a partition $P_{\mathrm{alg}}$
(resp. $P_{\mathrm{smooth},}$, $P_{\mathrm{ana}}$) of $K$ such that all the
cells are compact, and the boundaries of the cells are all algebraic
(resp. smooth, analytic) surfaces.

In general, such surfaces have no reason to be easy to compute and the more they
are endowed with structure, the more complicated to compute they are to be.
In the specific case of \prams, the decomposition can naturally be represented
with algebraic surfaces whose degree is bounded. This choice of
representation might not hold for any other model of computation, for which it
might be more interesting to consider surfaces of a different kind.

\begin{figure}
\centering
  \begin{tikzpicture}
    \foreach \i in {0,...,5}
    \foreach \j in {0,...,4}{
      \ifnum \j < 1
      \fill[black] (\i,\j) circle(2pt);
      \else
      \ifnum \j > 3
      \fill[black] (\i,\j) circle(2pt);
      \else
      \ifnum \i < 1
      \fill[black] (\i,\j) circle(2pt);
      \else
      \ifnum \i > 4
      \fill[black] (\i,\j) circle(2pt);
      \else
      \fill[red] (\i,\j) circle(2pt);
      \fi
      \fi
      \fi
      \fi
    };  
    \draw[rounded corners,color=blue] (0.5, 0.5) rectangle (4.5, 3.5) {};
    \draw[purple] (0.75,0.25) cos (1.8,1.3) sin (2.3,0) cos (4.5, 1.3) sin (4.7, 3.1);
    \draw[purple] (0.75,0.25) sin (0.9,3.8) cos (4.7,3.1);
  \end{tikzpicture}
  \caption{Two curves that define the same partition of $\integerN^2$}
  \label{fig:curves-partition}
\end{figure}

This set of algebraic surfaces is here built just as in our description of Ben-Or's 
result using the entropic co-tree: we construct the co-tree approximating the 
computation of a specific \pram and build along the branches of this co-tree a 
system of polynomial equations on a larger space than the space of variables 
actually used by the machine.
This system of integer polynomials of bounded degree then defines surfaces
exactly matching our needs, since the number of varieties and their maximal 
degrees are bounded using the technical lemma described above.

\subsection{A first improvement on Mulmuley's lower bounds}

Interestingly, this allows us to derive from the technical lemma a proof that algebraic \prams,
with division and arbitrary roots instructions on top of the instructions allowed in the
\prams without bit operations model, cannot compute \maxflow in polylogarithmic
time. This is a first (arguably mild) improvement over Mulmuley's proof. 
\begin{theorem}
  \label{cor:main-pram}
  Let $N$ be a natural number and $M$ be a real-valued \pram with
  at most $2^{O((\log N)^c)}$ processors, 
  where $c$ is any positive integer.
  
  Then $M$ does not decide \maxflow on inputs of length $N$ in $O((\log N)^c)$ steps.
\end{theorem}

Following Mulmuley's original method on division-free machines, one would then
like to lift this result to arithmetic \prams, i.e. \prams over the integers
with division (and possibly root operations). It turns out, however, that this
cannot be done, a result not surprising since such machines can be shown to be
equivalent to usual \prams.

More precisely, we can prove that algebraic \prams cannot compute euclidian
division. We provide two proofs of this result. We first show that algebraic
circuits (i.e. with division but no root operations) cannot compute euclidean division by $2$
using a (new) direct proof method bounding the number of continuous
pieces of the piece-wise continuous function computed by the circuit. This direct
proof is a specific but more concrete instance of the general entropic co-tree
approach, and illustrates well how the lower bounds are obtained. We further explain
how to abstract the mathematical problem corresponding to generalising this
bound to root operations, and illustrate the difficulty of doing so (even with
only square roots).

In a second step, we use our technical lemma, together with the geometric 
result of Mulmuley, to provide lower bounds for computing the remainder modulo $2$ 
 in the general algebraic \prams model. This leads to the main technical result
of this paper.

\setcounter{temp}{\value{theorem}}
\setcounterref{theorem}{mainthm}
\addtocounter{theorem}{-1}
\begin{theorem}
  Let $N$ be a natural number and $M$ be a real-valued \pram with
  at most $2^{O((\log N)^c)}$ processors, where $c$ is
  any positive integer.
  
  Then $M$ does not compute the euclidean division modulo $2$ on inputs of length $N$ in $O((\log N)^c)$ steps.
\end{theorem}
\setcounter{theorem}{\value{temp}}

This result shows a fundamental difference in expressive power between models
with real-valued division and models allowed to compute the euclidean division.
Indeed, this shows that a parallel model over the reals, even if allowed division
and arbitrary roots, cannot compute euclidean division in polylogarithmic time.


\subsection{Conclusion}

This work slightly strengthens Mulmuley's lower bounds on "\prams without bit
operations". More importantly, it shows how 
the semantic techniques based on abstract models of computation and graphings
can shed new light on several lower bound techniques. In particular, it establishes 
some relationship between the lower bounds and the notion of entropy which, 
although arguably still superficial in this work, could potentially become deeper 
and provide new insights and finer techniques. 

Showing that the interpretation of programs as graphings can translate, and even
refine, such strong lower bounds results is also important from another perspective. 
Indeed, the techniques of Ben-Or and Mulmuley (as well as other results of e.g.
Cucker \cite{Cucker92}, Yao \cite{YaoBetti}) seem at first sight restricted to algebraic 
models of computation due to their use of the Milnor-Oleĭnik-Petrovskii-Thom theorem (or other geometric
arguments) which holds only for real semi-algebraic sets. However, the first author's 
characterisations of Boolean complexity classes in terms of graphings acting on algebraic 
spaces \cite{seiller-ignda,seiller-hdr} opens the possibility of using such algebraic methods to provide
lower bounds for boolean models of computation.

\setcounter{theorem}{0}

\section{Programs as Dynamical systems}

\subsection{Abstract models of computation and graphings}
\label{sec:amc-graphings}

We consider computations as a dynamical process, hence model them as a dynamical
systems with two main components: a space \(\Space{X}\) that abstracts the notion
of configuration space (excluding control states) and a monoid acting on this space that 
represents the different operations allowed in the model of computation. Although the 
notion of \emph{space} considered can vary (one could consider e.g. topological spaces,
measure spaces, topological vector spaces), we restrict ourselves to topological 
spaces in this work.

This is part of an approach introduced by the first author under the name of \enquote{mathematical
informatics} \cite{seiller-hdr}. We here only
provide the basic definitions needed to establish the results in this paper. We note however
that the techniques do not apply solely to algebraic models of computation but 

\begin{definition}
  An \emph{abstract model of computation} (\amc) is a monoid action 
  $\alpha: M\acton\Space{X}$, i.e. a monoid morphism from $M$ to the group of
  endomorphisms of $\Space{X}$.
  The monoid $M$ is often given by a set \(G\) of generators and a set of
  relations \(\relations{R}\). 
 We denote such an \amc as $\AMC[X]{G}{R}{\alpha}$.
\end{definition}

Programs in an \amc $\AMC[X]{G}{R}{\alpha}$ is then defined as \emph{graphings},
i.e. graphs whose vertices are subspaces of the space \(\Space{X}\)
(representing sets of configurations on which the program act in the same way)
and edges are labelled by elements of \(\MonGaR{G}{R}\), together with a global
control state. More precisely, we use here the notion of \emph{topological
  graphings}\footnote{While \enquote{measured} graphings were already considered
  \cite{seiller-igg}, the definition adapts in a straightforward manner to allow for
  other notions such as graphings over topological vector spaces -- which would
  be objects akin to the notion of quiver used in representation theory.} 
  \cite{seiller-igg}.
  
\begin{definition}
  An $\alpha$-graphing representative $G$ w.r.t. a monoid action $\alpha: M\acton\Space{X}$ 
  is defined as a set of \emph{edges} $E^{G}$ together with a map that assigns to
  each element $e\in E^{G}$ a pair $(S^{G}_{e},m^{G}_{e})$ of a subspace 
  $S^{G}_{e}$ of $\Space{X}$ -- the \emph{source} of $e$ -- and an element 
  $m^{G}_{e}\in M$ -- the \emph{realiser} of $e$.
\end{definition}

While graphing representatives are convenient to manipulate, they do provide too
much information about the programs. Indeed, if one is to study programs as
dynamical systems, the focus should be on the \emph{dynamics}, i.e. on how the
object acts on the underlying space. The following notion of \emph{refinement} 
captures this idea that the same dynamics may have different graph-like 
representations.

\begin{definition}[Refinement]
  An $\alpha$-graphing representative $F$ is a refinement of an $\alpha$-graphing 
  representative $G$, noted $F\leqslant G$, if there exists a partition 
  $(E^{F}_{e})_{e\in E^{G}}$ of $E^{F}$ such that $\forall e\in E^{G}$:
  $$\begin{array}{c}
    \forall f\neq f'\in E^{F}_{e},~ S^{F}_{f} \mathbin{\triangle} S^{F}_{f'}
    = \emptyset;\\
    \hspace{0.5cm}
    \forall f\in E^{F}_{e},~ m^{F}_{f}=m_{e}^{G}.
    \end{array}
  $$
  This induces an equivalence relation defined as 
  \[ F\simref G \Leftrightarrow 
  \exists H, ~ H\leqslant F \wedge H\leqslant G.\]
\end{definition}

The notion of \emph{graphing} is therefore obtained by considering the quotient 
of  the set of graphing representatives w.r.t. $\simref$. Intuitively, this corresponds
to identifying graphings whose \emph{actions on the underlying space are equal}.

\begin{definition}
  An $\alpha$-\emph{graphing} is an equivalence class of $\alpha$-graphing 
  representatives w.r.t. the equivalence relation $\simref$.
\end{definition}

We can now define the notion of abstract program. These are defined as
graphings 

\begin{definition}
  Given an \amc $\alpha:M\acton\Space{X}$, an \emph{$\alpha$-program} 
  $A$ is a $\bar{\alpha}$-graphing $G^{A}$ w.r.t. the monoid action 
  $\bar{\alpha}=\alpha\times\mathfrak{S}_{k}\acton \Space{X}\times\Space{S^{A}}$,
  where $\Space{S^{A}}$ is a finite set of \emph{control states} of cardinality
  $k$ and $\mathfrak{S}_{k}$ is the group of permutations of $k$ elements.
\end{definition}

Now, as a sanity check, we will show how the notion of graphing do capture 
the dynamics as expected. For this, we restrict to \emph{deterministic
graphings}, and show the notion relates to the usual notion of dynamical
system.

\begin{definition}
  An $\alpha$-graphing representative $G$ is deterministic if for all $x\in\Space{X}$ 
  there is at most one $e\in\ E^{G}$ such that $x\in S^{G}_{e}$.
  An $\alpha$-graphing is \emph{deterministic} if its representatives are deterministic.
   An abstract program is \emph{deterministic} if its underlying graphing is
  deterministic.
 \end{definition}
 
 \begin{lemma}
 There is a one-to-one correspondence between the set of deterministic graphings
 w.r.t. the action $M\acton\Space{X}$ and the set of partial 
 dynamical systems $f:\Space{X}\hookrightarrow\Space{X}$ whose graph
 is contained in the preorder\footnote{When $\alpha$ is a group action
 acting by measure-preserving transformations, this is a \emph{Borel 
 equivalence relation} $\mathcal{R}$, and the condition stated here boils
 down to requiring that $f$ belongs to the \emph{full group} of $\alpha$.}
 $\{(x,y)\mid \exists m\in M, \alpha(m)(x)=y\}$.
 \end{lemma}
 
 Lastly, we define some restrictions of $\alpha$-programs that will be important
 later. First, we will restrict the possible subspaces considered as sources of the
 edges, as unrestricted $\alpha$-programs could compute even undecidable 
 problems by, e.g. encoding it into a subspace used as the source of an edge.
 Given an integer $k\in\omega$, we define the following subspaces of 
$\realN^{\omega}$, for $\star\in\{>,\geqslant,=,\neq,\leqslant,<\}$:
\[ \realN^{\omega}_{k\star 0}=\{(x_{1},\dots,x_{k},\dots)\in
	\realN^{\omega}\mid x_{k}\star 0\}.\]

\begin{definition}[Computational graphings] 
Let $\AMC{G}{R}{\alpha}$ be an \amc. 
A \emph{computational $\alpha$-graphing} is an $\alpha$-graphing $T$  with 
distinguished states $\top$, $\bot$ which admits a finite representative such that 
each edge $e$ has its source equal to one among 
$\realN^{\omega}$, 
$\realN^{\omega}_{k\geqslant 0}$, 
$\realN^{\omega}_{k\leqslant 0}$, 
$\realN^{\omega}_{k> 0}$, 
$\realN^{\omega}_{k< 0}$, 
$\realN^{\omega}_{k=0}$, and 
$\realN^{\omega}_{k\neq 0}$.
\end{definition}

\begin{definition}[treeings] 
Let $\AMC{G}{R}{\alpha}$ be an \amc. 
An \emph{$\alpha$-treeing} is an acyclic 
and finite $\alpha$-graphing, i.e. an $\alpha$-graphing $F$ for which 
there exists a finite $\alpha$-graphing representative $T$ whose set of control states 
$\Space{S^{T}}=\{0,\dots,s\}$ can be endowed with an order $<$ such that every edge of 
$T$ is state-increasing, i.e. for each edge $e$ of source $S_{e}$, for all $x\in S_{e}$, 
\[ \pi_{\Space{S^{T}}}(\alpha(m_{e})(x)>\pi_{\Space{S^{T}}}(x),\] where 
$\pi_{\Space{S^{T}}}$ denotes the projection onto the control states space.

A \emph{computational $\alpha$-treeing} is an $\alpha$-treeing $T$ which is a 
computational $\alpha$-graphing with the distinguished states $\top$, $\bot$ being 
incomparable maximal elements of the state space.
\end{definition}

\subsection{Quantitative Soundness}

As mentioned in the introduction, we will use in this paper the property of 
\emph{quantitative soundness} of the dynamic semantics just introduced.
This result is essential, as it connects the time complexity of programs in
the model considered (e.g. \prams, algebraic computation trees) with the
length of the orbits of the considered dynamical system. We here only state 
quantitative soundness for \emph{computational graphings}, i.e. graphings
that have distinguished states $\top$ and $\bot$ representing acceptance
and rejection respectively. In other words, we consider graphings which 
compute \emph{decision problems}.

Quantitative soundness is expressed with respect to a translation of machines
as graphings, together with a translation of inputs as points of the configuration 
space. In the following section, these operations are defined for each model
of computation considered in this paper. In all these cases, the representation 
of inputs is straightforward.

\begin{definition}
Let \amc $\alpha$ be an abstract model of computation, and $\mathbb{M}$ a
model of computation. A \emph{translation} of $\mathbb{M}$ w.r.t. $\alpha$ is
a pair of maps $\Interpret{\cdot}$ which associate to each machine $M$ in 
$\mathbb{M}$ computing a decision problem a computational $\alpha$-graphing 
$\Interpret{M}$ and to each input $\iota$ a point $\Interpret{\iota}$ in 
$\Space{X}\times\Space{S}$.
\end{definition}

\begin{definition}
Let \amc $\alpha$ be an abstract model of computation, $\mathbb{M}$ a model 
of computation. The \amc $\alpha$ is \emph{quantitatively sound} for $\mathbb{M}$ 
w.r.t. a translation $\Interpret{\cdot}$ if for all 
machine $M$ computing a decision problem and input $\iota$, $M$ accepts $\iota$
(resp. rejects $\iota$) in $k$ steps if and only if $\Interpret{M}^{k}(\Interpret{\iota})=\top$ 
(resp. $\Interpret{M}^{k}(\Interpret{\iota})=\bot$).
\end{definition}

\subsection{The algebraic \amcs}

We now define the actions $\amcfull$ and $\amcrealfull$. Those will capture all algebraic
models of computation considered in this paper, and the main lemma (\cref{mainlemma})
will be stated for this monoid action. All lower bounds results recovered from the literature, as 
well as the new lower bounds obtained in this work will be obtained as corollaries of this
technical lemma.

As we intend to consider \prams at some point, we consider from the beginning 
the memory of our machines to be separated in two infinite blocks
$\integerN^{\omega}$, intended to represent sets of both \emph{shared} and
\emph{private} memory cells\footnote{Obviously, this could be done without any
explicit separation of the underlying space, but this will ease the constructions 
of the next section.}. 

\begin{definition}
The underlying space of $\amcfull$ is
$\Space{X}= \integerN^{\integerN}\cong\integerN^{\omega}\times
\integerN^{\omega}$. The set of generators
is defined by their action on the underlying space, writing 
$k//n$ the floor $\floor{k/n}$ of $k/n$ with the conventions that $k//n=0$ 
when $n=0$ and $\sqrt[n]{k}=0$ when $k\leqslant 0$:
\begin{itemize}
\item $\boconst{i}{c}$ initialises the register $i$ with the constant $c\in\integerN$: $\amcfull(\boconst{i}{c})(\vec{x}) =
  (\vec{x}\{x_i:= c\})$;
\item $\bogen{i}{j}{k}$ ($\star\in\{+,-,\times,//\}$) performs the algebraic operation $\star$ on the values in registers $j$ and $k$ and store the result in register $i$:  $\amcfull(\bogen{i}{j}{k})(\vec{x}) =
  (\vec{x}\{x_i:= x_{j}\star x_k\})$;
\item $\bogenconst{i}{j}{c}$ ($\star\in\{+,-,\times,//\}$) performs the algebraic operation $\star$ on the value in register $j$ and the constant $c\in\integerN$ and store the result in register $i$:  $\amcfull(\bogenconst{i}{j}{c})(\vec{x}) =
  (\vec{x}\{x_i:= c\star x_{j}\})$;
\item $\copyy{i}{j}$ copies the value stored in register $j$ in register $i$: $\amcfull(\copyy{i}{j})(\vec{x}) =
  (\vec{x}\{x_i:= x_j\})$;
\item $\refcopy{i}{j}$ copies the value stored in register $j$ in the register whose index is the value stored in register $i$: $\amcfull(\refcopy{i}{j})(\vec{x}) =
  (\vec{x}\{x_{x_i}:= x_j\})$;
\item $\copyref{i}{j}$ copies the value stored in the register whose index is the value stored in register $j$ in register $i$: $\amcfull(\copyref{i}{j})(\vec{x}) =
  (\vec{x}\{x_i:= x_{x_j}\})$;
\item $\bosqrtn{i}{j}$ computes the floor of the $n$-th root of the value stored in register $j$ and store the result in register $i$: $\amcfull(\bosqrtn{i}{j})(\vec{x}) =
  (\vec{x}\{x_i:= \sqrt[n]{x_j}\})$.
\end{itemize}
\end{definition}

We also define the real-valued equivalent, which will be essential for the
proof of lower bounds. The corresponding \amc $\amcrealram$ is defined in 
the same way than the integer-valued one, but with underlying space
\(\Space{X}= \realN^{\integerN}\) and with instructions adapted accordingly:
\begin{itemize}[noitemsep,nolistsep]
\item the division and $n$-th root operations are the usual operations on the reals;
\item the three copy operators are only effective on integers.
\end{itemize}
Note that we consider the space $\realN^{\realN}$, i.e. an uncountable number of 
potential registers. This appears to us as the simplest way to represent the model
of real-valued \prams which includes indirect addressing. In practise, only a finite
number of registers can be accessed during a finite execution (since indexes need
to be computed), and therefore a countable number of potential registers would be
enough. However this poses the issue of defining the semantics properly: using 
maps from $\realN$ to $\naturalN$ do not work because this creates side-effects
giving more expressive power to the machines (e.g. considering that indirect
addressing $\refcopy{i}{j}$ modifies the register of index $\floor{i}$ -- where 
$\floor{\cdot}$ is the floor function -- turns out to 
provide a way to define euclidean division!). On the other hand, defining a dynamic
allocation of register should be possible but would complicate the definitions.

\begin{definition}\label{def:fullactionreal}
The underlying space of $\amcrealfull$ is
$\Space{X}= \realN^{\realN}\cong\realN^{\realN}\times
\realN^{\realN}$. The set of generators
is defined by their action on the underlying space, with the conventions that $k/n=0$ 
when $n=0$ and $\sqrt[n]{k}=0$ when $k\leqslant 0$:
\begin{itemize}
\item $\boconst{i}{c}$ initialises the register $i$ with the constant $c\in\realN$: $\amcrealfull(\boconst{i}{c})(\vec{x}) =
  (\vec{x}\{x_i:= c\})$;
\item $\bogen{i}{j}{k}$ ($\star\in\{+,-,\times,/\}$) performs the algebraic operation $\star$ on the values in registers $j$ and $k$ and store the result in register $i$:  $\amcrealfull(\bogen{i}{j}{k})(\vec{x}) =
  (\vec{x}\{x_i:= x_{j}\star x_k\})$;
\item $\bogenconst{i}{j}{c}$ ($\star\in\{+,-,\times,/\}$) performs the algebraic operation $\star$ on the value in register $j$ and the constant $c\in\realN$ and store the result in register $i$:  $\amcrealfull(\bogenconst{i}{j}{c})(\vec{x}) =
  (\vec{x}\{x_i:= c\star x_{j}\})$;
\item $\copyy{i}{j}$ copies the value stored in register $j$ in register $i$: $\amcrealfull(\copyy{i}{j})(\vec{x}) =
  (\vec{x}\{x_i:= x_j\})$;
\item $\refcopy{i}{j}$ copies the value stored in register $j$ in the register whose index is the value stored in register $i$: $\amcrealfull(\refcopy{i}{j})(\vec{x}) =
  (\vec{x}\{x_{x_i}:= x_j\})$;
\item $\copyref{i}{j}$ copies the value stored in the register whose index is the value stored in register $j$ in register $i$: $\amcrealfull(\copyref{i}{j})(\vec{x}) =
  (\vec{x}\{x_i:= x_{x_j}\})$;
\item $\bosqrtn{i}{j}$ computes the $n$-th real root of the value stored in register $j$ and store the result in register $i$: $\amcrealfull(\bosqrtn{i}{j})(\vec{x}) =
  (\vec{x}\{x_i:= \sqrt[n]{x_j}\})$.
\end{itemize}
\end{definition}

\section{Algebraic models of computations as \amcs}\label{sec:machines}

\subsection{Algebraic computation trees}


The first model considered here will be that of \emph{algebraic computation
tree} as defined by Ben-Or \cite{Ben-Or83}. Let us note this model refines
the \emph{algebraic decision trees} model of Steele and Yao \cite{SteeleYao82},
a model of computation consisting in binary trees for which each branching
performs a test w.r.t. a polynomial and each leaf is labelled $\texttt{YES}$ or 
$\texttt{NO}$. Algebraic computation trees only allow tests w.r.t. $0$, while 
additional vertices corresponding to algebraic operations can be used to
construct polynomials.

\begin{definition}[algebraic computation trees, {\cite{Ben-Or83}}]
  An \emph{algebraic computation tree} on $\realN^n$ is a binary tree $T$ with an
  function assigning:
  \begin{itemize}
  \item to any vertex $v$ with only one child (simple vertex) an operational
    instruction of the form
    $f_v = f_{v_i} \star f_{v_j}$,
     $f_v = c \star f_{v_i}$, or
     $f_v = \sqrt{f_{v_i}}$,
    where $\star \in \{ +, -, \times, /\}$, ${v_i}, {v_j}$ are ancestors of $v$
    and $c\in \realN$ is a constant;
  \item to any vertex $v$ with two children a test instruction of the form
   $
      f_{v_i} \star 0
   $,
    where $\star \in \{ >,=,\geqslant \}$, and $v_i$ is an ancestor of $v$ 
    or $f_{v_i} \in \{x_1,\dots,x_n\}$;
  \item to any leaf an output $\texttt{YES}$ or $\texttt{NO}$.
  \end{itemize}
\end{definition}


Let $W \subseteq \realN^n$ be any set and $T$ be an algebraic computation
tree. We say that $T$ computes the membership problem for $W$ if for all $x \in
\realN^n$, the traversal of $T$ following $x$ ends on a leaf labelled
\texttt{YES} if and only if $x\in W$.

As algebraic computation trees are \emph{trees}, they will be represented by treeings,
i.e. $\amcrealfull$-programs whose set of control states can be ordered so that any
edge in the graphing is strictly increasing on its control states component.

\begin{definition}
Let $T$ be a computational $\amcrealfull$-treeing. The set of inputs $\inputs{T}$ 
(resp. outputs $\outputs{T}$) is the set of integers $k$ (resp. $i$) such that there 
exists an edge $e$ in $T$ satisfying that:
\begin{itemize} 
\item either $e$ is realised by one of $\boadd{i}{j}{k}$, $\boadd{i}{k}{j}$,
$\bosubstract{i}{j}{k}$, $\bosubstract{i}{k}{j}$, $\bomultiply{i}{j}{k}$, 
$\bomultiply{i}{k}{j}$, $\bodivide{i}{j}{k}$, $\bodivide{i}{k}{j}$
$\boaddconst{i}{k}{c}$, $\bosubstractconst{i}{k}{c}$, $\bomultiplyconst{i}{k}{c}$,
$\bodivideconst{i}{k}{c}$, $\bosqrtn{i}{k}$;
\item or the source of $e$ is one among 
$\realN^{\omega}_{k\geqslant 0}$, 
$\realN^{\omega}_{k\leqslant 0}$, 
$\realN^{\omega}_{k> 0}$, 
$\realN^{\omega}_{k< 0}$, 
$\realN^{\omega}_{k=0}$, and 
$\realN^{\omega}_{k\neq 0}$.
\end{itemize}

The \emph{effective input space} $\inputspaceE{T}$ of an $\amcact$-treeing $T$ is defined 
as the set of indices $k\in\omega$ belonging to $\inputs{T}$ but not to $\outputs{T}$. The
\emph{implicit input space} $\inputspaceI{T}$ of an $\amcact$-treeing $T$ is defined 
as the set of indices $k\in\omega$ such that $k\not \in \outputs{T}$.
\end{definition}

\begin{definition}
Let $T$ be an $\amcrealfull$-treeing, and assume that $1,2,\dots,n \in \inputspaceI{T}$. We say 
that $T$ computes the membership problem for $W\subseteq\realN^{n}$ in $k$ steps if $k$ 
successive iterations of $T$ restricted to 
$\{(x_i)_{i\in\omega} \in \realN^\omega \mid \forall 1\leqslant i\leqslant n, x_i=y_i\}\times\{0\}$ 
reach state $\top$ if and only if $(y_1,y_2,\dots,y_n)\in W$. 
\end{definition}

\begin{remark}
Let $\vec{x}=(x_1,x_2,\dots,x_n)$ be an element of $\realN^n$ and consider two elements 
$a,b$ in the subspace $\{(y_1,\dots, y_n,\dots)\in\realN^{\omega}
\mid\forall 1\geqslant i\geqslant n, y_i=x_i\}\times\{0\}$. One easily checks
that $\pi_{\Space{S}}(T^{k}(a))=\top$ if and only if $\pi_{\Space{S}}(T^{k}(b))
=\top$, where $\pi_{\Space{S}}$ is the projection onto the state space and 
$T^{k}(a)$ represents the $k$-th iteration of $T$ on $a$. It is therefore
possible to consider only a standard representative $\Interpret{\vec{x}}$ of $\vec{x} \in 
\realN^n$, for instance $(x_1,\dots, x_n,0,0,\dots) \in \realN^\omega$, to
decide whether $\vec{x}$ is accepted by $T$.
\end{remark}

\begin{definition}
  Let $T$ be an algebraic computation tree on $\realN^{n}$, and $T^{\circ}$ 
  be the associated directed acyclic graph, built from $T$ by merging all the 
  leaves tagged $\texttt{YES}$ in one leaf $\top$ and all the leaves tagged 
  $\texttt{NO}$ in one leaf $\bot$. Suppose the internal vertices are numbered
  $\{n+1,\dots, n+\ell \}$; the numbers $1,\dots,n$ being reserved for the input.

  We define $\Interpret{T}$ as the $\amcact$-graphing with control states
  $\{n+1,\dots, n+\ell, \top, \bot \}$ and where each internal vertex $i$ of
  $T^{\circ}$ defines either:
  \begin{itemize}
  \item a single edge of source $\realN^{\omega}$ realized by:
    \begin{itemize}
    \item $(\bogen{i}{j}{k}, i \mapsto t)$ ($\star\in\{+,-,\times\}$) if $i$ is associated to $f_{v_i} =
      f_{v_j} \star f_{v_k}$ and $t$ is the child of $i$;
    \item $(\bogenconst{i}{j}{c}, i \mapsto t)$ ($\star\in\{+,-,\times\}$) if $i$ is associated to $f_{v_i} =
     c \star f_{v_k}$ and $t$ is the child of $i$;
    \end{itemize}
  \item a single edge of source $\realN^{\omega}_{k\neq 0}$ realized by:
    \begin{itemize}
    \item $(\bodivide{i}{j}{k}, i \mapsto t)$ if $i$ is associated to
      $f_{v_i} = f_{v_j} / f_{v_k}$ and $t$ is the child of $i$;
    \item $(\bodivideconst{i}{k}{c}, i \mapsto t)$ if $i$ is associated to
      $f_{v_i} = c / f_{v_k}$ and $t$ is the child of $i$;
    \end{itemize}
  \item a single edge of source $\realN^{\omega}_{k\geqslant 0}\times\{i\}$ 
  realized by $(\bosqrt[2]{i}{k}, i \mapsto t)$ if $i$ is associated to 
  $f_{v_i} = \sqrt{f_{v_k}}$ and $t$ is the child of $i$;
  \item two edges if $i$ is associated to 
      $f_{v_i} \star 0$ (where $\star$ ranges in $>$, $\geqslant$) and its two sons 
      are $j$ and $k$. Those are of respective sources
      $\realN^{\omega}_{k\star 0}\times \{i\}$
      and 
      $\realN^{\omega}_{k\bar{\star} 0}
      \times \{i\}$ (where $\bar{\star}='\leqslant'$ if $\star='>'$, $\bar{\star}='<'$ if 
      $\star='\geqslant'$, and $\bar{\star}='\neq'$ if $\star='='$.),
      respectively realized by $(\identity, i \mapsto j)$ and $(\identity, i \mapsto k)$ 
  \end{itemize}
\end{definition}

\begin{proposition}\label{prop:fullyact}
Any algebraic computation tree $T$ of depth $k$ is faithfully and quantitatively 
interpreted as the $\amcrealfull$-program $\Interpret{T}$.
I.e. $T$ computes the membership problem for $W\subseteq\realN^{n}$ if
and only if $\Interpret{T}$ computes the membership problem for $W$ in $k$ steps
-- that is $\pi_{\Space{S}}(\Interpret{T}^k(\Interpret{\vec{x}}))=\top$.
%
\end{proposition}

As a corollary of this proposition, we get quantitative soundness.

\begin{theorem}\label{thm:soundACT}\label{thm:quantsoundact}
  The representation of \acts as $\amcrealfull$-programs is quantitatively sound.
\end{theorem}

\subsection{Algebraic circuits}

As we will recover Cucker's proof that $\complexityclass[\realN]{NC}{}\neq\PtimeReal$, we
introduce the model of \emph{algebraic circuits} and their representation 
as $\amcrealfull$-programs.

\begin{definition}\label{def:circuitsreels}
An algebraic circuit over the reals with inputs in $\realN^n$ is a finite 
directed graph whose vertices have labels in $\naturalN\times\naturalN$, 
that satisfies the following conditions:
\begin{itemize}
\item There are exactly $n$ vertices $v_{0,1},v_{0,2},\dots,v_{0,n}$ with 
first index $0$, and they have no incoming edges;
\item all the other vertices $v_{i,j}$ are of one of the following types:
\begin{enumerate}
\item arithmetic vertex: they have an associated arithmetic operation 
$\{+,-,\times,/\}$ and there exist natural numbers $l, k, r, m$ with $l, k < i$ 
such that their two incoming edges are of sources $v_{l,r}$ and $v_{k,m}$;
\item constant vertex: they have an associated real number $y$ and no 
incoming edges;
\item sign vertex: they have a unique incoming edge of source $v_{k,m}$ 
with $k < i$.
\end{enumerate}
\end{itemize}
We call \emph{depth} of the circuit the largest $m$ such that there exist 
a vertex $v_{m,r}$, and \emph{size} of the circuit the total number of vertices. 
A circuit of depth $d$ is \emph{decisional} if there is only one vertex $v_{d,r}$ 
at level $d$, and it is a sign vertex; we call $v_{d,r}$ the \emph{end vertex} of 
the decisional circuit. 
\end{definition}


To each vertex $v$ one inductively associates a function $f_v$ of the input 
variables in the usual way, where a sign node with input $x$ returns $1$ if 
$x > 0$ and $0$ otherwise. The accepted set of a decisional circuit 
$C$ is defined as the set $S\subseteq\realN^n$ of points whose image by the 
associated function is $1$, i.e. $S=f_v^{-1}(\{1\})$ where $v$ is the end vertex 
of $C$. 
 
We represent algebraic circuit as computational $\amcrealfull$-treeings as follows. 
The first index in the pairs $(i,j)\in\naturalN\times\naturalN$ are represented as 
states, the second index is represented as an index in the infinite product 
$\realN^\omega$, and vertices are represented as edges.

\begin{definition}
Let $C$ be an algebraic circuit, defined as a finite directed graph $(V,E,s,t,\ell)$ 
where $V\subset\naturalN\times\naturalN$, and 
$\ell:V\rightarrow \{\mathrm{init},+,-,\times,/,\mathrm{sgn}\}\cup\{\mathrm{const}_c\mid c\in \realN\}$ 
is a vertex labelling map. We suppose without loss of generality that for each 
$j\in\naturalN$, there is at most one $i\in\naturalN$ such that $(i,j)\in V$. We 
define $N$ as $\max\{j\in\naturalN \mid \exists i\in \naturalN, (i,j)\in V\}$.

We define the $\amcrealfull$-program $\Interpret{C}$ by choosing as set of control 
states $\{ i\in \naturalN \mid \exists j\in \naturalN, (i,j)\in V\}$ and the collection 
of edges $\{e_{(i,j)}\mid i\in\naturalN^*, j\in\naturalN, (i,j)\in V\}\cup\{e^+_{(i,j)}\mid i\in\naturalN^*, j\in\naturalN, (i,j)\in V, \ell(v)=\mathrm{sgn}\}$ 
realised as follows:
\begin{itemize}
\item if $\ell(v)=\mathrm{const}_c$, the edge $e_{(i,j)}$ is realised as 
$(\boaddconst{j}{c}{n_v}, 0 \mapsto i)$ of source $\realN^{\omega}_{n_v= 0}\times \{0\}$;
\item if $\ell(v)=\star$ ($\star\in\{+,-,\times\}$) of incoming edges $(k,l)$ and $(k',l')$, 
the edge $e_{(i,j)}$ is of source $\realN^\omega\times\{\max(k,k')\}$ and realised by 
$(\bogen{j}{l}{l'}, \max(k,k') \mapsto i)$;
\item if $\ell(v)=/$ of incoming edges $(k,l)$ and $(k',l')$, the edge $e_{(i,j)}$ is 
of source $\realN^\omega_{l'\neq 0}\times\{\max(k,k')\}$ and realised by 
$(\bodivide{j}{l}{l'}, \max(k,k') \mapsto i)$;
\item if $\ell(v)=\mathrm{sgn}$ of incoming edge $(k,l)$, the edges $e_{(i,j)}$ 
and $e^+_{(i,j)}$ are of respective sources
 $\realN^\omega_{n_v=0\wedge x_l\leqslant 0}\times\{k\}$ 
and $\realN^\omega_{n_v=0\wedge x_l>0}\times\{k\}$ realised by 
$(\identity, k \mapsto i)$ and  $(\boadd{j}{n_v}{1}, k \mapsto i)$ respectively.
\end{itemize}
\end{definition}

As each step of computation in the algebraic circuit is translated as going
through a single edge in the corresponding $\amcrealfull$-program, the following
result is straightforward.

\begin{theorem}\label{thm:soundalgcirc}\label{thm:quantsoundalgcirc}
  The representation of \algcirc as $\amcrealfull$-programs is quantitatively sound.
\end{theorem}

\subsection{Algebraic \srams}
\label{sec:algebraic-prams}

In this paper, we will consider algebraic parallel random access machines, that
act not on strings of bits, but on integers. In order to define those properly,
we first define the notion of (sequential) random access machine (\sram) before
considering their parallelisation.

A \sram \emph{command} is a pair $(\ell, I)$ of a \emph{line}
$\ell \in \naturalN^{\star}$ and an \emph{instruction} $\command{I}$ among the
following, where $i, j \in \naturalN$, $\star \in \{+, -, \times, / \}$, $c\in \integerN$
is a constant and $\ell, \ell' \in \naturalN^{\star}$ are lines:
\[\begin{array}{c}
\Skip;\hspace{0.8em}  \command{X_{i} \coloneqq c};\hspace{0.8em}
   \command{X_{i} \coloneqq X_{j} \star X_{k}};\hspace{0.8em}
                     \command{X_{i} \coloneqq X_{j}};\\
                         \command{X_{i} \coloneqq \sharp{}X_{j}};\hspace{0.8em}
                                  \command{\sharp{X_{i}}\coloneqq X_{j}};\hspace{0.8em}
                    \conditional{X_i = 0}{\ell}{\ell'}.
                    \end{array}
\]

A \sram \emph{machine} $M$ is then a finite set of commands such that the set of
lines is $\{1,2,\dots,\length{M}\}$, with $\length{M}$ the \emph{length} of
$M$. We will denote the commands in $M$ by $(i,\instr{i})$, i.e. $\instr{i}$
denotes the line $i$ instruction. 

Following Mulmuley \cite{Mulmuley99}, we will here make the assumption that the input
in the \sram (and in the \pram model defined in the next section) is split into \emph{numeric} 
and \emph{nonumeric} data -- e.g. in the \maxflow problem the 
nonnumeric data would specify the network and the numeric data would specify the 
edge-capacities -- and that indirect references use pointers depending only on 
nonnumeric data\footnote{Quoting Mulmuley: "We assume that the pointer involved in an 
indirect reference is not some numeric argument in the input or a quantity that depends on
it. For example, in the max- flow problem the algorithm should not use an edge-capacity as 
a pointer—which is a reasonable condition. To enforce this restriction, one initially puts an 
invalid-pointer tag on every numeric argument in the input. During the execution of an 
arithmetic instruction, the same tag is also propagated to the result if any operand has that 
tag. Trying to use a memory value with invalid-pointer tag results in error." 
\cite[Page 1468]{Mulmuley99}.\label{footnote:indirectref}}.
We refer the reader to Mulmuley's article for more details.


Machines in the \sram model can be represented as graphings w.r.t. the action $\amcfull$.
Intuitively the encoding works as follows. The notion of \emph{control state} allows to
represent the notion of \emph{line} in the program. Then, the action just
defined allows for the representation of all commands but the conditionals. The
conditionals are represented as follows: depending on the value of $X_{i}$ one
wants to jumps either to the line $\ell$ or to the line $\ell'$; this is easily
modelled by two different edges of respective sources
$\hyperplan{i}=\{\vec{x}~|~ x_{i}=0\}$ and
$\complementset{\hyperplan{i}}=\{\vec{x}~|~ x_{i}\neq0\}$.

\begin{definition}
  Let $M$ be a \sram machine. We define the translation $\Interpret{M}$ as the $\amcram$-program with set of
  control states $\{0,1,\dots,L,L+1\}$ where each line $\ell$ defines (in the following, 
  $\star\in\{+,-,\times\}$ and we write $\ell\plusplus$ the map $\ell\mapsto\ell+1$):
  \begin{itemize}[nolistsep,noitemsep]
  \item a single edge $e$ of source $\Space{X}\times\{\ell\}$ and realised
    by:
    \begin{itemize}[noitemsep,nolistsep]
    \item $(\identity,\ell\plusplus)\text{ if }\instr{\ell}=\Skip$;
    \item $(\boconst{i}{c},\ell\plusplus)\text{ if }\instr{\ell}=\command{X_i \coloneqq c}$;
    \item $(\bogen{i}{j}{k},\ell\plusplus)\text{ if }\instr{\ell}=
    \command{X_i \coloneqq X_j \star X_k} $;
    \item $(\copyy{i}{j},\ell\plusplus)\text{ if }\instr{\ell}=
    \command{X_{i} \coloneqq X_{j}}$;
    \item $(\copyref{i}{j},\ell\plusplus)\text{ if }\instr{\ell}=
    \command{X_{i}\coloneqq \sharp X_{j}}$;
    \item $(\refcopy{i}{j},\ell\plusplus)\text{ if }\instr{\ell}=
     \command{\sharp{}X_{i} \coloneqq X_{j}}$.
     \end{itemize}
  \item an edge $e$ of source $\complementset{\hyperplan{k}}\times\{\ell\}$ realised
    by $(\boeuclidivide{i}{j}{k},\ell\plusplus)$ if $\instr{\ell}$ is
    $\command{X_i \coloneqq X_j/X_k}$;
  \item a pair of edges $e,\complementset{e}$ of respective sources
    $\hyperplan{i}\times\{\ell\}$ and $\complementset{\hyperplan{i}}\times\{\ell\}$
    and realised by respectively $(\identity,\ell\mapsto \ell^{0})$ and
    $(\identity,\ell\mapsto \ell^{1})$, if the line is a conditional
    $\command{if~X_{i}=0~goto~\ell^{0}~else~\ell^{1}}$.
  \end{itemize}
  The translation $\Interpret{\iota}$ of an input $\iota\in\mathbf{Z}^d$ is
  the point $(\bar{\iota},0)$ where $\bar{\iota}$ is the sequence
  $(\iota_1,\iota_2,\dots,\iota_k,0,0,\dots)$.
\end{definition}

Now, the main result for the representation of \srams is the following. The proof
is straightforward, as each instruction corresponds to exactly one edge, except for
the conditional case (but given a configuration, it lies in the source of at most one 
of the two edges translating the conditional). 

\begin{theorem}\label{thm:soundsram}
  The representation of \srams as $\amcfull$-programs is quantitatively sound w.r.t.
  the translation just defined.
\end{theorem}

\subsection{The Crew operation and \prams}
\label{sec:crew}

Based on the notion of \sram, we are now able to consider their parallelisation,
namely \prams. A \pram $M$ is given as a finite sequence of \sram machines
$M_{1},\dots,M_{p}$, where $p$ is the number of \emph{processors} of $M$. Each
processor $M_{i}$ has access to its own, private, set of registers
$(\command{X}_{k}^{i})_{k\geqslant 0}$ and a \emph{shared memory} represented as
a set of registers $(\command{X}_{k}^{0})_{k\geqslant 0}$.

One has to deal with conflicts when several processors try to access the shared
memory simultaneously. We here chose to work with the \emph{Concurrent Read,
  Exclusive Write} (\crew) discipline: at a given step at which several
processors try to write in the shared memory, only the processor with the
smallest index will be allowed to do so. In order to model such parallel
computations, we abstract the \crew at the level of monoids. For this, we
suppose that we have two monoid actions
$\MonGaR{G}{R}\acton \Space{X}\times\Space{Y}$ and
$\MonGaR{H}{Q}\acton \Space{X}\times\Space{Z}$, where $\Space{X}$ represents the
shared memory. We then consider the subset $\#\subset G\times H$ of pairs of
generators that potentially conflict with one another -- the conflict relation.

\begin{definition}[Conflicted sum]
  Let $\MonGaR{G}{R}$, $\MonGaR{G'}{R'}$ be two monoids and
  $\# \subseteq G \times G'$. The \emph{conflicted sum of $\MonGaR{G}{R}$ and
    $\MonGaR{G'}{R'}$ over $\#$}, noted
  $\MonGaR{G}{R} \ast_{\#} \MonGaR{G'}{R'}$, is defined as the monoid with
  generators \((\{1\} \times G) \cup (\{2\} \times G')\) and relations
  \[ \begin{array}{c}(\{1\} \times R) \cup (\{2\} \times R') \cup \{ (\mathbf{1}, e)\} \cup \{ (\mathbf{1},
    e')\} \\
  \hspace{1em}\cup \{ \big((1,g)(2,g'), (2,g')(1,g)
    \big) \mid (g,g') \notin \# \}\end{array}
  \] where $\mathbf{1}$, $e$, $e'$ are the units of
  $\MonGaR{G}{R} \ast_{\#} \MonGaR{G'}{R'}$, $\MonGaR{G}{R}$ and $\MonGaR{G'}{R'}$
  respectively.

  In the particular case where $\# = (G \times H') \cup (H \times G')$, with
  $H, H'$ respectively subsets of $G$ and $G'$, we will write the sum
  $\MonGaR{G}{R}\ncproduct{H}{H'}\MonGaR{G'}{R'}$.
\end{definition}

\begin{remark}
  When the conflict relation $\#$ is empty, this defines the usual direct product
  of monoids. This corresponds to the case in which no conflicts can arise
  w.r.t. the shared memory. In other words, the direct product of monoids
  corresponds to the parallelisation of processes \emph{without shared memory}.

  Dually, when the conflict relation is full ($\# = G \times G'$), this defines the free
  product of the monoids.
\end{remark}

\begin{definition}
  Let $\alpha: M\acton\Space{X}\times \Space{Y}$ be a monoid action. We say that
  an element $m\in M$ is \emph{central relatively to $\alpha$} (or just
  \emph{central}) if $m$ acts as the identity on
  $\Space{X}$,
  i.e.\footnote{Here and in the following, we denote by $;$ the sequential
  composition of functions. I.e. $f;g$ denotes what is usually written $g\circ f$.} 
  $\alpha(m);\pi_{X}=\pi_{X}$.
\end{definition}

Intuitively, central elements are those that will not affect the shared
memory. As such, only \emph{non-central elements} require care
when putting processes in parallel.

\begin{definition}
  Let $\MonGaR{G}{R}\acton \Space{X}\times\Space{Y}$ be an \amc. We note
  $Z_{\alpha}$ the set of central elements and $\bar{Z}_{\alpha}(G)=\{m\in G~|~ m\not\in Z_{\alpha}\}$.
\end{definition}

\begin{definition}[The \crew of \amcs]\label{def:crew}
  Let $\alpha: \MonGaR{G}{R}\acton \Space{X}\times\Space{Y}$ and
  $\beta: \MonGaR{H}{Q}\acton \Space{X}\times\Space{Z}$ be \amcs. We define the
  \amc
  $\crew(\alpha,\beta): \MonGaR{G}{R}\ncproduct{\bar{Z}_{\alpha}(G)}{\bar{Z}_{\beta}(G')}\MonGaR{G'}{R'}\acton \Space{X}\times\Space{Y}\times\Space{Z}$
  by letting $\crew(\alpha,\beta)(m,m')=\alpha(m)\ast\beta(m')$ on elements of
  $G\times G'$, where\footnote{We denote $\sigma_{X,Y}:\Space{Y}\times\Space{X}\rightarrow\Space{X}\times\Space{Y}$ the map defined as $(y,x)\mapsto(x,y)$.}:
  \begin{align}
    \lefteqn{\alpha(m)\ast\beta(m')=}\nonumber\\
    &\left\{
      \begin{array}{ll}
        \Delta_1;[\alpha(m);\pi_{Y},\beta(m')];[\sigma_{X,Y},\identity{Z}]  & \text{ if $m\not\in\bar{Z}_{\alpha}(G), m'\in\bar{Z}_{\beta}(G')$,}\\
        \Delta_2;[\alpha(m),\beta(m');\pi_{Z}] & \text{ otherwise,}
        \end{array}\right.\nonumber
  \end{align}
  with
  $\Delta_i: 
  \Space{X}\times\Space{Y}\times\Space{Z}\rightarrow\Space{X}\times\Space{Y}\times\Space{X}\times\Space{Z}$ defined\footnote{Formally, the definition of $\Delta_i$ is parametrised by the choice of a point $x_0\in\Space{X}$, but the map $\alpha(m)\ast\beta(m')$ does not depend on this choice because of the projections on $\Space{Y}$ and $\Space{Z}$.} as:
  \begin{align*}
  \Delta_1 :& (x,y,z) \mapsto (x_0,y,x,z) \\
  \Delta_2 :& (x,y,z) \mapsto (x,y,x_0,z).
  \end{align*}
\end{definition}

We can now define \amc of \prams and thus the interpretations of \prams as
abstract programs. For each integer $p$, we define the \amc $\crew^{p}(\amcfull)$. 
This allows the consideration of up to $p$ parallel \srams: the translation of such a
\sram with $p$ processors is defined by extending the translation of \srams
by considering a set of states equal to $L_{1}\times L_{2}\times \dots \times L_{p}$ 
where for all $i$ the set $L_{i}$ is the set of lines of the $i$-th processor.

Now, to deal with arbitrary large \prams, i.e. with arbitrarily large number of
processors, one considers the following \amc defined as a \emph{direct limit}.

\begin{definition}[The \amc of  \prams]
  Let $\alpha: M\acton\Space{X\times X}$ be the \amc $\amcfull$.
  The \amc of \prams is defined as $\amcprams=\varinjlim\crew^{k}(\alpha)$, where
  $\crew^{k-1}(\alpha)$ is identified with a restriction of $\crew^{k}(\alpha)$
  through
  $\crew^{k-1}(\alpha)(m_{1},\dots,m_{k-1})\mapsto
  \crew^{k}(\alpha)(m_{1},\dots,m_{k-1},1)$.
\end{definition}

Remark that the underlying space of the \pram{} \amc $\amcprams$ is defined 
as the union \( \cup_{n\in \omega} \integerN^{\omega} \times (\integerN^{\omega})^n\) 
which we will write \( \integerN^{\omega} \times (\integerN^{\omega})^{(\omega)}\).
In practise a given $\amcprams$-program admitting a finite $\amcprams$ representative
will only use elements in $\crew^{p}(\amcfull)$, and can therefore be understood as a 
$\crew^{p}(\alpha)$-program. 

\begin{theorem}\label{thm:soundpram}\label{thm:quantsoundprams}
  The representation of \prams as $\amcprams$-programs is quantitatively sound.
\end{theorem}

\subsection{Real \prams}

These definitions and results stated for integer-valued \prams can be adapted to 
define \emph{real-valued \prams} and their translation as $\amcrealfull$-programs.

A real-valued \sram \emph{command} is a pair $(\ell, I)$ of a \emph{line}
$\ell \in \naturalN^{\star}$ and an \emph{instruction} $\command{I}$ among the
following, where $i, j \in \naturalN$, $\star \in \{+, -, \times, / \}$, $c\in \integerN$
is a constant and $\ell, \ell' \in \naturalN^{\star}$ are lines:
\[\begin{array}{c}
\Skip;\hspace{0.8em}  \command{X_{i} \coloneqq c};\hspace{0.8em}
   \command{X_{i} \coloneqq X_{j} \star X_{k}};\hspace{0.8em}
                     \command{X_{i} \coloneqq X_{j}};\\
                         \command{X_{i} \coloneqq \sharp{}X_{j}};\hspace{0.8em}
                                  \command{\sharp{X_{i}}\coloneqq X_{j}};\hspace{0.8em}
                    \conditional{X_i = 0}{\ell}{\ell'}.
                    \end{array}
\]

We consider a restriction for pointers similar to that considered in the case
of integer-valued \srams.
A real-valued \sram \emph{machine} $M$ is then a finite set of commands 
such that the set of lines is $\{1,2,\dots,\length{M}\}$, with $\length{M}$ the 
\emph{length} of $M$. We will denote the commands in $M$ by $(i,\instr{i})$, 
i.e. $\instr{i}$ denotes the line $i$ instruction.

A real-valued \pram $M$ is given as a finite sequence of real-valued \sram machines
$M_{1},\dots,M_{p}$, where $p$ is the number of \emph{processors} of $M$. Each
processor $M_{i}$ has access to its own, private, set of registers
$(\command{X}_{k}^{i})_{k\geqslant 0}$ and a \emph{shared memory} represented as
a set of registers $(\command{X}_{k}^{0})_{k\geqslant 0}$.
Again, we chose to work with the \emph{Concurrent Read, Exclusive Write} (\crew) 
discipline as it is well translated through the $\crew$ operation of \amcs.

The following definition provides the straightforward definition of real-valued \prams
as graphings.

\begin{definition}
  Let $M$ be a real-valued \sram machine. We define the translation $\Interpret{M}$ 
  as the $\amcrealfull$-program with set of
  control states $\{0,1,\dots,L,L+1\}$ where each line $\ell$ defines (in the following, 
  $\star\in\{+,-,\times\}$ and we write $\ell\plusplus$ the map $\ell\mapsto\ell+1$):
  \begin{itemize}[nolistsep,noitemsep]
  \item a single edge $e$ of source $\Space{X}\times\{\ell\}$ and realised
    by:
    \begin{itemize}[noitemsep,nolistsep]
    \item $(\identity,\ell\plusplus)\text{ if }\instr{\ell}=\Skip$;
    \item $(\boconst{i}{c},\ell\plusplus)\text{ if }\instr{\ell}=\command{X_i \coloneqq c}$;
    \item $(\bogen{i}{j}{k},\ell\plusplus)\text{ if }\instr{\ell}=
    \command{X_i \coloneqq X_j \star X_k} $;
    \item $(\copyy{i}{j},\ell\plusplus)\text{ if }\instr{\ell}=
    \command{X_{i} \coloneqq X_{j}}$;
    \item $(\copyref{i}{j},\ell\plusplus)\text{ if }\instr{\ell}=
    \command{X_{i}\coloneqq \sharp X_{j}}$;
    \item $(\refcopy{i}{j},\ell\plusplus)\text{ if }\instr{\ell}=
     \command{\sharp{}X_{i} \coloneqq X_{j}}$.
     \end{itemize}
  \item an edge $e$ of source $\realN_{k\neq 0}\times\{\ell\}$ realised
    by $(\bodivide{i}{j}{k},\ell\plusplus)$ if $\instr{\ell}$ is
    $\command{X_i \coloneqq X_j/X_k}$;
  \item a pair of edges $e,\complementset{e}$ of respective sources
    $\realN_{i=0}\times\{\ell\}$ and $\realN_{i\neq 0}\times\{\ell\}$
    and realised by respectively $(\identity,\ell\mapsto \ell^{0})$ and
    $(\identity,\ell\mapsto \ell^{1})$, if the line is a conditional
    $\command{if~X_{i}=0~goto~\ell^{0}~else~\ell^{1}}$.
  \end{itemize}
  The translation $\Interpret{\iota}$ of an input $\iota\in\mathbf{Z}^d$ is
  the point $(\bar{\iota},0)$ where $\bar{\iota}$ is the sequence
  $(\iota_1,\iota_2,\dots,\iota_k,0,0,\dots)$.
\end{definition}

For each integer $p$, we then define the \amc $\crew^{p}(\amcrealfull)$. 
This allows the consideration of up to $p$ parallel real-valued \srams: the translation 
of such a \sram with $p$ processors is defined by extending the translation of 
real-valued \srams just defined
by considering a set of states equal to $L_{1}\times L_{2}\times \dots \times L_{p}$ 
where for all $i$ the set $L_{i}$ is the set of lines of the $i$-th processor.

Since we need to translate arbitrary large real-valued \prams, i.e. with arbitrarily large 
number of processors, one considers the following \amc defined as a \emph{direct limit}.

\begin{definition}[The \amc of real-valued \prams]
  Let $\alpha: M\acton\Space{X\times X}$ be the \amc $\amcrealfull$.
  The \amc of real-valued \prams is defined as $\amcrealpram=\varinjlim\crew^{k}(\alpha)$, 
  where $\crew^{k-1}(\alpha)$ is identified with a restriction of $\crew^{k}(\alpha)$
  through
  \[ \crew^{k-1}(\alpha)(m_{1},\dots,m_{k-1})\mapsto
  \crew^{k}(\alpha)(m_{1},\dots,m_{k-1},1).\]
\end{definition}

Then the following results are quite straightforward.
\begin{theorem}\label{thm:soundrealvaluedram}\label{thm:soundrealvaluedpram}
  The representation of real-valued \srams as $\amcrealfull$-programs is quantitatively sound.
  The representation of real-valued \prams as $\amcrealpram$-programs is quantitatively sound.
\end{theorem}

\section{Entropy and Cells}

\subsection{Topological Entropy}

Topological Entropy was introduced in the context of dynamical systems in an attempt to classify the latter w.r.t. conjugacy. The topological entropy of a dynamical system is a value representing the average exponential growth rate of the number of orbit segments distinguishable with a finite (but arbitrarily fine) precision. The definition is based on the notion of open covers.

\paragraph{Open covers.} Given a topological space $\Space{X}$, an \emph{open cover} of $\Space{X}$ is a family $\mathcal{U}=(U_{i})_{i\in I}$ of open subsets of $\Space{X}$ such that $\cup_{i\in I}U_{i}=\Space{X}$. A finite cover $\mathcal{U}$ is a cover whose indexing set is finite. A \emph{subcover} of a cover $\mathcal{U}=(U_{i})_{i\in I}$ is a sub-family $\mathcal{S}=(U_{j})_{j\in J}$ for $J\subseteq I$ such that $\mathcal{S}$ is a cover, i.e. such that $\cup_{j\in J}U_{j}=\Space{X}$.

We will denote by $\opencovers{X}$ (resp. $\finiteopencovers{X}$) the set of all open covers (resp. all finite open covers) of the space $\Space{X}$.

We now define two operations on open covers that are essential to the definition of entropy.
An open cover $\mathcal{U}=(U_{i})_{i\in I}$, together with a continuous function $f:\Space{X}\rightarrow\Space{X}$, defines the inverse image open cover $f^{-1}(\mathcal{U})=(f^{-1}(U_{i}))_{i\in I}$. Note that if $\mathcal{U}$ is finite, $f^{-1}(\mathcal{U})$ is finite as well. Given two open covers $\mathcal{U}=(U_{i})_{i\in I}$ and $\mathcal{V}=(V_{j})_{j\in J}$, we define their join $\mathcal{U}\vee\mathcal{V}$ as the family $(U_{i}\cap V_{j})_{(i,j)\in I\times J}$. Once again, if both initial covers are finite, their join is finite.

\paragraph{Entropy} Usually, entropy is defined for continuous maps on a compact set, following the original definition by Adler, Konheim and McAndrews \cite{AKmA}. Using the fact that arbitrary open covers have a finite subcover, this allows one to ensure that the smallest subcover of any cover is finite. I.e. given an arbitrary cover $\mathcal{U}$, one can consider the smallest -- in terms of cardinality -- subcover $\mathcal{S}$ and associate to $\mathcal{U}$ the finite quantity $\log_{2}(\card{\mathcal{S}})$. This quantity, obviously, need not be finite in the general case of an arbitrary cover on a non-compact set. 

However, a generalisation of entropy to non-compact sets can easily be defined by restricting the usual definition to \emph{finite} covers\footnote{This is discussed by Hofer \cite{Hofer75} together with another generalisation based on the Stone-\v{C}ech compactification of the underlying space.}. This is the definition we will use here.

\begin{definition}
Let $\Space{X}$ be a topological space, and $\mathcal{U}=(U_{i})_{i\in I}$ be a finite cover of $\Space {X}$. We define the quantity $H_{\Space{X}}^{0}(\mathcal{U})$ as 
\[\min\{\log_{2}(\card{J})~|~ J\subset I, \cup_{j\in J}U_{j}=\Space{X}\}.\]
\end{definition}
In other words, if $k$ is the cardinality of the smallest subcover of $\mathcal{U}$, $H^{0}(\mathcal{U})=\log_{2}(k)$.

\begin{definition}\label{def:Hk}
Let $\Space{X}$ be a topological space and $f: \Space{X}\rightarrow\Space{X}$ be a continuous map. For any finite open cover $\mathcal{U}$ of $\Space{X}$, we define:
\[ H_{\Space{X}}^{k}(f,\mathcal{U})=\frac{1}{k}H_{\Space{X}}^{0}(\mathcal{U}\vee f^{-1}(\mathcal{U})\vee\dots\vee f^{-(k-1)}(\mathcal{U})). \]
\end{definition}

One can show that the limit $\lim_{n\rightarrow\infty}H_{\Space{X}}^{n}(f,\mathcal{U})$ exists and is finite \cite{AKmA}; it will be noted $h(f,\mathcal{U})$. The topological entropy of $f$ is then defined as the supremum of these values, when $\mathcal{U}$ ranges over the set of all finite covers $\finiteopencovers{X}$.

\begin{definition}\label{def:entropy}
Let $\Space{X}$ be a topological space and $f: \Space{X}\rightarrow\Space{X}$ be a continuous map. The \emph{topological entropy} of $f$ is defined as $h(f)=\sup_{\mathcal{U}\in\finiteopencovers{X}} h(f,\mathcal{U})$.
\end{definition}

\subsection{Graphings and Entropy}

We now need to define the entropy of \emph{deterministic graphing}. As mentioned briefly already, deterministic graphings on a space $\Space{X}$ are in one-to-one correspondence with partial dynamical systems on $\Space{X}$. To convince oneself of this, it suffices to notice that any partial dynamical system can be represented as a graphing with a single edge, and that if the graphing $G$ is deterministic its edges can be glued together to define a partial continuous function $[G]$. Thus, we only need to extend the notion of entropy to partial maps, and we can then define the entropy of a graphing $G$ as the entropy of its corresponding map $[G]$. 

Given a finite cover $\mathcal{U}$, the only issue with partial continuous maps is that $f^{-1}(\mathcal{U})$ is not in general a cover. Indeed, $\{f^{-1}(U)~|~U\in \mathcal{U}\}$ is a family of open sets by continuity of $f$ but the union $\cup_{U\in \mathcal{U}}f^{-1}(U)$ is a strict subspace of $\Space{X}$ (namely, the domain of $f$). It turns out the solution to this problem is quite simple: we notice that $f^{-1}(\mathcal{U})$ is a cover of $f^{-1}(\Space{X})$ and now work with covers of subspaces of $\Space{X}$. Indeed, $\mathcal{U}\vee f^{-1}(\mathcal{U})$ is itself a cover of $f^{-1}(\Space{X})$ and therefore the quantity $H_{\Space{X}}^{2}(f,\mathcal{U})$ can be defined as $(1/2)H_{f^{-1}(\Space{X})}^{0}(\mathcal{U}\vee f^{-1}(\mathcal{U}))$.

We now generalise this definition to arbitrary iterations of $f$ by extending Definitions \ref{def:Hk} and \ref{def:entropy} to partial maps as follows.
\begin{definition}\label{def:partialentropy}
Let $\Space{X}$ be a topological space and $f: \Space{X}\rightarrow\Space{X}$ be a continuous partial map. For any finite open cover $\mathcal{U}$ of $\Space{X}$, we define:
\[ H_{\Space{X}}^{k}(f,\mathcal{U})=\frac{1}{k}H_{f^{-k+1}(\Space{X})}^{0}(\mathcal{U}\vee f^{-1}(\mathcal{U})\vee\dots\vee f^{-(k-1)}(\mathcal{U})). \]
The \emph{entropy} of $f$ is then defined as $h(f)=\sup_{\mathcal{U}\in\finiteopencovers{X}} h(f,\mathcal{U})$, where $h(f,\mathcal{U})$ is again defined as the limit $\lim_{n\rightarrow\infty}H_{\Space{X}}^{n}(f,\mathcal{U})$ \cite{Hofer75}.
\end{definition}

Now, let us consider the special case of a graphing $G$ with set of control states $S^{G}$. For an intuitive understanding, one can think of $G$ as the representation of a \pram machine. We focus on the specific open cover indexed by the set of control states, i.e. $\mathcal{S}=(\Space{X}\times\{s\}_{s\in S^{G}})$, and call it \emph{the states cover}. We will now show how the partial entropy $H^{k}(G,\mathcal{S})$ is related to the set of \emph{admissible sequence of states}. Let us define those first.

\begin{definition}
Let $G$ be a graphing, with set of control states $S^{G}$. An admissible sequence of states is a sequence $\mathbf{s}=s_{1}s_{2}\dots s_{n}$ of elements of $S^{G}$ such that for all $i\in\{1,2,\dots,n-1\}$ there exists a subset $C_i$ of $\Space{X}$ -- i.e. a set of configurations -- such that $G$ contains an edge from $C_i\times\{s_{i}\}$ to (a subspace of) $C_{i+1}\times\{s_{i+1}\}$ (with the convention that $C_n=\Space{X}$).
\end{definition}

\begin{example}
As an example, let us consider the very simple graphing with four control states $a,b,c,d$ and edges from $\Space{X}\times\{a\}$ to $\Space{X}\times\{b\}$, from $\Space{X}\times\{b\}$ to $\Space{X}\times\{c\}$, from $\Space{X}\times\{c\}$ to $\Space{X}\times\{b\}$ and from $\Space{X}\times\{c\}$ to $\Space{X}\times\{d\}$. 
Then the sequences $abcd$ and $abcbcbc$ are admissible, but the sequences $aba$, $abcdd$, and $abcba$ are not.
\end{example}

\begin{lemma}\label{lem:admseqHk}
Let $G$ be a graphing, and $\mathcal{S}$ its states cover. Then for all integer $k$, the set $\admss{k}$ of admissible sequences of states of length $k>1$ is of cardinality $2^{k.H^{k}(G,\mathcal{S})}$.
\end{lemma}

\begin{proof}
We show that the set $\admss{k}$ of admissible sequences of states of length $k$ has the same cardinality as the smallest subcover of $\mathcal{S}\vee [G]^{-1}(\mathcal{S})\vee\dots\vee [G]^{-(k-1)}(\mathcal{S}))$. Hence $H^{k}(G,\mathcal{S})=\frac{1}{k}\log_{2}(\card{\admss{k}})$, which implies the result.

The proof is done by induction. As a base case, let us consider the set of $\admss{2}$ of admissible sequences of states of length $2$ and the open cover $\mathcal{V}=\mathcal{S}\vee [G]^{-1}(\mathcal{S})$ of $D=[G]^{-1}(\Space{X})$. An element of $\mathcal{V}$ is an intersection $\Space{X}\times\{s_{1}\}\cap [G]^{-1}(\Space{X}\times\{s_{2}\})$, and it is therefore equal to $C[s_{1},s_{2}]\times\{s_{1}\}$ where $C[s_{1},s_{2}]\subset\Space{X}$ is the set $\{x\in\Space{X}~|~[G](x,s_{1})\in\Space{X}\times\{s_{2}\}\}$. This set is empty if and only if the sequence $s_{1}s_{2}$ belongs to $\admss{2}$. Moreover, given another sequence of states $s'_{1}s'_{2}$ (not necessarily admissible), the sets $C[s_{1},s_{2}]$ and $C[s_{1},s_{2}]$ are disjoint. Hence a set $C[s_{1},s_{2}]$ is \emph{removable from the cover $\mathcal{V}$} if and only if the sequence $s_{1}s_{2}$ is not admissible. This implies the result for $k=2$.

The step for the induction is similar to the base case. It suffices to consider the partition $\mathcal{S}_{k}=\mathcal{S}\vee [G]^{-1}(\mathcal{S})\vee\dots\vee [G]^{-(k-1)}(\mathcal{S}))$ as $\mathcal{S}_{k-1}\vee[G]^{-(k-1)}(\mathcal{S})$. By the same argument, one can show that elements of $\mathcal{S}_{k-1}\vee[G]^{-(k-1)}(\mathcal{S})$ are of the form $C[\mathbf{s}=(s_{0}s_{1}\dots s_{k-1}),s_{k}]\times\{s_{1}\}$ where $C[\mathbf{s},s_{k}]\subset\Space{X}$ is the set $\{x\in\Space{X}~|~\forall i=2,\dots,k, [G]^{i-1}(x,s_{1})\in \Space{X}\times\{s_{i}\}\}$. Again, these sets $C[\mathbf{s},s_{k}]$ are pairwise disjoint and empty if and only if the sequence $s_{0}s_{1}\dots s_{k-1},s_{k}$ is not admissible.
\end{proof}


A tractable bound on the number of admissible sequences of states can be obtained by noticing that the sequence $H^{k}(G,\mathcal{S})$ is \emph{sub-additive}, i.e. $H^{k+k'}(G,\mathcal{S})\leqslant H^{k}(G,\mathcal{S})+H^{k'}(G,\mathcal{S})$. A consequence of this is that $H^{k}(G,\mathcal{S})\leqslant kH^{1}(G,\mathcal{S})$. Thus the number of admissible sequences of states of length $k$ is bounded by $2^{k^{2}H^{1}(G,\mathcal{S})}$. We now study how the cardinality of admissible sequences can be related to the entropy of $G$.

\begin{lemma}
For all $\epsilon>0$, there exists an integer $N$ such that for all $k\geqslant N$, $H^{k}(G,\mathcal{U})<h([G])+\epsilon$.
\end{lemma}

\begin{proof}
Let us fix some $\epsilon>0$. Notice that if we let $H_{k}(G,\mathcal{U})=H^{0}(\mathcal{U}\vee [G]^{-1}(\mathcal{U})\vee\dots\vee [G]^{-(k-1)}(\mathcal{U})))$, the sequence $H_{k}(U)$ satisfies $H_{k+l}
(\mathcal{U})\leqslant H_{k}(\mathcal{U})+H_{l}(\mathcal{U})$. By Fekete's lemma on subadditive sequences, this implies that $\lim_{k\rightarrow\infty}H_{k}/k$ exists and is equal to $\inf_{k}H_{k}/k$. Thus $h([G],\mathcal{U})=\inf_{k}H_{k}/k$. 

Now, the entropy $h([G])$ is defined as $\sup_{\mathcal{U}} \lim_{k\rightarrow\infty} H_{k}(\mathcal{U})/k$. This then rewrites as $\sup_{\mathcal{U}} \inf_{k} H_{k}(\mathcal{U})/k$. We can conclude that $h([G])\geqslant \inf_{k} H_{k}(\mathcal{U})/k$ for all finite open cover $\mathcal{U}$. 

Since $\inf_{k} H_{k}(\mathcal{U})/k$ is the limit of the sequence $H_{k}/k$, there exists an integer $N$ such that for all $k\geqslant N$ the following inequality holds: $\abs{H_{k}(\mathcal{U})/k-\inf_{k}H_{k}(\mathcal{U})/k}<\epsilon$, which rewrites as  $H_{k}(\mathcal{U})/k-\inf_{k}H_{k}(\mathcal{U})/k<\epsilon$. From this we deduce $H_{k}(\mathcal{U})/k<h([G])+\epsilon$, hence $H^{k}(G,\mathcal{U})<h([G])+\epsilon$ since $H^{k}(G,\mathcal{U})=H_{k}(G,\mathcal{U})$.
\end{proof}

\begin{lemma}\label{lem:admseqHk2}
Let $G$ be a graphing, and let $c: k\mapsto\card{\admss{k}}$. Then $c(k) = O(2^{k.h([G])})$ as $k$ goes to infinity.
\end{lemma}

Lastly, we prove a result bounding the entropy of a map $\alpha(m)\ast\beta(m')$ in the \crew of \amcs. The result is essentially a consequence of the product rule (\cite[Theorem 3]{AKmA}, or \cite{entropyproduct}) stating that the entropy of a product $h(f\times g)$ is bounded above by the sum $h(f)+h(g)$.

\begin{lemma}\label{lem:entropycrew}
Let $\alpha: \MonGaR{G}{R}\acton \Space{X}\times\Space{Y}$ and
  $\beta: \MonGaR{H}{Q}\acton \Space{X}\times\Space{Z}$ be \amcs such that every non-central element of $\beta$ acts as the identity on $\Space{Z}$.
  Then for all $m\in \MonGaR{G}{R}$ and $m'\in \MonGaR{H}{Q}$, the entropy of $\alpha(m)\ast\beta(m')$ is bounded by the sum of the entropies of $\alpha(m)$ and $\beta(m')$:
  \[ h(\alpha(m)\ast\beta(m')\leqslant h(\alpha(m))+h(\beta(m')). \]
\end{lemma}

\begin{proof}
We show that the entropy of $\alpha(m)\ast\beta(m')$ is bounded by the entropy of $\alpha(m)\times\beta(m)$. The result then follows by the product rule \cite{entropyproduct}. We distinguish two cases: the first case is when one of $\alpha(m)$ or $\beta(m')$ is central, i.e. $\alpha(m);\pi_{\Space{X}}=\pi_{\Space{X}}$ or $\beta(m');\pi_{\Space{X}}=\pi_{\Space{X}}$, the second case is when both $\alpha(m)$ and $\beta(m')$ act non-trivially on $\Space{X}$.

For the first case, we may consider that $\beta(m')$ is central without loss of generality. It is then of the form $\tilde{\beta}\times\identity[\Space{X}]$ with $\tilde{\beta}:\Space{Z}\rightarrow\Space{Z}$, and the key observation is that 
\[ \alpha(m)\ast\beta(m') = \alpha(m)\times\tilde{\beta} \]
in this case. We now apply the product rule on both identities. From the first identity, we get 
\[ h(\beta(m'))=h(\tilde{\beta})+h(\identity[\Space{X}]=h(\tilde{\beta}), \]
since the entropy of the identity is equal to $0$, and from the second identity, we get
\[ h(\alpha(m)\ast\beta(m'))\leq h(\alpha(m'))+ h(\tilde{\beta}).\]
Combining both we obtain that $h(\alpha(m)\ast\beta(m'))=h(\alpha(m))+h(\beta(m'))$.

For the second case, the definition of $\alpha(m)\ast\beta(m')$ states that it is equal to the following map:
\[ \Delta_2;\alpha(m)\times (\beta(m');\pi_\Space{Z}). \]
Diagrammatically, this is defined as:
\begin{center}
\begin{tikzpicture}
	\node (Xll) at (-2,0) {\scriptsize{$\Space{X}$}};
	\node (Yll) at (-2,-0.5) {\scriptsize{$\Space{Y}$}};
	\node (Zll) at (-2,-1) {\scriptsize{$\Space{Z}$}};

	\node (X1l) at (0,0) {\scriptsize{$\Space{X}$}};
	\node (Yl) at (0,-0.5) {\scriptsize{$\Space{Y}$}};
	\node (X2l) at (0,-1) {\scriptsize{$\Space{X}$}};
	\node (Zl) at (0,0-1.5) {\scriptsize{$\Space{Z}$}};
	
	\node (X1r) at (3,0) {\scriptsize{$\Space{X}$}};
	\node (Yr) at (3,-0.5) {\scriptsize{$\Space{Y}$}};
	\node (X2r) at (3,-1) {\scriptsize{$\Space{X}$}};
	\node (Zr) at (3,0-1.5) {\scriptsize{$\Space{Z}$}};

	\node (Xrr) at (5,0) {\scriptsize{$\Space{X}$}};
	\node (Yrr) at (5,-0.5) {\scriptsize{$\Space{Y}$}};
	\node (Zrr) at (5,0-1) {\scriptsize{$\Space{Z}$}};
	
	\draw[->] (Xll) -- (X1l) {};
	\draw[->] (Xll) -- (X2l) {};
	\draw[->] (Yll) -- (Yl) {};
	\draw[->] (Zll) -- (Zl) {};
	
	\draw[-] (1,0.2) -- (2,0.2) -- (2,-0.7) -- (1,-0.7) -- (1,0.2) {};
	\draw[->] (X1l) -- (1,0) {};
	\draw[->] (Yl) -- (1,-0.5) {};
	\draw[->] (2,0) -- (X1r) {};
	\draw[->] (2,-0.5) -- (Yr) {};
	\node (alpha) at (1.5,-0.25) {\scriptsize{$\alpha(m)$}};
	\draw[-] (1,-0.8) -- (2,-0.8) -- (2,-1.7) -- (1,-1.7) -- (1,-0.8) {};
	\draw[->] (X2l) -- (1,-1) {};
	\draw[->] (Zl) -- (1,-1.5) {};
	\draw[->] (2,-1) -- (X2r) {};
	\draw[->] (2,-1.5) -- (Zr) {};
	\node (beta) at (1.5,-1.25) {\scriptsize{$\beta(m')$}};
	
	\draw[->] (X1r) -- (Xrr) {};
	\draw[->] (Yr) -- (Yrr) {};
	\draw[->] (Zr) -- (Zrr) {};
		
	\draw[-o] (X2r) -- (4,-1) {};
\end{tikzpicture}
\end{center}
We will now bound the entropy of $\alpha(m)\ast\beta(m')$. This is where we will use the hypothesis that $\beta(m')$ acts as the identity on $\Space{Z}$, i.e. that $\beta(m')(x,z)=(x',z)$. Indeed, from this hypothesis, one can deduce that
\[ \alpha(m)\ast\beta(m')(x,y,z)=\alpha(m)(x,y)\times \identity[\Space{Z}](z), \]
hence 
\[ h(\alpha(m)\ast\beta(m'))=h(\alpha(m))+h(\identity[\Space{Z}])=h(\alpha(m))\leq h(\alpha(m))+h(\beta(m')),\]
since $h(\identity[\Space{Z}])=0$ and $h(\beta(m'))\geq 0$.
\end{proof}


\subsection{Cells Decomposition}

Now, let us consider a deterministic graphing $G$, with its state cover $\mathcal{S}$. We fix a length $k>2$ and reconsider the sets $C[\mathbf{s}]=C[(s_{1}s_{2}\dots s_{k-1},s_{k})]$ (for a sequence of states $\mathbf{s}=s_{1}s_{2}\dots s_{k}$) that appear in the proof of Lemma \ref{lem:admseqHk}. The set $(C[\mathbf{s}])_{\mathbf{s}\in\admss{k}}$ is a partition of the space $[G]^{-k+1}(\Space{X})$.  



This decomposition splits the set of initial configurations into cells satisfying the following property: \emph{for any two initial configurations contained in the same cell $C[\mathbf{s}]$, the $k$-th first iterations of $G$ goes through the same admissible sequence of states $\mathbf{s}$}.


\begin{definition}
Let $G$ be a deterministic graphing, with its state cover $\mathcal{S}$. Given an integer $k$, we define the $k$-fold decomposition of $\Space{X}$ along $G$ as the partition $\{ C[\mathbf{s}] ~|~ \mathbf{s}\in\admss{k} \}$.
\end{definition}

Then Lemma \ref{lem:admseqHk} provides a bound on the cardinality of the $k$-th cell decomposition.
Using the results in the previous section, we can then obtain the following proposition.

\begin{proposition}\label{prop:entropy-k-cell}
Let $G$ be a deterministic graphing, with entropy $h(G)$. The cardinality of the $k$-th cell decomposition of $\Space{X}$ w.r.t. $G$, as a function $c(k)$ of $k$, is asymptotically bounded by $g(k)=2^{k.h([G])}$, i.e. $c(k)=O(g(k))$.
\end{proposition}

We also state another bound on the number of cells of the $k$-th cell decomposition, based on the state cover entropy, i.e. the entropy with respect to the state cover rather than the usual entropy which takes the supremum of cover entropies when the cover ranges over all finite covers of the space. This result is a simple consequence of \cref{lem:admseqHk}.

\setcounter{temp}{\value{theorem}}
\setcounterref{theorem}{prop:entropy-k-cell-h0}
\addtocounter{theorem}{-1}
\begin{proposition}
Let $G$ be a deterministic graphing. We consider the \emph{state cover entropy} $h_0([G])=\lim_{n\rightarrow\infty}H_{\Space{X}}^{n}([G],\mathcal{S})$ where $\Space{S}$ is the state cover. The cardinality of the $k$-th cell decomposition of $\Space{X}$ w.r.t. $G$, as a function $c(k)$ of $k$, is asymptotically bounded by $g(k)=2^{k.h_0([G])}$, i.e. $c(k)=O(g(k))$.
\end{proposition}
\setcounter{theorem}{\value{temp}}

\section{First lower bounds}\label{sec:entropylowerbounds}

We will now explain how to obtain lower bounds for algebraic models of 
computation based on the interpretation of programs as graphings and
entropic bounds. These results make use of the Milnor-Oleĭnik-Petrovskii-Thom theorem
which bounds the sum of the Betti numbers of algebraic varieties. In fact,
we will use a version due to Ben-Or of this theorem.

\subsection{Milnor-Oleĭnik-Petrovskii-Thom theorem}

Let us first recall the classic Milnor-Oleĭnik-Petrovskii-Thom theorem. This theorem provides a
bound on the sum of Betti numbers $\beta_i(V)$ of an algebraic variety $V$. 
Recall that the $i$-th Betti number is the dimension of the $i$-th \u{C}ech cohomology group $H_i(V)$; the
number $\beta_0(V)$ then coincides with the number of connected components of the variety. 
In the following $H^\ast (V)$ denotes the sequence of cohomology groups, and $\rank H^{\ast}V$ 
should be understood as standing for the sum of all $\beta_i(V)$.

\begin{theorem}[{\cite[\textrm{Theorem 3}]{Milnor:1964}}]
  If $V \subseteq \realN^m$ is defined by polynomial identities of the form
  \begin{align*}
    f_1 \geqslant 0, \ldots, f_p \geqslant 0
  \end{align*}
  with total degree $d=\degre f_1 + \cdots + \deg f_p$, then
  \begin{align*}
    \rank H^{\ast}V \leqslant \frac{1}{2} (2+d)(1+d)^{m-1}.
  \end{align*}
\end{theorem}

We will use in the proof the following variant,
stated and proved by Ben-Or. 

\begin{theorem}
  \label{thm:milnor-ben-or}
Let $V \subseteq \realN^n$ be a set defined by polynomial 
in\pointmedian{}equations
($n,m,h\in\naturalN$):
  \[
    \left\{
      \begin{array}{l}
        q_1(x_1,\dots,x_n) = 0\\
        \vdots\\
        q_m(x_1,\dots,x_n) = 0\\
        p_1(x_1,\dots,x_n) > 0\\
        \vdots\\
        p_s(x_1,\dots,x_n) > 0\\
        p_{s+1}(x_1,\dots,x_n) \geqslant 0\\
        \vdots\\
        p_h(x_1,\dots,x_n) \geqslant 0\\
      \end{array}
    \right.
  \]
  for $p_i, q_i \in \realN[X_1,\dots,X_n]$ of degree lesser than $d$.

Then $\beta_0(V)$ is at most $d(2d-1)^{n+h-1}$,
where $d=\max\{2, \deg(q_i),\deg(p_j)\}$.
\end{theorem}

\subsection{Algebraic decision trees}
\label{algebraic-decision-trees}

From \cref{prop:entropy-k-cell}, one obtains easily the following technical lemma.

\begin{lemma}
Let $T$ be a $d$-th order algebraic decision tree deciding a subset $W
\subseteq \realN^n$. Then the number of connected components of $W$ is 
bounded by $2^h d(2d-1)^{n+h-1}$, where $h$ is the height of $T$.
\end{lemma}

\begin{proof}
We let $h$ be the height of $T$, and $d$ be the maximal degree of the 
polynomials appearing in $T$. Then the $h$-th cell decomposition of $[T]$ 
defines a family of semi-algebraic sets defined by $h$ polynomials equalities
and inequalities of degree at most $d$. By \cref{thm:milnor-ben-or}, each of 
the cells have at most $d(2d-1)^{n+h-1}$ connected components. Moreover, 
\cref{prop:entropy-k-cell-h0}
states that this family has cardinality bounded by $2^{h.h_0([T])}$; since 
$h_0([T])=1$ because each state has at most one antecedent state, this bound
becomes $2^{h}$. Thus, the $h$-th cell decomposition defines at most $2^{h}$ 
algebraic sets which have at most $d(2d-1)^{n+h-1}$ connected components. 
Since the set $W$ decided by $T$ is obtained as a union of the semi-algebraic 
sets in the $h$-th cell decomposition, it has at most $2^h d(2d-1)^{n+h-1}$
connected components.
\end{proof}

\setcounter{temp}{\value{theorem}}
\setcounterref{theorem}{thm:SteeleYao}
\addtocounter{theorem}{-1}
\begin{corollary}[Steele and Yao \cite{SteeleYao82}]
A $d$-th order algebraic decision tree deciding a subset $W\subseteq \realN^n$ 
with $N$ connected components has height $\Omega(\log N)$.
\end{corollary}
\setcounter{theorem}{\value{temp}}

This result of Steele and Yao adapts in a straightforward manner
to a notion of algebraic computation trees describing the construction of the 
polynomials to be tested by mean of multiplications and additions of the 
coordinates. The authors remarked this result uses techniques quite similar to 
that of Mulmuley's lower bounds for the model of \emph{\prams without bit
operations}. It is also strongly similar to the techniques used by Cucker in 
proving that \complexityclass[\realN]{NC}{}$\subsetneq$ \PtimeReal \cite{Cucker92}.

However, a refinement of Steele and Yao's method was quickly obtained by 
Ben-Or so as to allow for computing divisions and taking square roots in this 
notion of algebraic computation trees. In the next section, we explain Ben-Or techniques
from within the framework of graphings through the introduction of \emph{entropic co-trees}. 
We first explain how the present approach
already captures Mulmuley's proof of lower bounds for \emph{\prams without bit operations}, 
which we will later strengthen using entropic co-trees.

\subsection{Mulmuley's result}

At this point, we are already capable of recovering Mulmuley's proof of lower 
bounds for \emph{\prams without bit operations} \cite{Mulmuley99}. The gist
of the proof is to notice that given a \pram $M$ over integers \emph{not using} the
instructions $\boeuclidivide{\cdot}{\cdot}{\cdot}$ or $\bosqrtn{\cdot}{\cdot}$, one
can define a real-valued \pram $\tilde{M}$ \emph{not using} the instructions
$\bodivide{\cdot}{\cdot}{\cdot}$ or $\bosqrtn{\cdot}{\cdot}$ such that:
\begin{quote}
$M$ accepts an integer-valued point $\vec{x}$ in $k$
steps if and only if $\tilde{M}$ accepts $\vec{x}$ in $k$ steps.
\end{quote}
In particular, the subset $W\subseteq\integerN^n$ decided by $M$ is contained in the
subset $\bar{W}\subseteq\realN^n$ decided by $\bar{M}$, and 
$\complementset{W}\subseteq\complementset{\bar{W}}$ -- where the complement of 
$W$ is taken in $\integerN^n$ and the complement of $\bar{W}$ is taken in $\realN^n$. 
Moreover, the number of steps needed by $\bar{M}$ to decide if $\vec{x}$ belongs to $\bar{W}$ 
is equal to the number of steps needed by $M$ to decide if $\vec{x}$ belongs to $W$.

The proof of this is straightforward, as each instructions available in the integer-valued 
\prams model coincide with the restriction to integer values of an instruction available
in the integer-valued \prams model. We can then prove the following result.

\begin{lemma}
A set $V$ decided by a \emph{\pram without bit operations} $M$ with $p$ processors in 
$t$ steps is described by a set of in\pointmedian{}equations of total degree\footnote{The
sum of the degrees of the polynomials defining $V$.} bounded by $pt.2^{O(pt)}$.
\end{lemma}

\begin{proof}
Now, one supposes that $M$ computes some set $V$ in $t$ steps. Then it is computed 
by $\bar{M}$ in $t$ steps. Applying \cref{prop:entropy-k-cell} on the translation of $\bar{M}$
we obtain that the $t$-th cell decomposition of $\Interpret{\bar{M}}$ contains at most 
$2^{t.h_0([T])}$ cells. Moreover, each cell is described by a system of at most $pt$ 
polynomial in\pointmedian{}equations of degree at most $2$. Thus, the whole set 
decided by $\bar{M}$ is described by at most $pt.2^{t.h_0([T])}$ in\pointmedian{}equations
of degree at most $2$. Lastly, one can use the fact that any non-central element in $\amcrealfull$
act as the identity on the private memory\footnote{The only non-central elements are those 
modifying the shared memory, and those do not modify the private memory.} and \cref{lem:entropycrew} 
to establish that $h_0([T])$ grows linearly w.r.t. the number of processors.
\end{proof}

This lemma mimicks the first part of Mulmuley's proof of lower bounds. The second part of the proof, 
which does not differ from Mulmuley's argument, is detailed and reformulated in Section 
\ref{sec:geometry}. By following Mulmuley's argument, one can obtain the following theorem.

\begin{corollary}
  Let $M$ be a \pram without bit operations, with at most
  $O((\log N)^c)$ processors, where $N$ is the length of the inputs and $c$
  any positive integer.
  
  Then $M$ does not decide \maxflow in $O((\log N)^c)$ steps.
\end{corollary}

We state this theorem without proof, as it is established in the same way as \Cref{thm:Mulmuley}.
Note however that this statement is weaker than Mulmuley's result, as it proves that \maxflow is 
not computed by \prams without bit operations in polylogarithmic time with a \emph{polylogarithmic}
number of processors. This comes from our bound on the entropy of $\amcrealfull$, which is linear 
in the number of processors. But Mulmuley establishes a better bound on the systems of equations
by restricting to inputs described by linear forms from $\realN^d$, and exploiting the Milnor-Thom theorem.
This reduces the dependency on $p$ from $O(2^p)$ to $O(p)$ but has a cost in terms of time complexity: 
the number of cells then grows in $O(2^{t^2})$ instead of $O(2^t)$.

\begin{lemma}\label{mulmuley2entropy}
Let $V$ be a set of inputs defined by linear forms of domain $\realN^d$ and decided by a \emph{\pram 
without bit operations} $M$ with $p$ processors in $t$ steps. Then $V$ is described by a set of 
in\pointmedian{}equations of total degree bounded by $O(p^t 3^{t^2})$.
\end{lemma}

\begin{proof}
This is essentially Mulmuley's Theorem 5.1. Since we assume that the inputs are given by linear forms
of domain $\realN^d$, then the values obtained after $k$ steps of computation (elements of the $n-k$
cell decomposition) are defined by $O(2^k)$ polynomial equations of constant degree. By the Milnor-Thom 
theorem, these have $O(3^t)$ connected components. Now, the cells in the $n-k+1$-cell decomposition 
are obtained by comparing the values obtained after $k$ steps, i.e. they are described as collections of connected
components in the $n-k$-cell decomposition. As a consequence, the number of $n-k+1$-cells is at most $O(p 3^k)$ 
times the number of $k$-cells. By induction, this gives that the number of $n$-cells is at most $O(p^n 3^{n^2})$. 

Moreover, since the $n-k+1$ cells are obtained by making at most $p$ comparisons between the sets obtained in 
the $k$-cell decomposition, they are defined by at most $O(p 3^k)$ times the number of in\pointmedian{}equations 
of degree $2$ needed to define the $n-k$-cells. Combined with the bound on the number of cells at a given depth, 
a simple induction then establishes the result.
\end{proof}

\begin{remark}
This bound, since it does not depend exponentially on $t$ but rather depends exponentially on $t^2$, does 
not translate directly as a bound on the entropy. It rather translates into a bound on a notion of 
\emph{relative $2$-entropy}. Define $L$ the set of inputs defined as a linear form $\realN^d\rightarrow\realN^p$,
and consider an open cover $\mathcal{U}$: the $2$-entropy relative to $L$ is defined as:
\[ h_{L}^{(2)}(f,\mathcal{U})=\lim_{n\rightarrow\infty} \frac{1}{n^2} H_{L}^{0}(\mathcal{U}\vee f^{-1}(\mathcal{U})\vee \dots\vee f^{-k}(\mathcal{U})) .\]
Then the above result states that $h^{(2)}(f)$ is bounded by $O(\log(p))$, a result that can be used directly in the proofs, modulo the straightforward adaptation of the relevant results (for instance, that the number of cells in the $k$-cell decomposition is bounded by $O(2^{t^2 h_L^{(2)}})$).
\end{remark}

Based on this last lemma, we can establish Mulmuley's result. We only provide a quick sketch of Mulmuley's result here, based on result
presented in this later section. Note that the proof of the more general result \Cref{cor:main-pram}
follows the same general pattern.

\setcounter{temp}{\value{theorem}}
\setcounterref{theorem}{thm:Mulmuley}
\addtocounter{theorem}{-1}
\begin{corollary}
  Let $M$ be a \pram without bit operations, with at most
  $2^{O((\log N)^c)}$ processors, where $N$ is the length of the inputs and $c$
  any positive integer.
  
  Then $M$ does not decide \maxflow in $O((\log N)^c)$ steps.
\end{corollary}
\setcounter{theorem}{\value{temp}}

\begin{proof}[Sketch]
The crux is the obtention of \Cref{thm:mulmuley-geometric}
which, combined with \Cref{thm:param}, implies that there exists a polynomial $P$ such that no surface of total degree
$\delta$ can separate the integer points defined by \maxflow as long as $2^{\Omega(N)}>P(\delta)$.
Hence, if we suppose that $M$ is such that the running time $t$ is $O((\log^k (N))^c)$ and the number of processors $p$ is $2^{O((\log N)^c)}$ (for any positive integers $k$ and $c$), the previous
result (\autoref{mulmuley2entropy}) implies that the total degree $\delta$ of the set decided by $\bar{M}$ is at most $O(p2^{t^2})$. But that bound on the total degree $\delta$ implies that 
$2^{\Omega(N)}>P(\delta)$ for any polynomial $P$. This is enough to conclude that $\bar{M}$ does not compute
a set separating the integer points defined by \maxflow, hence $M$ does not decide \maxflow.
\end{proof}

\section{Refining the method}

It is not a surprise then that similar bounds to that of algebraic decisions trees 
can be computed using similar methods in the restricted fragment without 
division and square roots. An improvement on this is the result of Ben-Or
generalising the technique to algebraic computation trees with division and
square root nodes. The principle is quite simple: one simply adds additional
variables to avoid using the square root or division, obtaining in this way a 
system of polynomial equations. For instance, instead of writing the equation
$p/q<0$, one defines a fresh variable $r$ and considers the system 
\[ p=qr; r<0 \]

This method seems different from the direct entropy bound obtained in the 
case of algebraic decision trees. However, we will see how it can be adapted
directly to graphings.

\subsection{Entropic co-trees and $k$-th
  computational forests}

\begin{definition}[$k$-th entropic co-tree]
\label{def:co-tree}
Consider a deterministic graphing representative $T$, and fix an element $\top$ 
of the set of control states. We can define the $k$-th entropic co-tree of $T$ 
along $\top$ and the state cover inductively:
\begin{itemize}
\item $k=0$, the co-tree $\cotree{0}$ is simply the root 
$n^{\epsilon}=\realN^n\times\{\top\}$;
\item $k=1$, one considers the preimage of $n^{\epsilon}$ through $T$, i.e. 
$T^{-1}(\realN^n\times\{\top\})$ the set of all non-empty sets 
$\alpha(m_{e})^{-1}(\realN^n\times\{\top\})$ and intersects it pairwise with the 
state cover, leading to a finite family (of cardinality bounded by the number of 
states multiplied by the number of edges fo $T$) $(n_{e}^{i})_{i}$ 
defined as $n^{i}=T^{-1}(n^{\epsilon})\cap \realN^n\times\{i\}$. The first entropic
co-tree  $\cotree{1}$ of $T$ is then the tree defined by linking each $n_{e}^{i}$ 
to $n^{\epsilon}$ with an edge labelled by $m_{e}$;
\item $k+1$, suppose defined the $k$-th entropic co-tree of $T$, defined as a 
family of elements $n_{\seq{e}}^{\pi}$ where $\pi$ is a finite sequence of states of 
length at most $k$ and $\seq{e}$ a sequence of edges of $T$ of the same
length, and where $n_{\seq{e}}^{\pi}$ and $n_{\seq{e'}}^{\pi'}$ are linked by an 
edge labelled $f$ if and only if $\pi'=\pi.s$ and $\seq{e'}=f.\seq{e}$ where $s$ is 
a state and $f$ an edge of $T$. We consider the subset of elements 
$n_{\seq{e'}}^{\pi}$ where $\pi$ is exactly of length $k$, and for each such 
element we define new nodes $n_{e.\seq{e'}}^{\pi.s}$ defined as 
$\alpha(m_{e})^{-1}(n_{\seq{e'}}^{\pi})\cap \realN^n\times\{s\}$ when it is non-empty.
The $k+1$-th entropic co-tree $\cotree{k+1}$ is defined by extending the $k$-th 
entropic co-tree  $\cotree{k}$, adding the nodes $n_{e.\seq{e'}}^{\pi.s}$ and linking 
them to $n_{\seq{e'}}^{\pi}$ with an edge labelled by $e$.
\end{itemize}
\end{definition}

\begin{remark}
  The co-tree can alternatively be defined non-inductively in the following way:
  the $n_{\seq{e}}^{\pi}$ for $\pi$ is a finite sequence of states and $\seq{e}$
  a sequence of edges of $T$ of the same length by
  $n_{\epsilon}^{\epsilon} = \realN^n\times\{\top\}$ and
  \begin{align*}
    n_{\seq{e}.e}^{\pi.s} = \left[\alpha(m_{e})^{-1}(n_{\seq{e}}^{\pi})\right] \cap \left[
    \realN^n \times \{s\}\right]
  \end{align*}
  The $k$-th entropic co-tree of $T$ along $\top$ has as vertices the non-empty
  sets $n_{\seq{e}}^{\pi}$ for $\pi$ and $\seq{e}$ of length at most $k$ and as
  only edges, links $n_{\seq{e}.e}^{\pi.s} \to n_{\seq{e}}^{\pi}$ labelled by
  $m_e$.
\end{remark}

This definition formalises a notion that appears more or less clearly in the work of 
Lipton and Steele, and of Ben-Or, as well as in the proof by Mulmuley. The nodes 
for paths of length $k$ in the $k$-th co-tree corresponds to the $k$-th cell 
decomposition, and the corresponding path defines the polynomials describing the
semi-algebraic set decided by a computational tree. The co-tree can be used to
reconstruct the algebraic computation tree $T$ from the graphing representative 
$[T]$, or constructs \emph{some} algebraic computation tree (actually a forest) that 
approximates the  computation of the graphing $F$ under study when the latter is 
not equal to $[T]$ for some tree $T$.

\begin{definition}[$k$-th computational forest]
Consider a deterministic graphing $T$, and fix an element $\top$ of the set of
control states. We define the $k$-th computational forest of $T$ along $\top$ and
the state cover as follows. Let $\cotree{k}$ be the $k$-th entropic co-tree of $T$. 
The $k$-th computational forest of $T$ is defined by regrouping all elements 
$n_{e.\vec{e}'}^{\pi}$ of length $m$: if the set $N_{e}^{m}=\{n_{e.\vec{e}'}^{\pi}\in
\cotree{k} ~|~ \len(\pi)=m\}$ is non-empty it defines a new node $N_{e}^{m}$. 
Then one writes down an edge from $N_{e}^{m}$ to $N_{e'}^{m-1}$, labelled by
$e$, if and only if there exists $n_{e.e'.\vec{f}}^{s.\pi}\in N_{e}^{m}$ such that 
$n_{e'.\vec{f}}^{\pi}\in N_{e'}^{m-1}$.
\end{definition}

One checks easily that the $k$-th computational forest is indeed a forest: an edge
can exist between $N_{e}^{m}$ and $N_{f}^{n}$ only when $n=m+1$, a property 
that forbids cycles. The following proposition shows how the $k$-th computational
forest is linked to computational trees.

\begin{proposition}
If $T$ is a computational tree of depth $k$, the $k$-th computational forest of $[T]$
is a tree which defines straightforwardly a graphing (treeing) representative of $T$.
\end{proposition}

We now state and prove an easy bound 
on the size of the entropic co-trees.

\begin{proposition}[Size of the entropic co-trees]
Let $T$ be a graphing representative, $E$ its set of edges, and $\seqedges{k}{E}$ 
the set of paths of length $k$ in $T$. The 
number of nodes of its $k$-th entropic co-tree $\cotree{k}$, as a function $n(k)$ 
of $k$, is asymptotically bounded by $\card{\seqedges{k}{E}}.2^{(k+1).h([G])}$, 
itself bounded by $2^{\card{E}}.2^{(k+1).h([G])}$.  
\end{proposition}

\begin{proof}
For a fixed sequence $\vec{e}$, the number of elements $n_{\vec{e}}^{\pi}$ 
of length $m$ in $\cotree{k}$ is bounded by the number of elements in the $m$-th 
cell decomposition of $T$, and is therefore bounded by $g(m)=2^{m.h([T])}$ 
by \cref{prop:entropy-k-cell}. 
The number of sequences $\vec{e}$ is bounded by $\card{\seqedges{k}{E}}$ and therefore 
the size of $\cotree{k}$ is thus bounded by $\card{\seqedges{k}{E}}.2^{(k+1).h([T])}$.
\end{proof} 

From the proof, one sees that the following variant of \cref{prop:entropy-k-cell-h0} 
holds.
\begin{proposition}\label{prop:entropy-k-cell-h0-edges}
Let $G$ be a deterministic graphing with a finite set of edges $E$, and 
$\seqedges{k}{E}$ the set of paths of length $k$ in $G$. 
We consider the \emph{state cover entropy} 
$h_0([G])=\lim_{n\rightarrow\infty}H_{\Space{X}}^{n}([G],\mathcal{S})$ where 
$\Space{S}$ is the state cover. The cardinality of the length $k$ nodes of the 
entropic co-tree of $G$, as a function $c(k)$ of $k$, is asymptotically bounded 
by $g(k)=\card{\seqedges{k}{E}}.2^{k.h_0([G])}$, which is itself bounded by
$2^\card{E}.2^{k.h_0([G])}$.
\end{proposition}

%



\subsection{The technical lemma}

This definition formalises a notion that appears more or less clearly in the work of 
Steele and Yao, and of Ben-Or, as well as in the proof by Mulmuley. It will be key in 
establishing the main technical lemma, namely \cref{mainlemma}.

The vertices 
for paths of length $k$ in the $k$-th co-tree corresponds to the $k$-th cell 
decomposition, and the corresponding path defines the polynomials describing the
semi-algebraic set decided by a computational tree. While in Steele and Yao and 
Mulmuley's proofs, one obtain directly a polynomial for each cell, we here need to
construct a system of equations for each branch of the co-tree. 
%

%
%
%

Given a $\crew^{p}(\amcrealfull)$-graphing representative $G$ we will write $\rootdegree{G}$ 
the maximal value of $n$ for which an instruction $\bosqrtn{i}{j}$ appears in the realiser of an 
edge of $G$. 

The proof of this theorem is long but simple to understand as it follows Ben-Or's
method. We define, for each
vertex of the $k$-th entropic co-tree, a system of algebraic equations (each of 
degree at most 2). The system is defined by induction on $k$, and uses the
information of the specific instruction used to extend the sequence indexing 
the vertex at each step. For instance, the case of division follows Ben-Or's 
method, introducing a fresh variable and writing down two equations. As 
mentioned in \cref{footnote:indirectref},
the input variables are split into numerical and non-numerical inputs, and one
assumes that indirect references do not depend on non-numerical inputs. This 
implies that all indirect references have a fixed value determined by the 
non-numerical input; hence in the analysis below -- which focuses on numerical
inputs -- indirect references correspond to references to a fixed value register.

\begin{lemma}\label{thm:graphingsBenOrsystems}
Let $G$ be a computational graphing representative with edges realised only by 
generators of the \amc $\crew^{p}(\amcrealfull)$, and $\seqedges{k}{E}$ the 
set of paths of length $k$ in $G$. Suppose $G$ computes the 
membership problem for $W \subseteq 
\realN^n$ in $k$ steps, i.e. for each element of $\realN^n$, 
$\pi_{\Space{S}}(G^{k}(x))=\top$ if and only if $x\in W$. Then $W$ is a semi-algebraic 
set defined by at most $\card{\seqedges{k}{E}}.2^{k.h_0([G])}$ systems of $pk$ 
equations of degree at most $\max(2,\rootdegree{G})$ and involving at most $p(k+n)$ 
variables and $p(k+n)$ inequalities.
\end{lemma}

\begin{proof}
If $G$ computes the membership problem for $W$ in $k$ steps, it means $W$
can be described as the union of the subspaces corresponding to the nodes 
$n^{\pi}_{\seq{e}}$ with $\pi$ of length $k$ in $\cotree{k}$. Now, each such 
subspace is an algebraic set, as it can be described by a set of polynomials as 
follows.

Finally let us note that, as in Mulmuley's work \cite{Mulmuley99}, since in our 
model the memory pointers are allowed to depend only on the nonnumeric 
parameters, indirect memory instructions can be treated as standard -- direct --
memory instructions. In other words, whenever an instruction involving a memory 
pointer is encountered during the course of execution, the value of the pointer 
is completely determined by nonnumerical data, and the index of the involved 
registers is completely determined, independently of the numerical inputs.


We define a system of equations $(E^{\seq{e}}_i)_{i}$ for each node 
$n^{\pi}_{\seq{e}}$ of the entropic co-tree $\cotree{k}$. We explicit the construction
for the case $p=1$, i.e. for the \amc $\crew^1(\amcrealfull)=\amcrealfull$; the case 
for arbitrary $p$ is then dealt with by following the construction and introducing
$p$ equations at each step (one for each of the $p$ instructions in $\amcrealfull$
corresponding to an element of $\crew^p(\amcrealfull)$). This is done 
inductively on the size of the path $\vec{e}$, keeping track of the last modifications
of each register. I.e. we define both the system of equations $(E^{\seq{e}}_i)_{i}$
and a function\footnote{The use of $\bot$ is to allow for the creation of
fresh variables not related to a register.} 
$\history{\seq{e}}: \realN^{\omega}\cup\{\bot\}\rightarrow \omega$ (which 
is almost everywhere null). This function increases each time a register is modified, 
and will be used to create a new variable corresponding to the value of the register
\emph{at this precise moment in the computation}. The
additional value $\bot$ will be used to create new variables not related to a specific
register (used in the case of comparisons below).

 For an empty sequence, the system of
equations is empty, and the function $\history{\epsilon}$ is constant, equal to $0$. 
The system of equation, as well as the function $\history{\epsilon}$, are then jointly 
defined inductively as follows.
Suppose that $\vec{e'} = (e_1, \dots, e_m, e_{m+1})$, with $\vec{e}=(e_1, 
\dots, e_m)$, and that one already computed $(E^{\seq{e}}_i)_{i\geqslant m}$ 
and the function $\history{\seq{e}}$. We now consider the edge $e_{m+1}$ and 
let $(r,r')$ be its realizer. We extend the system of equations 
$(E^{\seq{e}}_i)_{i\geqslant m}$ by a new equation $E_{m+1}$ and define the 
function $\history{\seq{e'}}$ as follows: 
  \begin{itemize}
  \item if $r = \boadd{i}{j}{k}$, $\history{\seq{e'}}(u)=\history{\seq{e}}(u)+1$ if 
  	$u=i$, and $\history{\seq{e'}}(u)=\history{\seq{e}}(u)$ otherwise; then
	$E_{m+1}$ is $\benorvar{i} = \benorvar{j}+\benorvar{k}$;
  \item if $r = \bosubstract{i}{j}{k}$, $\history{\seq{e'}}(u)=\history{\seq{e}}(u)+1$ if 
  	$u=i$, and $\history{\seq{e'}}(u)=\history{\seq{e}}(u)$ otherwise; then
	$E_{m+1}$ is $\benorvar{i} = \benorvar{j} - \benorvar{k}$;
  \item if $r = \bomultiply{i}{j}{k}$, $\history{\seq{e'}}(u)=\history{\seq{e}}(u)+1$ if 
  	$u=i$, and $\history{\seq{e'}}(u)=\history{\seq{e}}(u)$ otherwise; then
	$E_{m+1}$ is $\benorvar{i} = \benorvar{j}\times \benorvar{k}$;
  \item if $r = \bodivide{i}{j}{k}$, $\history{\seq{e'}}(u)=\history{\seq{e}}(u)+1$ if 
  	$u=i$, and $\history{\seq{e'}}(u)=\history{\seq{e}}(u)$ otherwise; then
	$E_{m+1}$ is $\benorvar{i}\times \benorvar{k} = \benorvar{j}$;
  \item if $r = \boaddconst{i}{k}{c}$, $\history{\seq{e'}}(u)=\history{\seq{e}}(u)+1$ if 
  	$u=i$, and $\history{\seq{e'}}(u)=\history{\seq{e}}(u)$ otherwise; then
	$E_{m+1}$ is $\benorvar{i} = c + \benorvar{k}$;
  \item if $r = \bosubstractconst{i}{k}{c}$, $\history{\seq{e'}}(x)=\history{\seq{e}}(x)+1$ if 
  	$x=i$, and $\history{\seq{e'}}(u)=\history{\seq{e}}(u)$ otherwise; then
	$E_{m+1}$ is $\benorvar{i} = c - \benorvar{k}$;
  \item if $r = \bomultiplyconst{i}{k}{c}$, $\history{\seq{e'}}(u)=\history{\seq{e}}(u)+1$ if 
  	$u=i$, and $\history{\seq{e'}}(u)=\history{\seq{e}}(u)$ otherwise; then
	$E_{m+1}$ is $\benorvar{i} = c\times \benorvar{k}$;
  \item if $r = \bodivideconst{i}{k}{c}$, $\history{\seq{e'}}(u)=\history{\seq{e}}(u)+1$ if 
  	$u=i$, and $\history{\seq{e'}}(u)=\history{\seq{e}}(u)$ otherwise; then
	$E_{m+1}$ is $\benorvar{i}\times c = \benorvar{k}$;
  \item if $r = \bosqrtn{i}{k}$, $\history{\seq{e'}}(u)=\history{\seq{e}}(u)+1$ if 
  	$u=i$, and $\history{\seq{e'}}(u)=\history{\seq{e}}(u)$ otherwise; then
	$E_{m+1}$ is $(\benorvar{i})^n = \benorvar{k}$;
  \item if $r = \copyy{i}{j}$,  $\history{\seq{e'}}(u)=\history{\seq{e}}(u)+1$ if 
  	$u=i$, and $\history{\seq{e'}}(u)=\history{\seq{e}}(u)$ otherwise; then
	$E_{m+1}$ is $\benorvar{i} = \benorvar{j}$;
\item if $r = \refcopy{i}{j}$, then the value of $\sharp{}i$ does not depend on
	the numerical inputs and corresponds to a fixed value $a\in\realN$; we then
	define $\history{\seq{e'}}(u)=\history{\seq{e}}(u)+1$ if 
  	$u=a$, and $\history{\seq{e'}}(u)=\history{\seq{e}}(u)$ otherwise; then
	$E_{m+1}$ is $\benorvar{a} = \benorvar{j}$;
\item if $r = \copyref{i}{j}$, then the value of $\sharp{}j$ does not depend on
	the numerical inputs and corresponds to a fixed value $a\in\realN$; we then
	define $\history{\seq{e'}}(u)=\history{\seq{e}}(u)+1$ if 
  	$u=i$, and $\history{\seq{e'}}(u)=\history{\seq{e}}(u)$ otherwise; then
	$E_{m+1}$ is $\benorvar{i} = \benorvar{a}$;
  \item if $r = \identity$, the source of the edge $e_q$ is of the form 
    $\{ (x_1,\dots,x_{n+\ell}) \in \realN^{n+\ell} \mid P(x_k)\}\times \{i\}$ where
    $P$ compares the variable $x_k$ with $0$:
    \begin{itemize}
    \item if $P(x_k)$ is $x_k \neq 0$,  $\history{\seq{e'}}(u)=\history{\seq{e}}(u)+1$ if 
  	$u=\bot$, and $\history{\seq{e'}}(u)=\history{\seq{e}}(u)$ otherwise then 
	$E_{m+1}$ is $\benorvar{\bot}\benorvar{k} -1 = 0$;
    \item otherwise (e.g. $P(x_k)$ is the inequality $x_k\leq 0$) we set 
    $\history{\seq{e'}}=\history{\seq{e}}$ and $E_{m+1}$ 
    is defined as $P(\benorvar{k})$.
  \end{itemize}
  \end{itemize}
  
  We now consider the system of equations $(E_{i})_{i=1}^{k}$ defined from
  the path $\seq{e}$ of length $k$ corresponding to a node $n^{\pi}_{\seq{e}}$
  of the $k$-th entropic co-tree of $G$. This system consists in $k$ equations 
  of degree at most $\max(2,\rootdegree{G})$ and containing at most $k+n$ variables, 
  counting the variables $x_1^0,\dots,x_n^0$ corresponding to the initial registers, 
  and adding at most $k$ additional variables since an edge of $\vec{e}$ introduces 
  at most one fresh variable. Among these equations, at most $k$ are inequalities, 
  since each edge introduces at most one inequation.
  Since the number of vertices $n^{\pi}_{\seq{e}}$ is bounded by 
  $\card{\seqedges{k}{E}}.2^{k.h_0([G])}$ by
  \cref{prop:entropy-k-cell-h0-edges}, 
  we obtained the stated result in the case $p=1$.
  
  The case for arbitrary $p$ is then deduced by noticing that each step in the induction
  would introduce at most $p$ new equations and $p$ new variables. The resulting 
  system thus contains at most $pk$ equations of degree at most $\max(2,\rootdegree{G})$ 
  and containing at most $p(k+n)$ variables.\end{proof}

This theorem extends to the case of general computational graphings by 
considering the \emph{algebraic degree} of the graphing. 

\begin{definition}[Algebraic degree]\label{def:algebraicdegree}
Let $\AMC{G}{R}{\alpha}$ be an \amc. The algebraic degree of an element of 
$\MonGaR{G}{R}$ is the minimal number of generators needed to express it. 
The algebraic degree of an $\alpha$-graphing is the maximum of the algebraic 
degrees of the realisers of its edges. 
\end{definition}

If an edge is realised by an element $m$ of 
algebraic degree $D$, then the method above applies by introducing the $D$ 
new equations corresponding to the $D$ generators used to define $m$. The 
general result then follows.

\setcounter{temp}{\value{theorem}}
\setcounterref{theorem}{mainlemma}
\addtocounter{theorem}{-1}
\begin{lemma}
Let $G$ be a $\crew^{p}(\amcrealfull)$-computational graphing representative, 
$\seqedges{k}{E}$ the set of paths of length $k$ in $G$, and $D$ 
its algebraic degree. 
Suppose $G$ computes the membership 
problem for $W \subseteq \realN^n$ in $k$ steps, i.e. for each element of 
$\realN^n$, $\pi_{\Space{S}}(G^{k}(x))=\top$ if and only if $x\in W$.
 Then $W$ is a semi-algebraic 
set defined by at most $\card{\seqedges{k}{E}}.2^{k.h_0([G])}$ systems of $pkD$ 
equations of degree at most $\max(2,\rootdegree{G})$ and involving at most 
$pD(k+n)$ variables.
\end{lemma}
\setcounter{theorem}{\value{temp}}

Lastly, we specify this result to inputs described by linear forms. More precisely,
we show in this case a bound on the 2-entropy which is logarithmic in the number
of processors. This follows the argument from \autoref{mulmuley2entropy} and 
will be used to establish the generalisation of Mulmuley's result.

\begin{lemma}\label{entropiccotrees2entropy}
Let $V$ be a set of inputs defined by linear forms of domain $\realN^d$ and decided by a 
$\crew^{p}(\amcrealfull)$-computational graphing representative $G$ in $t$ steps. Let 
$\seqedges{k}{E}$ be the set of paths of length $k$ in $G$, and $D$ be the algebraic 
degree of $G$. 

Then $V$ is described by a set of 
in\pointmedian{}equations of total degree bounded by $O(p^t 3^{t^3})$.
\end{lemma}

\begin{proof}
The argument follows the one of \autoref{mulmuley2entropy}, when we focus on one branch in the
entropic co-tree. Since the inputs are given by linear forms of domain $\realN^d$, the values 
obtained after $n-k$ steps of computation are all defined by $O(2^k)$ 
polynomial equations of constant degree. These have at most $O(3^k)$ connected components by the 
Milnor-Thom theorem. The cells in the $n-k+1$-cell decomposition are either obtained by applying an algebraic operation,
or comparing the values obtained after $k$ steps, in which case they are described as collections of connected
components in the $k$-cell decomposition. This second case is the one that creates most new cells. As a consequence, 
the number of $n-k+1$-cells is at most $O(p 3^k)$ times the number of $k$-cells. 
By induction, the number of $t$-cells for a given branch is at most $O(p^t 3^{t^2})$, which implies that 
the total number of $t$-cells is $O(p^t 3^{t^3})$ since there are $O(3^t)$ elements in $\seqedges{k}{E}$.

Finally, since the $n-k+1$ cells are obtained by making at most $p$ comparisons between the sets obtained in 
the $k$-cell decomposition, they are defined by at most $O(p 3^k)$ times the number of in\pointmedian{}equations 
of degree $2$ needed to define the $k$-cells. This proves the result by a simple induction.
\end{proof}

\section{Recovering Ben Or and Cucker's theorems}

\subsection{Ben-Or}

We now recover Ben-Or result by obtaining a bound on the number of
connected components of the subsets $W \subseteq \realN^n$ whose
membership problem is computed by a graphing in less than a given number
of iterations. This theorem is obtained by applying the Milnor-Oleĭnik-Petrovskii-Thom theorem
on the obtained systems of equations to bound the number of connected
components of each cell. Notice that in this case $p=1$ and $\rootdegree{G}=2$ 
since the model of algebraic
computation trees use only square roots. A mode general result 
holds for algebraic computation trees extended with arbitrary roots, but we
here limit ourselves here to the original model.

\begin{theorem}\label{thm:graphingsBenOr}
Let $G$ be a computational $\amcrealfull$-graphing representative translating
an algebraic computational tree, $\seqedges{k}{E}$ the 
set of length $k$ sequences of edges in $G$. 
Suppose $G$ computes the membership 
problem for $W \subseteq \realN^n$ in $k$ steps. Then $W$ 
has at most $\card{\seqedges{k}{E}}.2^{k.h_0([G])+1}3^{2k+n-1}$ connected 
components.
\end{theorem}

\begin{proof}
By \Cref{thm:graphingsBenOrsystems} (using the fact that $p=1$ and 
$\rootdegree{G}=2$), the problem $W$ decided by $G$ in $k$ steps
is described by at most $\card{\seqedges{k}{E}}.2^{k.h_0([G])}$ systems of $k$ 
equations of degree $2$ involving at most $k+n$ variables and at most $k$ 
inequalities. Applying 
\Cref{thm:milnor-ben-or}, we deduce that each such system of 
in\pointmedian{}equations (of $k$ equations of degree $2$ in $\realN^{k+n}$) 
describes a semi-algebraic variety $S$ such that
$\beta_0(S)<2.3^{(n+k)+k-1}$. This being true for each of the 
$\card{\seqedges{k}{E}}.2^{k.h_0([G])}$ cells, we have that 
$\beta_0(W)<\card{\seqedges{k}{E}}.2^{k.h_0([G])+1}3^{2k+n-1}$.
\end{proof}

Since a subset computed by a tree $T$ of depth $k$ is computed by
$\Interpret{T}$ in $k$ steps by \Cref{thm:quantsoundact}, we get as a
corollary the original theorem by Ben-Or 
relating the number of connected components of a set $W$ and the depth of 
the algebraic computational trees that compute the membership problem for $W$.

\setcounter{temp}{\value{theorem}}
\setcounterref{theorem}{ben-or}
\addtocounter{theorem}{-1}
\begin{corollary}[{\cite[Theorem 5]{Ben-Or83}}]
  Let $W \subseteq \realN^n$ be any set, and let $N$ be the maximum of the
  number of connected components of $W$ and $\realN^n \setminus W$.
  An algebraic computation tree computing the membership problem for $W$ has
  height $\Omega(\log N)$.
\end{corollary}
\setcounter{theorem}{\value{temp}}

\begin{proof}
Let $T$ be an algebraic computation tree computing the membership problem 
for $W$, and consider the computational treeing $[T]$. Let $d$ be the height of
$T$; by definition of $[T]$ the membership problem for $W$ is computed in
exactly $d$ steps. Thus, by the previous theorem, $W$ has at most 
$\card{\seqedges{k}{E}}.2^{d.h_0([T])+1}3^{2d+n-1}$ connected components. 
As the interpretation of an algebraic computational tree, $h_0([T])$ is at most 
equal to $2$, and $\card{\seqedges{k}{E}}$ is bounded by $2^{d}$. Hence
$N\leqslant 2^d.2^{2d+1}3^{n-1}3^{2d}$, i.e. $d=\Omega(\log N)$.
\end{proof}

We immediately deduce an application that will be useful to us in the
remainder. Let $m \in \naturalN$ and $0 < x < 2^m$. Let $k \in \naturalN$ be
such that $1 \leqslant k \leqslant m$. We call
$\left\lfloor\frac{x}{2^{k-1}}\right\rfloor - 2\left\lfloor \frac{x}{2^k}
\right\rfloor$ the \emph{$k$-th bit} of $x$.

\begin{lemma}
  An algebraic computation tree computing the $k$-th bit of $x$ has height
  $\Omega(m-k)$.
\end{lemma}

\begin{proof}
  Let
  \begin{align*}
    W = \left\{ x \in \realN \mid \left\lfloor\frac{x}{2^{k-1}}\right\rfloor -
    2\left\lfloor \frac{x}{2^k} \right\rfloor = 1 \right\}
  \end{align*}
  $W$ is the disjoint union of $2^{m-k+1}$ intervals, and so is its complement in
  $]0;2^m[$. So, by Theorem \ref{ben-or}, any algebraic computation tree
  computing the $k$-th bit has height $\Omega(m-k)$.
\end{proof}

We will see later that bit-extraction is also difficult for the \pram model (cf. 
Prop.~\ref{prop:bit-in-pram}). This is an
essential difference between the booleans and algebraic models.

\begin{remark}
In the case of algebraic \prams discussed in the next sections, the $k$-th entropic
co-tree $\cotree{k}[M]$ of a machine $M$ defines an algebraic computation tree 
which follows the $k$-th first steps of computation of $M$. I.e. the algebraic 
computation tree $\cotree{k}[M]$ approximate the computation of $M$ in such a way 
that $M$ and $\cotree{k}[M]$ behave in the exact same manner in the first $k$ steps.
\end{remark}

\subsection{Cucker's theorem}


Cucker's proof considers the problem defined as the following algebraic set.

\begin{definition}
Define $\cuckersproblem$ to be the set:
\[\{x\in \realN^\omega \mid \abs{x}=n \Rightarrow x_1^{2^n}+x_2^{2^n}=1 \},\]
where $\abs{x}=\max\{n\in\omega\mid x_n\neq 0\}$.
\end{definition}

It can be shown to lie within $\PtimeReal$, i.e. it is decided by a 
real Turing machine \cite{Blum:1989} -- i.e. working with real numbers and real 
operations --, running in polynomial time.

\begin{theorem}[Cucker (\cite{Cucker92}, Proposition 3)]
The problem $\cuckersproblem$ belongs to $\PtimeReal$.
\end{theorem}

We now prove that $\cuckersproblem$ is not computable by an algebraic 
circuit of polylogarithmic depth. The proof follows Cucker's argument, but 
uses the lemma proved in the previous section.

\setcounter{temp}{\value{theorem}}
\setcounterref{theorem}{thm:cucker}
\addtocounter{theorem}{-1}
\begin{corollary}[Cucker (\cite{Cucker92}, Theorem 3.2)]
No algebraic circuit of depth $k=\log^i n$ and size\footnote{We notice here that 
we do not assume any bounds on the number of processors.} $kp$ compute 
$\cuckersproblem$.
\end{corollary}
\setcounter{theorem}{\value{temp}}


\begin{proof}
For this, we will use the lower bounds result 
obtained in the previous section. Indeed, by \Cref{thm:quantsoundalgcirc} and 
\Cref{mainlemma}, any problem decided 
by an algebraic circuit of depth $k$ is a semi-algebraic set defined by at most
$\card{\seqedges{k}{E}}.2^{k.h_0([G])}$ systems of $k$ 
equations of degree at most $\max(2,\rootdegree{G})=2$ (since only square roots
are allowed in the model) and involving at most 
$k+n$ variables. But the curve $\mathfrak{F}_{2^n}^{\realN}$ defined as 
$\{x_1^{2^n}+x_2^{2^n}-1=0\mid x_1,x_2\in\realN\}$ is infinite. As a consequence,
one of the systems of equation must describe a set containing an infinite number 
of points of $\mathfrak{F}_{2^n}^{\realN}$.

This set $S$ is characterized, up to some transformations on the set of 
equations obtained from the entropic co-tree, by a finite system of inequalities of the form
\[ \bigwedge_{i=1}^{s} F_i(X_1,X_2)= 0\wedge \bigwedge_{j=1}^{t} G_j(X_1, X_2) < 0, \]
where $t$ is bounded by $kp$ and the degree of the polynomials 
$F_i$ and $G_i$ are bounded by $2^k$. Moreover, since $\mathfrak{F}_{2^n}^{\realN}$ 
is a curve and no points in $S$ must lie outside of it, we must have $s>0$.

Finally, the polynomials $F_i$ vanish on that infinite subset of the curve and thus in a 
1-dimensional component of the curve. Since the curve is an irreducible one, this 
implies that every $F_i$ must vanish on the whole curve. Using the fact that the ideal 
$(X_1^{2^n} + X_2^{2^n}- 1)$ is prime (and thus radical), we conclude that all the $F_i$ 
are multiples of $X_1^{2^n} + X_2^{2^n}- 1$ which is impossible if their degree is 
bounded by $2^{\log^i n}$ as it is strictly smaller than $2^n$.
\end{proof}

\section{Algebraic surfaces for an optimization problem}\label{sec:geometry}


\subsection{Geometric Interpretation of Optimization Problems}
\label{subsec:optprob}

We start by showing how decision problems of a particular form induce a binary
partition of the space $\integerN^{d}$: the points that are accepted and those
that are rejected. Intuitively, the machine decides the problem if the partition
it induces refines the one of the problem.

We will consider problems of a very specific form: decisions problems in
$\integerN^3$ associated to optimization problems. Let $\OptProb$ be an
optimization problem on $\realN^d$. Solving $\OptProb$ on an instance $t$
amounts to optimizing a function $f_t(\cdot)$ over a space of parameters. We
note $\MaxOptProb(t)$ this optimal value. An affine function
$\Parametrization : [p;q] \to \realN^d$ is called a \emph{parametrization} of
$\OptProb$. Such a parametrization defines naturally a decision problem
$\DecProb$: for all $(x,y,z) \in \integerN^3$, $(x,y,z) \in \DecProb$ iff
$ z >0$, $x/z \in [p;q]$ and $y/z \leq \MaxOptProb\circ \Parametrization(x/z)$.

In order to study the geometry of $\DecProb$ in a way that makes its connection
with $\OptProb$ clear, we consider the ambient space to be $\realN^3$, and we
define the \emph{ray} $[p]$ of a point $p$ as the half-line starting at the
origin and containing $p$. The projection $\projectionAz{p}$ of a point $p$ on a
plane is the intersection of $[p]$ and the affine plane $\AffinePlane$ of
equation $z=1$. For any point $p \in \AffinePlane$, and all $ p_1 \in [p]$,
$\projectionAz{p_1} = p$. It is clear that for
$(p,p',q) \in \integerN^2\times \naturalN^+$,
$\projectionAz{(p,p',q)} = (p/q,p'/q,1)$.

The \emph{cone} $[C]$ of a curve $C$ is the set of rays of points of the
curve. The projection $\projectionAz{C}$ of a surface or a curve $C$ is the set of
projections of points in $C$. We note $\Frontier$ the frontier set
\begin{align*}
  \Frontier = \{(x,y,1) \in \realN^3 \mid y = \MaxOptProb\circ \Parametrization(x) \}.
\end{align*}
and we remark that
\begin{align*}
  [\Frontier] &= \{(x,y,z) \in \realN^2 \times \realN^+ \mid y/z =
                \MaxOptProb \circ\Parametrization (x/z)\}.
\end{align*}

Finally, a machine $M$ decides the problem $\DecProb$ if the sub-partition of 
accepting cells in $\integerN^3$ induced by the
machine is finer than the one defined by the problem's frontier $[\Frontier]$ (which is
defined by the equation $y/z \leq \MaxOptProb \circ \Parametrization(x/z)$).

\subsection{Parametric Complexity}
\label{subsec:param}

We now further restrict the class of problems we are interested in: we will only
consider $\OptProb$ such that $\Frontier$ is simple enough. Precisely:
\begin{definition}
  We say that $\Parametrization$ is an \emph{affine parametrization} of
  $\OptProb$ if $\Parametrization;\MaxOptProb$ is
  \begin{itemize}[nolistsep,noitemsep]
  \item convex
  \item piecewise linear, with breakpoints $\lambda_1 < \cdots < \lambda_{\rho}$
  \item such that the $(\lambda_i)_i$ and the
    $(\MaxOptProb \circ \Parametrization(\lambda_i))_i$ are all rational.
  \end{itemize}

  The \emph{(parametric) complexity} $\rho(\Parametrization)$ is defined as the
  number of breakpoints of $\Parametrization;\MaxOptProb$.
\end{definition}

An optimization problem that admits an affine parametrization of complexity
$\rho$ is thus represented by a surface $[\Frontier]$ that is quite simple: the
cone of the graph of a piecewise affine function, constituted of $\rho$
segments. We say that such a surface is a \emph{$\rho$-fan}. This restriction
seems quite serious when viewed geometrically. Nonetheless, many optimization
problems admit such a parametrization. Before giving examples, we introduce
another measure of the complexity of a parametrization.
\begin{definition}
  Let $\OptProb$ be an optimization problem and $\Parametrization$ be an affine
  parametrization of it. The \emph{bitsize} of the parametrization is the
  maximum of the bitsizes of the numerators and denominators of the coordinates
  of the breakpoints of $\Parametrization;\MaxOptProb$.
  
  In the same way, we say that a $\rho$-fan is of \emph{bitsize} $\beta$ if all
  its breakpoints are rational and the bitsize of their coordinates is lesser
  thant $\beta$.
\end{definition}

\begin{theorem}[Murty \cite{murty1980computational}, Carstensen
  \cite{Carstensen:1983}]
  \label{thm:param}~\newline
  \begin{enumerate}[nolistsep,noitemsep]
  \item there exists an affine parametrization of bitsize $O(n)$ and complexity
    $2^{\Omega(n)}$ of combinatorial linear programming, where $n$ is the total
    number of variables and constraints of the problem.
  \item \label{maxflow-param} there exists an affine parametrization of bitsize
    $O(n^2)$ and complexity $2^{\Omega(n)}$ of the \maxflow problem for directed
    and undirected networks, where $n$ is the number of nodes in the network.
  \end{enumerate}
\end{theorem}
We refer the reader to Mulmuley's paper \cite[Thm. 3.1.3]{Mulmuley99} for proofs, 
discussions and references.

\subsection{Algebraic Surfaces}
\label{sec:alg-surfaces}

An algebraic surface in $\realN^3$ is a surface defined by an equation of the
form $p(x,y,z)=0$ where $p$ is a polynomial. If $S$ is a set of surfaces, each
defined by a polynomial, the \emph{total degree} of $S$ is defined as the sum of
the degrees of polynomials defining the surfaces in $S$.

Let $K$ be a compact of $\realN^3$ delimited by algebraic surfaces and $S$ be a
finite set of algebraic surfaces, of total degree $\delta$. We can assume that
$K$ is actually delimited by two affine planes of equation $z=\mu$ and
$z=2\mu_z$ and the cone of a rectangle
$\{ (x,y,1) \mid |x|, |y| \leqslant \mu_{x,y}\}$, by taking any such compact
containing $K$ and adding the surfaces bounding $K$ to $S$. $S$ defines a
partition of $K$ by considering maximal compact subspaces of $K$ whose
boundaries are included in surfaces of $S$. Such elements are called the
\emph{cells} of the decomposition associated to $S$.

The cell of this partition can have complicated shapes: in particular, a cell
can have a arbitrarily high number of surfaces of $S$ as boundaries. We are
going to refine this partition into a partition $\Collins{S}$ whose cells are
all bounded by cones of curves and at most two surfaces in $S$.

\subsection{Collins' decomposition}

We define the \emph{silhouette} \cite[Section 5.3]{Mulmuley99} of a surface defined by the equation
$p(x,y,z) = 0$ by:
\begin{align*}
  \left\{
  \begin{array}{l}
    p(x,y,z) = 0 \\
    x\frac{\partial p}{\partial x} + y \frac{\partial
    p}{\partial y} + z \frac{\partial p}{\partial z} = 0.
  \end{array}
  \right.
\end{align*}
The silhouette of a surface is the curve on the surface such that all points
$(x,y,z)$ of the silhouette are such that the ray $[(x,y,z)]$ belongs to
the tangent plane of the surface on $(x,y,z)$.

Up to infinitesimal perturbation of the coefficients of the polynomials, we can
assume that the surfaces of $S$ have no integer points in $K$.

$\projectionAz{K} = \{ \projectionAz{x} \mid x \in K\}$ is a compact of the affine plane
$\AffinePlane$. Let us consider the set $\projectionAz{S}$ of curves in
$\projectionAz{K}$ containing:
\begin{itemize}[nolistsep,noitemsep]
\item the projection of the silhouettes of surfaces in $S$;
\item the projection of the intersections of surfaces in $S$ and of the
  intersection of surfaces in $S$ with the planes $z=\mu(1+\frac{n}{6\delta})$,
  $n \in \{1,\ldots,6\delta-1\}$, where $\delta$ is the total degree of $S$;
\item \emph{vertical lines} of the form $\{(x,a,1) \mid |x| \leq 2^{\beta+1} \}$
  for $a$ a constant such that such lines pass through:
  \begin{itemize}[nolistsep,noitemsep]
  \item all intersections among the curves;
  \item all singular points of the curves;
  \item all critical points of the curves with a tangent supported by
    $\vec{e}_y$.
  \end{itemize}
\end{itemize}
$\projectionAz{S}$ defines a Collins decomposition \cite{Collins:1975} of
$\projectionAz{K}$. The intersection of any affine line supported by $\vec{e}_y$ of
the plane with a region of this decomposition is connected if nonempty.

Let $c$ be a cell in $\projectionAz{S}$. It is enclosed by two curves in
$\projectionAz{K}$ and at most two vertical lines. The curves can be parametrized by
$c_{\max} : x \mapsto \max \{y \in \realN \mid (x,y,1) \in c \}$ and
$c_{\min} : x \mapsto \min \{y \in \realN \mid (x,y,1) \in c \}$, which are both
smooth functions. The \emph{volatility} of $c$ is defined as the number of
extrema of the second derivatives $c_{\min}''$ and $c_{\max}''$ on their domains
of definition.

This set of curves $\projectionAz{S}$ can be lifted to a set of surfaces
$\Collins{S}(K)$ of $K$ that contains:
\begin{itemize}[nolistsep,noitemsep]
\item the surfaces of $S$;
\item the cones $[s]$ of every curve $s$ in $\projectionAz{S}$;
\item the planes bounding $K$;
\item $6\delta-2$ \emph{dividing planes} of equation
  $z=\mu(1+\frac{n}{6\delta})$, $n \in \{1,\ldots,6\delta-1\}$.
\end{itemize}
The projection of a cell of $\Collins{S}$ is a cell of $\projectionAz{S}$. We
say that a cell of $\Collins{S}(K)$ is \emph{flat} if none of its boundaries are
included in surfaces of $S$.

Let us call $\Deg$ the number of cells in $\Collins{S}(K)$.

Let $c$ be a cell in $\Collins{S}(K)$. Its \emph{volatility} is defined as the
volatility of its projection in $\projectionAz{S}$.

\subsection{Volatility and Separation}
\label{subsec:volatility-separation}

We here follow but rephrase Mulmuley \cite[Section 5.3; Sample points]{Mulmuley99}

\begin{definition}
  Let $K$ be a compact of $\realN^3$.
  
  A finite set of surfaces $S$ on $K$ \emph{separates} a $\rho$-fan $\Fan$ on
  $K$ if the partition on $\integerN^3 \cap K$ induced by $S$ is finer than the
  one induced by $\Fan$.
\end{definition}

We now establish the following key result of Mulmuley \cite[Theorem 5.9]{Mulmuley99}.

\begin{theorem}
  \label{thm:mulmuley}
  Let $S$ be a finite set of algebraic surfaces of total degree $\delta$, and
  $\Fan$ a $\rho$-fan of bitsize $\beta$.

  If $S$ separates $\Fan$, there exists a compact $K$ and a cell of
  $\Collins{S}(K)$ with volatility greater than $\rho/\Deg$.
\end{theorem}

In order to prove this theorem, we will build explicitely the compact $K$ and
this cell by considering sample points on $\Fan$ and show in Lemma
\ref{lem:volatility} a bound on the volatility of this cell.

Let $K$ be a compact delimited by the cone of a rectangle
$\{ (x,y,1) \mid |x|, |y| \leqslant 2^{\beta+1}\}$ and two planes of equation
$z=\mu$ and $z=2\mu$, with $\mu>(6\delta+1)2^{\beta}$. We first remark that all
affine segments of $\Fan$ are in the rectangle base of $K$.

For each affine segment of $\Fan$ with endpoints $(x_i, y_1, 1)$ and
$(x_{i+1}, y_{i+1},1)$ let, for $0 < k < 10 \Deg$, $y_i^k$ be such that
$(x_i^k, y_i^k,1)$ is in the affine segment, where
$x_i^k = \frac{(10 \Deg-k)x_i + kx_{i+1}}{10 \Deg}$. We remark that, as
$|x_i-{x_{i+1}}| > 2^{-\beta}$, we have, for $k,k'$,
$|x_i^k - x_i^{k'}| > 2^{-\beta}/10\Deg$.

\begin{lemma}
  For all sample points $(x_i^k,y_i^k,1)$, there exists a flat cell in
  $\Collins{S}$ that contains an integer point of $[(x_i^k,y_i^k,1)]$.
\end{lemma}
\begin{proof}
  Let $(x_i^k,y_i^k,1)$ be a sample point. $[(x_i^k,y_i^k,1)]$ is divided in
  $N+1$ intervals by the dividing planes. On the other hand, $[(x_i^k,y_i^k,1)]$
  intersects surfaces of $S$ in at most $\delta$ points, by Bézout theorem. So,
  there exists an interval $e$ of $[(x_i^k,y_i^k,1)]$ that is bounded by the
  dividing planes and that do not intersect any surface in $S$. By construction,
  $e$ is included in a flat cell, and its projection on the $z$-axis has length
  $\mu/(6\delta+1)$, so, as $(x_i^k,y_i^k,1)$ is of bitsize $\beta$,
  $(n2^{\beta}x_i^k,n2^{\beta}y_i^k,n2^{\beta})$ is, for all $n\in \naturalN$ an
  integer point of the ray, so, as $\mu > (6\delta+1) 2^{\beta}$, $e$ contains
  an integer point.
\end{proof}

So, for each affine segment of $\Fan$, there exists a flat cell in
$\Collins{S}$ that contains integer points in the ray of at least $10$ sample
points of the affine segment. Going further, there exists a cell $c$ of
$\Collins{S}$ that contains integer points in the ray of at least 10 sample
points of $\rho/\Deg$ affine segments of $\Fan$.

\begin{lemma}
  \label{lem:volatility}
  The volatility of $c$ is at least $\rho/\Deg$.
\end{lemma}

This is achieved by applying the mean value theorem on the function
$\projectionAz{c}_{\max}'$ on pairs of sample points. In particular, this proof
uses no algebraic geometry. 

\begin{proof}
  Let $e$ be a segment of $\Fan$ such that the ray of 10 of its sample points
  contain an integer point in $c$. Let $p = (x,y,z)$ be one of its integer point
  and $\projectionAz{p}=(x_p,y_p,1)$ its projection, which is a sample point in
  $\projectionAz{c}$. Let $q=(x,y+1,z)$. As $\projectionAz{p}$ is in $\Fan$, and $S$
  separates $\Fan$, $q$ is not in $c$, and $\projectionAz{q} = (x_q,y_q,1)$ is not
  in $\projectionAz{c}$. By Thalès theorem, $0< y_q-y_p < \frac{1}{\mu}$. So, as
  $y_q > \projectionAz{c}_{\max}(x_p) > y_p$, we have in particular that
  $0< \projectionAz{c}_{\max}(x_p) - y_p < \frac{1}{\mu}$.

  So, the 10 sample points have coordinates that approximate the graph of
  $\projectionAz{c}_{\max}$ with an error bounded by $\frac{1}{\mu}$.  Consider two
  of them $p_1=(x_1,y_1,1)$ and $p_2=(x_2,y_2,1)$, such that $x_1<x_2$. Let $a$
  be the slope of $e$ (in particular $a=(y_2-y_1)/(x_2-x_1)$. By the mean value
  theorem, there exists $\alpha \in [x_1, x_2]$ such that
  $\projectionAz{c}_{\max}'(\alpha)=\frac{ \projectionAz{c}_{\max}(x_2) -
    \projectionAz{c}_{\max}(x_1) }{x_2-x_1}$. But
  $|\projectionAz{c}_{\max}(x_2) - \projectionAz{c}_{\max}(x_1)| \leq |y_2-y_1| +
  \frac{2}{\mu}$ and $|x_2-x_1| > \frac{1}{10\Deg 2^{\beta}}$. So,
  $|\projectionAz{c}_{\max}'(\alpha) - a| \leq 2 \frac{10\Deg 2^{\beta}}{\mu}$.

  So, the function $\projectionAz{c}_{\max}'$ is close to the value $a$, with error
  bounded, between all the sample points. By applying the mean value theorem
  again, we get that there exists a point in the interval such that
  $\projectionAz{c}_{\max}''$ is close to 0, with an error bounded by $2
  \frac{10\Deg 2^{\beta}}{\mu}$.

  In the same way, let $e'$ be another segment of $\Fan$ such that the ray of 10
  of its sample points contain an integer point in $c$, of slope $a'$. Let two
  of them be $p_1' = (x_1',y_1',1)$ and $p_2' = (x_2',y_2',1)$, and suppose
  $x_2'>x_1'>x_2$. By the same reasoning as above, there exists
  $\alpha' \in [x_1',x_2']$ such that
  $|\projectionAz{c}_{\max}'(\alpha') - a'| \leq 2 \frac{10\Deg 2^{\beta}}{\mu}$. By
  the mean value theorem, there exists $\beta \in [\alpha,\alpha']$ such that
  $\projectionAz{c}_{\max}''(\beta)=\frac{ \projectionAz{c}_{\max}'(\alpha') -
    \projectionAz{c}_{\max}'(\alpha) }{\alpha'-\alpha}>\frac{1}{\mu}(|a-a'|-2
  \frac{10\Deg 2^{\beta}}{\mu})$.

  So, for each of the $\rho/\Deg$ segments of $\Fan$, we can exhibit a point
  such that $\projectionAz{c}_{\max}''$ is close to zero, and for each successive
  segment, a point such that it is far. So $\projectionAz{c}_{\max}''$ has at least
  $\rho/\Deg$ extrema.
\end{proof}

\subsection{Volatility and Degree}
\label{subsec:volatility-degree}

We can now state the following essential result.
\begin{theorem}[Mulmuley]
  \label{thm:mulmuley-geometric}
  Let $S$ be a finite set of algebraic surfaces of total degree $\delta$.

  There exists a polynomial $P$ such that, for all $\rho > P(\delta)$, $S$ does
  not separate $\rho$-fans.
\end{theorem}

While this Theorem underlies Mulmuley's proof technique, it is not explicitly stated in his article.
This result follows from \cref{thm:mulmuley} and the following two lemmas which appear as claims at the end of Section 5.3 in \cite{Mulmuley99}.

\begin{lemma}
  Let $S$ be a finite set of curves of total degree $\delta$, and $K$ be a
  compact. The cells of the decomposition $\Collins{S}$ of $K$ have a volatility
  bounded by a polynomial in $\delta$.
\end{lemma}
\begin{proof}
  Let $c$ be a cell in $\Collins{S}$ and $g(x,y)=0$ be the equation of one of
  the boundaries of $\projectionAz{c}$ in the affine plane. The degree of $g$ is
  bounded by the degree of the intersection of surfaces in $S$. Any extrema $x$
  of $f''$, where $f$ is a parametrization $y=f(x)$ of this boundary, can be
  represented as a point $(x,y,y^{(1)},y^{(2)},y^{(3)})$ in the 5-dimensional
  phase space that satisfy polynomial equations of the form:
  \begin{align*}
    g(x,y)=0,\qquad
    &g_1(x,y,y^{(1)})=0, \qquad
    g_2(x,y,y^{(1)},y^{(2)})=0\\
    &g_3(x,y,y^{(1)},y^{(2)},y^{(3)})=0, \qquad
    y^{(3)}=0,
  \end{align*}
  where all the polynomials' degrees are all bounded by the degree of the
  intersection of surfaces in $S$ (as they are the derivatives of $g$). So, by
  the Milnor--Thom theorem, such points are in number polynomial in the total
  degree of the surfaces of $S$.
\end{proof}

\begin{lemma}
  The number of cells $\Deg$ of the Collins decomposition of $S$ is polynomial
  in $\delta$.
\end{lemma}
\begin{proof}
  The intersection of the surfaces in $S$ are algebraic varieties of number
  bounded by $\delta$, by the Milnor--Thom theorem. Moreover, so are the
  silhouettes of the surfaces, as they are the intersection of two algebraic
  varieties of total degree smaller than $\delta$. So, the number of cells in
  $\Collins{S}$ is bounded by the number of cells of $S$ times the number of
  dividing planes times the number of intersections, silhouettes and vertical
  lines they engender.
\end{proof}

\section{Improving Mulmuley's result}

\subsection{\prams over $\realN$ and maxflow}

We will now prove our strengthening of Mulmuley's lower bounds for \enquote{\prams
without bit operations} \cite{Mulmuley99}. For this, we will combine the results from
previous sections to establish the following result.

\setcounter{temp}{\value{theorem}}
\setcounterref{theorem}{cor:main-pram}
\addtocounter{theorem}{-1}
\begin{theorem}
  Let $N$ be a natural number and $M$ be a real-valued \pram with
   at most $2^{O((\log N)^c)}$ processors, 
  where $c$ is
  any positive integer.
  
  Then $M$ does not decide \maxflow on inputs of length $N$ in $O((\log N)^c)$ steps.
\end{theorem}
\setcounter{theorem}{\value{temp}}

%


\begin{proof}[Proof of \Cref{cor:main-pram}]
Let $N$ be an integer. Suppose that a real-valued \pram $M$ with division and roots, with at most 
$p=2^{O((\log N)^c)}$ processors, computes \maxflow on inputs of length at most 
$N$ in time $k=2^{O((\log N)^c)}$. 

We know that $\Interpret{M}$ has a finite set of edges $E$. 
Since the running time of $M$ is 
equal, up to a constant, to the computation time of the $\crew^p(\amcrealfull)$-program 
$\Interpret{M}$, we deduce that if $M$ computes \maxflow in $k$ steps, then 
$\Interpret{M}$ computes \maxflow in at most $Ck$ steps where $C$ is a 
fixed constant.

By \Cref{thm:graphingsBenOrsystems}, the problem decided by $\Interpret{M}$ 
in $Ck$ steps defines a system of equations separating the integral inputs accepted by 
$M$ from the ones rejected. I.e. if $M$ computes \maxflow in $Ck$ steps, then this system
of equations defines a set of algebraic surfaces that separate the $\rho$-fan defined by 
\maxflow. Moreover, by \Cref{entropiccotrees2entropy}, this system of equation has a total 
degree bounded by \(O(2^{k^3(\log p))})\).

By \Cref{thm:param} and \Cref{thm:mulmuley-geometric}, there exists a polynomial $P$
such that a finite set of algebraic surfaces of total degree $\delta$ cannot separate the 
$2^{\Omega(N)}$-fan defined by \maxflow as long as $2^{\Omega(N)}>P(\delta)$. 
Hence, if $M$ is computing \maxflow, we need $P(\delta)\leq 2^{\Omega(N)}$, but we
have established that here $\delta=O(p 2^{k^3})$, so this contradicts the hypotheses 
that $p=2^{O((\log N)^c)}$ and $k=O((\log N)^c)$.
\end{proof}

This extends Mulmuley's result because of the following fact.
\begin{proposition}\label{prop:NCminusZequivR}
A subset $A\subseteq \integerN^k$ is decided by a division-free integer-valued \prams with $k$ processors in time $t$ if and only if there exists a division-free real-valued \prams with $k$ processors computing in time $t$ a subset $B\subseteq\realN^k$ such that:
\[ \{(x_1,x_2,\dots,x_k)\in\integerN^k \mid (x_1,x_2,\dots,x_k)\in B\} = A. \]
\end{proposition}

\begin{proof}
The proof is rather straightforward, since addition and multiplication of integers yield integers. As a consequence, the available operations in division-free real \prams cannot be used to construct non-integer values from solely integer inputs.
\end{proof}

A consequence of the previous result is that Mulmuley's original result is obtained as a corollary of \cref{cor:main-pram}.
Indeed, suppose a \pram without bit operations computes the \maxflow problem in polylogarithmic time. Then
there would exist a real-valued \pram computing \maxflow in polylogarithmic time, a result contradicting \Cref{cor:main-pram}.

Let us now consider the possibility of lifting this result to integer-valued machines using division. 
Let $M$ be an integer-valued \pram. We would like to associate to it a real-valued \pram
\(\widetilde{M}\) such that $M$ and $\widetilde{M}$ accept the same (integer) values, with at most
a polylogarithmic running time overhead. 
This implies in particular that real-valued \prams (with division, and potentially roots) should be able to 
compute euclidean division efficiently. 
It turns out that this is not the case. Indeed, we will show that the remainder of the
euclidean division by $2$ is in fact not computable in polylogarithmic time by real-valued \prams 
even in the presence of division and arbitrary root operations. This will be obtained using the above results
based on entropic co-trees and Mumuley's geometric argument. However, before providing the proof of this result,
we provide a more concrete version of a similar result on algebraic circuits. This proof illustrates in a simpler setting
the abstract techniques we developed above.

\subsection{Real \prams and euclidean division}\label{sec:directproof}

%


We now give a direct proof that 
algebraic circuits cannot compute the parity function. More precisely, we consider the modulo function (noted $\% $) over $ \llbracket 0;2^n \rrbracket$. We prove that the function $(n \mapsto n \% 2)$ cannot be computed by a family of real circuits of \algcirc of largeness $\mathrm{poly}(n)$ and of depth $\mathrm{polylog}(n)$ over integers of $ \llbracket 0;2^n \rrbracket$. 
In our proof, we consider that comparison gates returns $1$ if true, $0$ otherwise. This will ease the proof, while being equivalent to \Cref{def:circuitsreels}. 

In this section we do a direct analysis of circuits of \algcirc (circuits over
reals with division gate, but no root gates). We will prove that the modulo
function, cannot be computed by such a circuit when we restrict inputs to
integers. 
One
objection quoted by Mulmuley in his paper \cite{Mulmuley99} is that
proving computing the parity function cannot be done with \enquote{\pram without
  bit operation}, is akin to proving that non monotone functions cannot be
computed by monotone circuits. We disagree to the extent that proving that
euclidian division is not computed by a circuit of \algcirc is not trivial
like in the monotone function case and actually encapsulates the reason why
algebraic circuits of low depth are fundamentally limited. The core of our
argument is the same as the one Mulmuley gives in his paper, however we feel
that our formalism may help in understanding what are \prams of polylog depth
and why they are weak. On top of his comparison with monotone functions Mulmuley
adds that anyhow \enquote{a lower bound in the \pram model without bit
  operations is interesting only if the problem has an efficient sequential
  algorithm that does not use bit operations [...] One trivial example of a
  problem that does not have a strongly polynomial time algorithm is bit
  extraction itself (because such an algorithm for bit extraction must work
  within O(1) arithmetic and comparison operations, which is impossible)}. We
then give an example of a very simple problem which has a strongly polynomial
time algorithm but is not computed by algebraic circuits: deciding the set
\[\left\{ (x_1, \ldots, x_n) \in \naturalN^n ; (\sum^{n}_{i=1} x_i )\% 2 = 0 \right\}.\]
A similar argument can also be used to prove that the
natural bijection from $\naturalN$ to $\naturalN^2$ cannot be computed by 
algebraic circuits (and therefore by \prams without bit operations in polylogarithmic 
time), which is an interesting limitation to notice.

The main idea of the proof is a
more concrete statement of the technical lemma above: we show that circuits of
\algcirc compute piece-wise polynomial fractions with
few zeroes and few \emph{pieces} and argue that it therefore can't compute the
euclidean division.

\begin{definition}
A set $I$ over $\realN$ is an interval if there exists two real numbers $a,b \in \realN \cup \{-\infty ; + \infty \}$ such that either $I=[a,b]$, or $I=]a,b]$, or $I=[a,b[$, or $I=]a,b[$.
\end{definition}

\begin{definition}[Comparison of intervals]
  Let $I$ and $J$ be two non empty, non intersecting intervals over $\realN$. We
  say $I<J$ if $\forall x \in I, \forall y \in J, x<y$.

  Given a collection of $r>0$ intervals $(I_i)_{i=0}^r$, we write
  $I_1 < I_2 < \ldots < I_r$ to indicate the family is ordered and all the
  intervals are empty and non intersecting.
\end{definition}

Note that what we call here "pieces" will play the role of (and are in fact) the cells 
in the $k$-cell decomposition exposed above. In particular, the main 
technical result leading to the lower bound is an upper bound on the number of
pieces and on the \emph{extended degree} of the piece-wise rational functions
defined by a circuit of a given depth. This should be understood as a special case 
of our general technical lemma (\Cref{mainlemma}) giving bounds on the number 
of cells and the number and degrees of the polynomials defining each cell (here
the notion of extended degree captures both the number of equations and their
degree, which is possible because we work in a simpler setting).

\begin{definition}[Interval cut and pieces]
Given a family of $r$ intervals $I_1 < I_2 < \ldots < I_r$ ($r>0$), we say it is a \emph{size $r$ interval cut} if $I_1, I_2, \ldots ,I_r$ partitions $\realN$. Each $I_i$ is called a \emph{piece}.
\end{definition}

In the following we may use as a shorthand the notation $I$ for a family of $r$ intervals $(I_i)_{1 \leq i \leq r}$.

\begin{definition}[Intertwining interval cuts]
Suppose given a size $n$ interval cut $I_1 < I_2 < \ldots < I_n$ and a size $m$ interval cut $J_1 < J_2 < \ldots < J_m$. There is at least one size $k$ interval cut $K_1, \ldots, K_k$ such that $\forall h, \exists i,j , K_h \subset I_i \wedge K_h \subset J_j$. If $k$ is moreover minimal then we call $K_1, \ldots, K_k$ an \emph{intertwining} of $I$ and $J$.
\end{definition}

As an example of intertwining consider the two families
$I=(\mbox{\(]-\infty ; 0]\)}, \mbox{\(]0; +\infty[\)})$ and
$J= (\mbox{\(]- \infty ; 1]\)}, \mbox{\(]1; +\infty[\)})$; one possible
intertwining is $(]-\infty ; 0],]0;1], ]1; +\infty[)$. Note that the size $k$ of
an intertwining between size $n$ and size $m$ interval cuts satisfies
$k\leqslant n+m$ (consider the extremal points of each interval).

\begin{definition}[Piece-wise polynomial fraction]
Consider given a size $n$ interval cut $I_1 < I_2 < \ldots < I_n$. A function $f$ over $D \subset \realN$ is said to be a piece-wise polynomial fraction over $I$ if there exists $n$ polynomial fractions $f_1, \ldots, f_{n}$ such that for all integer $1\leqslant i\leqslant n$: $\forall x \in I_i \cap D , f(x)=f_{i}(x)$.

The number of pieces of a piece-wise polynomial fraction $f$ is the smallest number $n$ such that $f$ is a polynomial fraction over a size $n$ interval cut. We write $c$ the function which associates to each piece-wise polynomial fraction $f$ and return its minimal number of pieces. In the following, we will abusively apply $c$ to circuits that compute piece-wise polynomial fractions.
\end{definition}

Let us note that the poles of a polynomial fraction do not requires a splitting in pieces. For instance, the function $f(x)=\frac{1}{x^2-1}$ defined for all $x\in \realN \setminus \{ -1; 1\}$ is a one piece polynomial fraction. If one wanted to have functions defined over the whole of $\realN$ it would be no issue as whenever performing a division we could add a comparison gate in the circuit to verify that we are not dividing by 0, and if so return an arbitrary value.\\

Let $F$ be a piece-wise polynomial fraction over $D \subset \realN$, there can be multiple interval cut associated to the number of pieces of $F$, for instance the function which is 0 for negative numbers and $x \mapsto x$ for non negative numbers as two minimal cuts : $(] - \infty ; 0 ] , ]0;+ \infty[)$ or $(] - \infty ; 0 [ , [0;+ \infty[)$. So in the following we will speak of \textit{an} interval cut of $F$. Once we have fixed an interval cut $I$, the associated polynomial functions are unique except on intervals $I_i$ such that $I_i \cap D = \{ a \}$ is a singleton, for those intervals we always associate the polynomial of degree 0, $f_i = (x \mapsto f(a))$. This is important for the following definition.

\begin{definition}[Augmented degree of a polynomial fraction]
Let $P$ and $Q$ be polynomials over $\realN$. We define the \emph{augmented degree} $d(F)$ of the polynomial fraction $F=\frac{P}{Q}$ as $d(F)= d(P)+d(Q)+1$, where $d(P), d(Q)$ are the degrees of the polynomials $P,Q$ in the standard meaning.

Let $f$ be a piece-wise polynomial fraction over $\realN$ with a minimal interval cut $ I_1 < \ldots < I_n$ and $ f_1, \ldots, f_{n}$  the associated polynomial fractions. We define the extended degree of $f$ by $d(f)= max_{1 \leq i \leq n} d(f_i)$. In the following we will abusively apply $d$ to circuits that compute piece-wise polynomial fractions.

\end{definition}

\begin{theorem}
\label{main}
Let $P$ be a real circuit of depth $k$ taking only one input. Then $P$ computes a piece-wise polynomial fraction such that $d(P) \leq 2^{k}$ and $c(P) \leq 2^{\frac{k^2+3k}{2}}$.
\end{theorem}

\begin{proof}
We prove the result for each gate by induction on the depth of the gate. I.e. we prove that the circuit computes a piece-wise polynomial fraction, and that the stated bounds are correct. We denote by $c_k$ (resp. $d_k$) the maximal number of pieces (resp.  the maximal augmented degree) of a function computed by a gate at depth $k$. Let $p$ be a gate, $c(p)$ corresponds to the augmented degree of the function computed by $p$, $d$ to its augmented degree. 

For depth $k=1$ the function computed is either $f= (x \mapsto x)$ or $g=(x \mapsto 1) $. Notice that $d(f)= 2 $, $c(f)=1$, $d(g) =1$ and $c(g)=1$. Therefore $d_1 \leq 2^1$ and $c_1 \leq 2^2$.

For the induction step, let us consider a gate called $p$ at depth $k+1$. Our induction hypothesis is that $d_k \leq 2^{k}$ and $c_k \leq 2^{\frac{k(k+3)}{2}}$. We now consider the different possible gates $p$.
\begin{itemize}[nolistsep,noitemsep]
\item If $p$ is a multiplication gate: $p=p_1*p_2$. Let $I_1 < \ldots < I_r$ and $J_1, < \ldots < J_m$ be the minimal size interval cuts of $p_1$ and $p_2$, and $r=c(p_1)$ and $m=c(p_2)$. Let us consider $K_1, \ldots, K_k$ an intertwining of $I$ and $J$ such that $k \leqslant r+m$ . We note that $p$ is a piece-wise polynomial fraction over $K$. Therefore $c(p)\leq c(p_1)+c(p_2) \leq 2*c_k   \Rightarrow c_(p) \leq 2*c_k \leq 2^{\frac{(k+1)(k+4)}{2}} $ . By reasoning over each piece of the interval cut $K$ we also obtain that $d(p) \leq d(p_1)+d(p_2) \leq 2d_k$.
\item If $p$ is either an addition or a division gate, the reasoning is similar.
\item If $p$ is a comparison gate: $p= (p_1 \leq p_2)$. Let $A_1 \leq A_2 \leq \ldots \leq A_n$ be the pieces of $p_1$, $B_1 \leq B_2 \leq \ldots \leq B_m$ be the pieces of $p_2$, and define $n = c(p_1)$ and $m=c(p_2)$. Let $I_1 \leq I_2 \leq \ldots \leq I_{r}$ be an intertwining of $A$ and $B$. For all integer $1\leqslant i\leqslant r$, $p_1$ and $p_2$ are both polynomial fraction over $I_i$ and we can canonically write $p_{i,1} = \frac{q_{i,1}}{t_{i,1}}$ and $p_{i,2} = \frac{q_{i,2}}{t_{i,2}}$. To count the number of pieces of $p$, we focus on each $I_i$ separately. One can notice that over $I_i$ any new piece of $p$ contained in $I_i$ may only occur on places where $p_{i,1} -p_{i,2}$ changes sign, moreover $p_{i,1} \leq p_{i,2}$ if and only if either $q_{i,1}*t_{i,2} - t_{i,1}*q_{i,1} \leq 0$ or $q_{i,1}*t_{i,2} - t_{i,1}*q_{i,1} \geq 0$ (depending on the signs of the denominators). But the polynomial $q_{i,1}*t_{i,2} - t_{i,1}*q_{i,1}$ has at most $d(p_{i,1})+d(p_{i,2})-1$ roots, therefore over $I_i$ we "add" at most $d(p_{i,1})+d(p_{i,2})$ pieces. Since we have at most $r \leq n+m = c(p_1) +c(p_2)$ pieces $I_i$, we have that $c(p)$ is at most $(d(p_1)+d(p_2))*(c(p_1)+c(p_2)) \leq 2d_k * 2c_k \leq 2*2^k * 2*2 ^{\frac{k^2 + 3k}{2}} \leq 2 ^{\frac{(k+1)^2 + 3(k+1)}{2}} $. And $d(p)$ is equal to $1 \leq 2^{k+1}$. 
\end{itemize}
\end{proof}

The next lemma state that no piece-wise polynomial fraction with few pieces or low augmented degree can agree with the remainder function on many consecutive integers.

\begin{lemma}
\label{bigdegree}
Let $f$ be a piece-wise polynomial fraction over $\realN$, $N$ an integer. If $\forall k \in \llbracket 0;N \rrbracket , f(k)=k\%2$, then $d(f)*c(f) \geq \frac{N}{10}$.
\end{lemma}

\begin{proof}
Let $f$ be a piece-wise polynomial fraction as described in the theorem, and $c$ its number of pieces with $I_1, \ldots , I_c$ a corresponding interval cut. If $c$ is less than $N/10$, there must be a piece $I_j$ containing at least $N/c$ integers of $\llbracket 0;N \rrbracket$. We remind that over $I_j$, the function $f$ is a polynomial fraction which we note $\frac{P}{Q}$. Since $k \mapsto k\%2$ has $N/2c$ zeroes over $I_j$, so must $P$. The polynomial $P$ also cannot be the zero function because $k \mapsto k\%2$ is equal to $1$ for some integers in $I_j$. Therefore P is of degree at least $N/2c$, and $d(f)*c \geq (N/2c) * c \geq N/10$.
\end{proof}

\begin{theorem}
\label{remaindernotinNCR}
For any function family $(f_n)$ computed by a real-valued circuit family of depth $c \log^d(n)$, there exists $N \in \naturalN$ such that $\forall n > N, \exists k \in \llbracket 0;2^n \rrbracket, f_n(k) \neq k\%2$.
\end{theorem}

\begin{proof}
Let $(C_n)_n$ be a sequence of circuits of depth $k_n= O(c \log^d(n))$. Since $n\mapsto n \%2 $ is a function with one input, we may consider $(C_n)_n$ to be a sequence of circuits with one input. For any $n$, the function computed by $C_n$ is a piece-wise polynomial fraction $f$. By theorem \ref{main} we have that $c(f) \leq 2^{\frac{k_n^2 + 3 k_n}{2}}$ and $d(f) \leq 2^{k_n}$ therefore $c(f)d(f) = o(2^n)$. Therefore for large enough $n$, using lemma \ref{bigdegree}, $f$ cannot coincide with $n \mapsto n \% 2 $ for all integers in $\llbracket 0;2^n \rrbracket$.
\end{proof}

We continue by proving that the two problems given in introduction are indeed not computable by algebraic circuits.

\begin{theorem}
The set \[\left\{ (x_1, \ldots, x_n) \in [0,2^n]^n \mid (\sum^{n}_{i=1} x_i )\% 2 = 0 \right\}.\] 
has polynomially bounded arithmetic circuits but is not computable by polylogarithmic depth algebraic circuits.
\end{theorem}

\begin{proof}
One can retrieve with a circuit of size linear in $n$ the last bit of $\sum^{n}_{i=1} x_i$. On the other hand the set \[\left\{ (0, \ldots,0, x_n) \in [0,2^n]^n \mid (\sum^{n}_{i=1} x_i )\% 2 = 0 \right\}.\] is not computable by polylogarithmic depth algebraic circuits by \Cref{remaindernotinNCR}.
\end{proof}

\begin{theorem}
Let $f:   N \mapsto N \times N$ be the usual Cantor pairing function $f(n,m)=\frac{1}{2}(m + n)(m + n + 1) + m$. This function is bijective and its reciprocal $f^{-1}$ cannot be computed by a polylogarithmic depth algebraic circuit.
\end{theorem}

\begin{proof}
Consider the function $g(x,y)=(x<y)$. (it returns either 0 or 1) The function $g(f^{-1})$ has $2^{\frac{n}{2}}$ alternations between $0$ and $1$ over $\llbracket 0;2^n \rrbracket$ therefore it cannot be computed by a polylogarithmic depth algebraic circuit (by arguments used to prove theorem \ref{bigdegree} and \ref{remaindernotinNCR} . But $g$ can be computed by a polylogarithmic depth algebraic circuit. As a consequence, it must be $f^{-1}$ which cannot be computed by a polylogarithmic depth algebraic circuit.
\end{proof}

We haven't been able to use this analysis to include circuits with square root gates. Everything may still hold, but we haven't been able to prove satisfactory bounds on the number of zeros of function computed by circuits with square root gates. This is important because when considering the comparison gates, the number of pieces created directly depends on the number of zeroes of the function.

So, while the arguments of this section show a more concrete application of the ideas behind the method used in this paper, this version is more limited. In the next section, we will extend this result to the model with arbitrary roots using the more abstract method developed in previous sections.

\subsection{Extending to \prams with roots}

The above theorem is a particular case of the following result, which we prove using the general technique 
developed in the previous sections. We note that we do not know of a more concrete proof using arguments similar
to the ones used in the previous section. 

Using 
the general bounds provided by entropic co-trees (\Cref{thm:graphingsBenOrsystems}) and the 
geometric result extracted from Mulmuley's geometric proof of lower bounds 
(\Cref{thm:mulmuley-geometric}), we show that euclidean division by $2$ cannot be computed
by real-valued \prams (with division and arbitrary root operations) on inputs of length $N$ in polylogarithmic time in $N$.

\setcounter{temp}{\value{theorem}}
\setcounterref{theorem}{mainthm}
\addtocounter{theorem}{-1}
\begin{theorem}
  Let $N$ be a natural number and $M$ be a real-valued \pram with
  at most $2^{O((\log N)^c)}$ processors, where $c$ is
  any positive integer.
  
  Then $M$ does not compute euclidean division by $2$ on inputs of length $N$ in $O((\log N)^c)$ steps.
\end{theorem}
\setcounter{theorem}{\value{temp}}

\begin{proof}
Suppose that a real-valued \pram $M$ with division and roots, with at most $p=2^{O((\log N)^c)}$ processors, 
computes euclidean division by $2$ on inputs of length at most $N$ in time $k=2^{O((\log N)^c)}$. 
We know that $\Interpret{M}$ has a finite set of 
edges $E$, and the running time of $M$ is equal, up to a constant, to the computation time of 
the $\crew^p(\amcrealfull)$-program $\Interpret{M}$, we deduce that if $M$ computes 
euclidean division by $2$ in $k$ steps, then $\Interpret{M}$ computes euclidean division by $2$ in at most $Ck$ steps 
where $C$ is a fixed constant.

Consider the following optimisation problem:
\[ \{(x,y,1) \mid y\leqslant x//2 \} \]
It is defined by the frontier set 
\begin{align*}
  \Frontier = \{(x,y,1) \in \realN^3 \mid y = x//2 \}.
\end{align*}
and we remark that the induced cone
\begin{align*}
  [\Frontier] &= \{(x,y,z) \in \realN^2 \times \realN^+ \mid y/z =
                \MaxOptProb \circ\Parametrization (x/z)\}.
\end{align*}
is a $\rho$-fan where $\rho=2^{\Omega(N)}$ is exponential in the maximal size of the inputs.

By \Cref{thm:graphingsBenOrsystems}, the problem decided by $\Interpret{M}$ 
in $Ck$ steps defines a system of equations separating the integral inputs accepted by 
$M$ from the ones rejected. I.e. if $M$ computes euclidean division by $2$ in $Ck$ steps, then this 
system of equations defines a set of algebraic surfaces that separate the $\rho$-fan defined above.
Moreover, this system of equation has a total degree bounded 
by \(O(2^{k^3(\log p)})\).

Now by \Cref{thm:mulmuley-geometric}, there exists a polynomial $P$
such that a finite set of algebraic surfaces of total degree $\delta$ cannot separate the 
$2^{\Omega(N)}$-fan defined by euclidean division by $2$ as long as $2^{\Omega(N)}>P(\delta)$. 
But $\delta=O(p 2^{k^{3}})$, contradicting the hypotheses 
that $p=2^{O((\log N)^c)}$ and $k=O((\log N)^c)$.
\end{proof}

\clearpage
\bibliographystyle{elsarticle-num} 
\bibliography{biblio}

\end{document}